\documentclass[review]{elsarticle}
\usepackage[top=2.5cm,bottom=2.5cm,left=2.5cm,right=2.5cm]{geometry}

\usepackage{amsthm}
\usepackage{pbox}
\usepackage{mathrsfs}

\makeatletter
\def\ps@pprintTitle{%
	\let\@oddhead\@empty
	\let\@evenhead\@empty
	\def\@oddfoot{}%
	\let\@evenfoot\@oddfoot}
\makeatother

\usepackage[lined,ruled]{algorithm2e}
\usepackage{algorithmic}
\usepackage{booktabs}

\usepackage{tikz}
\usepgflibrary{shapes.symbols}
\usetikzlibrary{shapes,snakes}
\usepackage{pgfplots}
\usetikzlibrary{decorations.pathreplacing}

\newtheorem{definition}{Definition}
\newtheorem{theorem}{Theorem}
\newtheorem{lemma}{Lemma}

\newtheorem{remark}{Remark}

\theoremstyle{definition}

\usepackage{overpic} 
\usepackage{amsfonts,amssymb}
\usepackage{mdwmath}
\usepackage{subfigure}
\usepackage{footnote}
\usepackage{graphicx,amsmath}

\usepackage{etoolbox}
\newcommand{\circled}[2][]{%
	\tikz[baseline=(char.base)]{%
		\node[shape = circle, draw, inner sep = 1pt]
		(char) {\phantom{\ifblank{#1}{#2}{#1}}};%
		\node at (char.center) {\makebox[0pt][c]{#2}};}}
\robustify{\circled}

\usepackage{ragged2e}
\usepackage{etoolbox}

\usepackage{listings}
\usepackage{enumitem}

\journal{Performance Evaluation}

\begin{document}

\begin{frontmatter}

\title{Asymptotic Performance Evaluation of Battery Swapping and Charging Station for Electric Vehicles}


\author[ust]{Xiaoqi Tan}
\ead{ecexiaoqi.tan@connect.ust.hk}

\author[ust]{Bo Sun}
\ead{bsunaa@connect.ust.hk}

\author[zjut]{Yuan Wu}
\ead{iewuy@zjut.edu.cn}

\author[ust]{Danny H.K. Tsang}

\ead{eetsang@ust.hk}

\address[ust]{Department of Electronic and Computer Engineering,\\ Hong Kong University of Science and Technology, Clear Water Bay, Kowloon, Hong Kong}
\address[zjut]{College of Information Enginnering, Zhejiang University of Technology, Hangzhou, China}

\begin{abstract}
A battery swapping and charging station (BSCS) is an energy refueling station, where i) electric vehicles (EVs) with depleted batteries (DBs) can swap their DBs for fully-charged ones, and ii) the swapped DBs are then charged until they are fully-charged. Successful deployment of a BSCS system necessitates a careful planning of swapping- and charging-related infrastructures, and thus a comprehensive performance evaluation of the BSCS is becoming crucial. This paper studies such a performance evaluation problem with a novel mixed queueing network (MQN) model and validates this model with extensive numerical simulation. We adopt the EVs' blocking probability as our quality-of-service measure and focus on studying the impact of the key parameters of the BSCS (e.g., the numbers of parking spaces, swapping islands,  chargers, and batteries) on the blocking probability. We prove a necessary and sufficient condition for showing the ergodicity of the MQN when the number of batteries approaches infinity, and further prove that the blocking probability has two different types of asymptotic behaviors. Meanwhile, for each type of asymptotic behavior, we analytically derive the asymptotic lower bound of the blocking probability.
\end{abstract}

\begin{keyword}
	Battery Swapping and Charging Station, Electric Vehicles, Mixed Queueing Network, Asymptotic Analysis, Capacity Planning.
\end{keyword}

\end{frontmatter}

\section{Introduction}
\vspace{-0.4cm}

The transportation sector accounts for a substantial portion (over a 20\% share in the United States \cite{transportation_sector_greenhouse_gas}) of greenhouse gas emissions and over 70\% of the global oil consumption. Therefore,  it is conceived that transportation electrification, especially the deployment of electric vehicles (EVs), will be the most promising medium-term solution to reduce carbon emissions and oil supply risks \cite{evaluation_envoronmental_impact}. The speed of uptake of EVs, however, is highly sensitive to the well-known \textit{range anxiety} issue (i.e., the worry that an EV will fail to reach its destination due to insufficient energy). Although some recent advancements in energy refueling solutions can mitigate this problem to some extent\footnote{ For instance, the Supercharging technology from Tesla Motors Inc. makes it less difficult to charge an EV in a public charging station \cite{tesla_supercharging_stations}. Meanwhile, the widespread deployment of Level-1 and Level-2 plug-in slow charging spots also enable EV owners to charge their EVs for hours when  parked at the work place or to perform overnight charging at home \cite{charging_station}.}, it is still far from practical to refuel an EV within a reasonably short time in the middle of a trip. Moreover, the range anxiety issue is further complicated by the current scarcity of public charging stations, which further discourages EVs from being massively adopted.

A promising solution that can potentially overcome the range anxiety issue is the deployment of battery swapping and charging stations (BSCSs) \cite{battery_charger_for_bss, bss_electric_bus,tesla_battery_swapping, bss_qingdao}. Specifically, a BSCS is an energy refueling station where i) the depleted batteries (DBs) of EVs can be swapped for fully-charged ones (i.e., the swapping service), and ii) the swapped DBs can then be charged  until they are fully-charged (i.e., the charging service). The key advantage of the BSCS is that the EV owners wait for only a short period of time for swapping their batteries, and the swapped DBs can be charged in standalone mode at any time. For instance, for the electric-bus BSCS project in Qingdao, China \cite{bss_qingdao}, the swapping service takes only several minutes. It is even more impressive that Tesla Motors can finish swapping a battery for its Model S in around 90 seconds \cite{tesla_battery_swapping}, which is even faster than refueling a gasoline tank for conventional internal combustion engine vehicles.



In addition to the fast swapping service, the battery-swapping mode can bring more advantages for both the  EV  customers and the BSCS operators in the following three aspects. \textit{First}, as a practical business model that is currently adopted by some companies in China (e.g.,  \cite{bss_qingdao}), the batteries can be owned by the BSCS operator and leased to customers, and the payment can be charged based on monthly driving mileage. Therefore, the battery-swapping mode can decouple the ownership of batteries and vehicles and thus can significantly reduce the upfront cost of purchasing an EV. As a result, the adoption rate of EVs might be largely increased. \textit{Second}, owing to the operator's proficiency, the swapped DBs can be charged in a more appropriate manner (e.g., to prolong batteries' lifetimes) than that of being charged individually by the EV owners. \textit{Third}, the swapped DBs can be  aggregated in a large quantity and form a gigantic battery energy storage system. Therefore, the BSCS can provide enormous flexibility for grid operators to perform critical tasks such as balancing the grid \cite{scheduling_hybrid_service_scheduling_bss} and buffering intermittent renewable energy \cite{evaluation_economic_value_PV} \cite{planning_optimization_issue}, which will considerably improve the stability of power networks.

Undoubtedly, successful deployment of an advanced energy refueling network of BSCSs necessitates a careful planning of swapping- and charging-related infrastructures \cite{planning_infrastructure_planning}. As a fundamental step to design such an energy refueling network, a comprehensive performance evaluation of each BSCS is important. To this end, this paper focuses on the theoretical modeling and asymptotic performance evaluation, with their applications to the capacity planning of a BSCS. Before presenting the contribution of this paper, we first review our prior studies regarding the modeling and performance evaluation of the BSCS.



\begin{figure*}
	\centering
	\begin{tikzpicture}[scale=0.55, transform shape, >=latex]
	\draw[->,line width=1pt] (-8,0.5) to (-7,0.5); 
	\draw (-9,0.5) node[]{EV Arrivals};
	
	\fill[bottom color=red] (-6.2,0) -- (-6.2,1) -- (-3,1) -- (-3,0);
	\draw (-7,0) rectangle (-3,1);
	\draw (-7,0) to (-3,0); 
	\draw (-7,1) to (-3,1);
	
	\foreach \x in {-3.2,-3.4,...,-6}
	{
		\draw (\x,0) to (\x,1);
	}
	
	\draw (-5,-0.5) node[]{ Buffer Size $ V $};
	\draw (-5,2.5)  node[right]{ \textbf{Open EV-Queue}};
	
	\fill[bottom color=red] (-9.4,5) -- (-9.4,5.2) -- (-8.4,5.2) -- (-8.4,5);
	\draw (-9.4,5) rectangle (-8.4,5.2);
	\draw (-8,5.1) node[right]{EVs};
	
	\fill[bottom color=gray] (-9.4,4) -- (-9.4,4.2) -- (-8.4,4.2) -- (-8.4,4);
	\draw (-9.4,4) rectangle (-8.4,4.2);
	\draw (-8,4.1) node[right]{Batteries};
	
	\draw (-8.9,3)  circle (0.5cm);
	\draw (-8,3)   node[right]{Servers};
	
	\draw (-9.8,2.4) rectangle (-6.1,5.4);
	\draw[->] (-3,0.5) to (-1.5,1.5);
	\draw[->] (-3,0.5) to (-1.5,-0.5);
	\draw (-1,1.5)  circle (0.5cm); \draw (-1,1.5) node[]{$ \mathrm{SS}_{1}$};
	\draw (-1,0.5)  node[] {\vdots};
	\draw (-1,-0.5) circle (0.5cm);\draw (-1,-0.5) node[]{$ \mathrm{SS}_{S}$};
	\draw (-1,-1.5) node[below]{Swapping Service};
	\draw[dashed] (-1.7,-1.2) rectangle (-0.3,5.2);
	\draw[dashed] (-7.4,-1.5) rectangle ( 0,2.2);
	
	\draw[] (-0.5,1.5) to (0.5,0.5);
	\draw[] (-0.5,-0.5) to (0.5,0.5);
	\draw[->] (0.5,0.5) to (3,0.5); 
	\draw[] (1.75,0.5) node[above]{DB Arrivals}; 
	\draw (4.5,2.5) node[above] {\textbf{DB-Queue}};
	\fill[bottom color=gray] (3.3,0) -- (3.3,1) -- (4.5,1) -- (4.5,0);
	\draw (3.5,0) rectangle (4.5,1);
	\draw (2.9,0) to (3.5,0); 
	\draw (2.9,1) to (3.5,1);
	\draw (3.3,0) to (3.3,1);
	\draw (3.5,0) to (3.5,1);
	\draw (3.7,0) to (3.7,1);
	\draw (3.9,0) to (3.9,1);
	\draw (4.1,0) to (4.1,1);
	\draw (4.3,0) to (4.3,1);
	\draw (6.5,1.5) circle(0.5cm); \draw (6.5,1.5) node[]{$ \mathrm{CS}_{1} $}; 
	\draw (6.5,0.5) node[]{\vdots};
	\draw (6.5,-0.5) circle(0.5cm); \draw (6.5,-0.5) node[]{$ \mathrm{CS}_{C} $};
	\draw (6.5,-1.5) node[below]{Charging Service};
	\draw[dashed] (2.2,-1.5) rectangle (7.2,2.2);
	
	\draw[->] (4.5,0.5) to (6, 1.4);
	\draw[->] (4.5,0.5) to (6,-0.4);
	
	\draw (6.9,1.2)  to  (7.65,0.5);
	\draw (6.9,-0.2) to  (7.65,0.5);
	\draw [->] (7.65,0.5) to (10,0.5) to (10,5.4) to (-1,5.4) to (-1,4.5);
	\fill[bottom color=gray] (-1.5,3) -- (-1.5,3.6) -- (-0.5,3.6) -- (-0.5,3);
	\draw(-1.5,3) rectangle (-0.5,3.6);
	\draw (-1.5,3.2) to (-0.5,3.2);
	\draw (-1.5,3.4) to (-0.5,3.4);
	\draw (-1.5,3.6) to (-1.5,3.6);
	\draw (-0.5,5)   to (-0.5,3.6);\draw (-1.5,5) to (-1.5,3.6);
	\draw[->] (-1,3) to (-1,2.2);
	\draw(-2,5.5) node[above]{\textbf{FB-Queue}};
	\draw( 4,5.5) node[above]{FB Arrivals};
	\end{tikzpicture}
	\caption{The MQN  that models the operations of the BSCS. We denote the different SSs and CSs by $ \mathrm{SS}_{i}, i\in\{1,\cdots,S\} $ and $\mathrm{CS}_{k}, k\in\{1,\cdots,C\}  $, respectively.}
	\label{the_model}
\end{figure*}
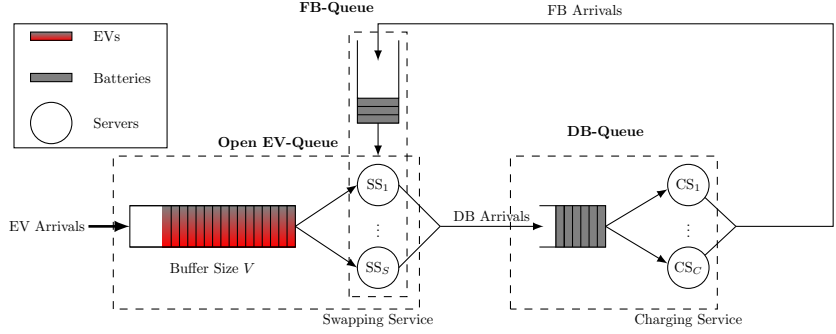

\subsection{Our Prior Work}
\label{our_contribution} 
In our previous papers \cite{evaluation_own_queueing_model_battery_swapping,scheduling_own_control_battery_swapping}, we have shown that a BSCS can be modeled as a mixed queueing network (MQN) (see Fig. \ref{the_model}). The proposed queueing network is mixed in the sense that it consists of two coupled queueing systems, i.e., an open EV-queue and a closed battery-queue. Meanwhile, the closed battery-queue further consists of two sub-queues, which are respectively denoted as the DB-queue and the FB-queue. The open EV-queue absorbs EVs from outside, and provides swapping services for EVs by first unloading DBs from vehicles to the station, and then loading FBs from the station to the vehicles. Each EV departs with a FB, and the DB unloaded from the EV will be left in the closed battery-queue, waiting for the charging service. Note that a new DB joins the closed battery-queue if and only if there is a departure of a FB, and thus a BSCS always has a fixed number of batteries in the closed battery-queue. The open EV-queue strongly couples with the closed battery-queue  since the FB-queue shares the same swapping servers (SSs) with the open EV-queue. Such a coupling effect between these two queues makes our MQN different from the standard Jackson networks \cite{jackson_networks} \cite{closed_queueing_network_exponential_servers}, which renders the analysis of the MQN nontrivial.

In \cite{evaluation_own_queueing_model_battery_swapping}, we adopted the embedded Markov chain approach to analyze the steady-state distribution of the proposed MQN. Based on some mild approximation, we obtained the steady-state distribution of the MQN, with which we further quantified various performance metrics. In \cite{scheduling_own_control_battery_swapping}, we formulated the charging control problem of a BSCS as an Markov decision process, which aims at finding an optimal policy  to minimize the total charging cost while guaranteeing a certain level of quality-of-service (QoS). To the best of our knowledge, our proposed MQN for modeling the BSCS has never been studied  by any related literature before.


\subsection{Contribution and Organization of This Paper}
Different from the steady-state analysis in \cite{evaluation_own_queueing_model_battery_swapping} and the charging control in \cite{scheduling_own_control_battery_swapping}, this paper focuses on the asymptotic performance analysis with its application to the capacity planning of the BSCS. In particular, as we can see from Fig. \ref{the_model}, there are four important parameters that determine the size/capacity of a BSCS, namely, i) the number of parking spaces, which corresponds to the buffer size of the open EV-queue, ii) the number of swapping islands\footnote{A swapping island can either represent a robot if the swapping service is performed autonomously, or it can represent a worker if the swapping service is performed manually.}, which corresponds to the number of SSs, iii) the number of chargers, which corresponds to the number of charging servers (CSs), and iv) the number of batteries in the closed battery-queue. Intuitively, each feasible tuple of the four parameters defines a planning decision. \textit{The asymptotic queueing analysis in this paper aims to analyze how the blocking probability of the BSCS behaves when the number of batteries becomes sufficiently large. More importantly, we will particularly show how the asymptotic analysis facilitates a good capacity planning for the BSCS}.  The key contributions of the paper are summarized as follows:
\begin{enumerate}[leftmargin=*]
	\item \textbf{Asymptotic Analysis}. We derive the balance equations for the MQN and calculate the corresponding steady-state distribution based on a two-dimensional continuous-time Markov chain (CTMC) approach. The key theoretical result established in this paper is the proof of an analytical necessary and sufficient condition, by which we can show the ergodicity of the MQN when the number of batteries approaches infinity.  Moreover, depending on whether this condition is satisfied or not,  we show that the MQN is asymptotically equivalent to two sub-queueing networks with much simpler structures. To the best of our knowledge, the asymptotic analysis of the MQN model has never been studied in existing literature. 
	
	\item \textbf{Capacity Planning}. The theoretical development of this paper contributes to the capacity planning of BSCSs in practice. Specifically, we classify the four parameters into three stages and propose the concept of multi-stage capacity planning. In Stage-I, we analytically investigate how the number of parking spaces and the number of SSs influence the $(N,S)$-limiting lower bound of the blocking probability no matter how many CSs and batteries are deployed in the MQN.  In Stage-II, we define two operating modes for the BSCS (i.e., the charging-limiting mode and the swapping-limiting mode) based on the number of CSs. In Stage-III, we prove that the $(N,S)$-limiting lower bound of the blocking probability is achievable only in the swapping-limiting mode but not in the charging-limiting mode. For the latter working-mode, we further analytically derive the achievable lower bound of the blocking probability. In summary, our proposed multi-stage capacity planning framework reveals the nature of how the four parameters influence the overall QoS of the MQN, which is of great importance in the planning of real BSCSs. 
\end{enumerate}

The rest of the paper is organized as follows. We extensively survey the related literature in Sec. \ref{sec_literature_review}. We present the details of the MQN model and calculate the associated steady-state distribution in Sec. \ref{queueing_network_models}. We then introduce the concepts of multi-stage capacity planning, the $(N,S)$-limiting lower bound of the blocking probability, the two operating modes, and the asymptotic properties in Sec. \ref{QoS_Multi_Stage_Capacity_Planning}. As the main theoretical contribution of this paper, the necessary and sufficient condition to show the asymptotic ergodicity of the MQN is proved in Sec. \ref{proof_of_theorem_1}. 
Numerical simulation and discussion are presented in Sec. \ref{sec_simulation_results}. We finally conclude our paper in Sec. \ref{conclusion}.

\section{Literature Review}
\label{sec_literature_review}
Motivated by the aforementioned advantages of the battery-swapping concept, there has been a growing amount of  research on BSCSs in recent years. In this section, we classify the related literature into the following three streams\footnote{The classification is just for the purpose of a clear presentation, and some of the surveyed literature may cover multiple streams. }. 


The first stream focuses on the \textit{modeling} and \textit{evaluation} of the BSCSs. The evaluation of BSCSs includes the  performance analysis of environmental impact
\cite{evaluation_envoronmental_impact}, economic benefits \cite{evaluation_economic_value_PV, evaluation_economic_microgrid},  and reliability impact on the power system \cite{evaluation_system_reliability_battery_exchange_mode}, etc. For instance, the authors of  \cite{evaluation_envoronmental_impact} perform a comparison study on the battery-swapping mode and the conventional fixed-battery mode (i.e., the conventional plug-in charging mode) in terms of their capabilities of reducing oil dependence and carbon emission. \cite{evaluation_economic_value_PV} proposes to  absorb the surplus electricity from photovoltaics (PVs) by using the unloaded batteries of a BSCS. It is demonstrated that the marginal economic value of the PV inverter and that of batteries heavily depends on each other's capacity. Therefore, the capacities of both the PV inverter and the batteries should be properly selected. 
The authors of \cite{evaluation_economic_microgrid} propose  an energy dispatching strategy for a microgrid system containing a BSCS, wind generator, PV system, fuel cell, and etc. The simulation results in \cite{evaluation_economic_microgrid} show that a considerable amount of profit can be generated by appropriately operating the BSCS as an energy storage system.  
In addition to the potential economic and environmental benefits, it is demonstrated that the reliability of power systems can be significantly improved if the battery-swapping mode is adopted \cite{evaluation_system_reliability_battery_exchange_mode}. All these models and evaluation methods
\cite{evaluation_envoronmental_impact,evaluation_economic_value_PV,evaluation_economic_microgrid, evaluation_system_reliability_battery_exchange_mode} demonstrate that appropriate capacity planning of the BSCS is very important in determining the overall benefits, especially when the  capital cost of batteries is high.

The second stream focuses on the \textit{planning} and \textit{design} of BSCSs \cite{planning_infrastructure_planning,planning_optimization_issue, planning_electric_taxi, planning_battery_swapping_distribution_systems, planning_design_mobile_charging_service}. The key focus of this stream of research is to strategically determine the location and capacity of a BSCS or a network of BSCSs. Mak \textit{et al.} \cite{planning_infrastructure_planning}  study the infrastructure planning problem for EVs with the battery-swapping mode. The planning problem is to locate a number of battery swapping stations at strategic locations along a network of freeways. Based on two different planning objectives (cost-concerned or profit-driven), the planning problem is formulated as two different robust optimization problems. By solving these two robust optimization problems,  the corresponding strategies to minimize the total expected planning cost and to achieve a certain amount of profit are obtained in  \cite{planning_infrastructure_planning}. In \cite{planning_optimization_issue}, the authors propose to utilizing the gigantic storage system (i.e., aggregated batteries from a BSCS) to integrate renewable energy into the power system. The objective is to determine the optimal capacity of the gigantic storage system such that the total cost of the system is minimized. The authors find that the optimal capacity highly depends on the number of charge-discharge cycles of the batteries. In \cite{planning_electric_taxi}, the authors propose an battery swapping station planning algorithm for urban electrical taxis, whose target was to minimize the total time required for refueling all  of the taxis. 

The third stream focuses on the \textit{scheduling} and \textit{operation} of BSCSs \cite{scheduling_bss_OR_letters, scheduling_power_supply_optimization_tvt,scheduling_phev_exchange_station,scheduling_hybrid_service_scheduling_bss,scheduling_optimal_management_policies,scheduling_battery_purchasing_optimization,scheduling_own_control_battery_swapping}. As defined by \cite{scheduling_bss_OR_letters}, the scheduling of a BSCS is a new inventory management problem, whose key is to obtain an optimal charging (and possibly discharging) strategy that optimizes a certain objective (e.g., minimize the total charging cost, etc) and guarantees a certain amount of FBs simultaneously. Based on different application scenarios and assumptions, several papers have further studied this new inventory management problem. For instance,  in \cite{scheduling_power_supply_optimization_tvt} and \cite{scheduling_hybrid_service_scheduling_bss}, the authors propose an optimal cost-effective operation of a BSCS with dynamic electricity price and uncertain FB demand, and the authors of \cite{scheduling_hybrid_service_scheduling_bss} further investigate the economic benefits of services like battery-to-grid and battery-to-battery. In \cite{scheduling_optimal_management_policies}, the optimal charging and discharging policies for maximizing the expected total profit over a fixed time horizon (i.e., short-term) have been proposed. Different from \cite{scheduling_optimal_management_policies}, the authors of \cite{scheduling_battery_purchasing_optimization} investigate the joint optimization of the battery charging and purchasing strategies for a single BSCS and a network of BSCSs. Therefore, the long-term investment in batteries and the short-term operational cost can be balanced.

Despite the above literature, there has been little work focusing on the QoS analysis of BSCSs. Our previous work \cite{evaluation_own_queueing_model_battery_swapping} aims to fill this gap by proposing an MQN model for the BSCS, in which only the steady-state QoS analysis has been studied. Different from \cite{evaluation_own_queueing_model_battery_swapping}, this paper performs an asymptotic queueing analysis for the MQN model, and aims to link the theoretical results to the practical capacity planning of BSCSs. Note that this work mainly contributes to the first stream of research, but our asymptotic queueing analysis provides rich insights for the multi-stage capacity planning of BSCSs, which thus also contributes to the second stream of research.

\section{Steady-State Distribution}
\label{queueing_network_models}
In this section, we derive the balane equations for the MQN based on the two-dimensional CTMC approach and then calculate its steady-state distribution. We first present the assumptions and notations regarding the proposed MQN model. 


\subsection{Assumptions and Notations} 
As illustrated in Fig. \ref{the_model}, we assume that the EVs arrive at the BSCS according to a Poisson process with rate $ \lambda $. Each EV will either be served immediately or wait for service, and then immediately leave the system after service. We use $S$ to denote the total number of SSs, and each SS is assumed to have an exponentially distributed service time with service rate $ \nu $. The  total number of parking spaces is assumed to be $ V=N-S $, where $ N $ denotes the total capacity of the open EV-queue. Recall that there exists a fixed number of batteries circulating through the closed battery-queue at all times, and this fixed number of batteries is denoted by $ B $. We further use $ C $ to denote the number of  CSs. Due to the randomness of the initial state-of-charge (SoC) of the DBs, we assume that the charging time is exponentially distributed with rate $ \mu $\footnote{We point out that although the assumptions about the swapping time and the charging time are motivated for mathematical tractability, they are in general close to reality and widely used in the queueing theory related literature (e.g.,  \cite{ev_power_allocation, evaluation_economic_value_PV, planning_infrastructure_planning}). }. Since it is often practically feasible to have enough space to store all the batteries in the BSCS, we assume that the buffer sizes of the DB-queue and the FB-queue are all infinite. Therefore, the closed battery-queue in the MQN does not have the blocking phenomenon \cite{closed_queueing_networks}.

\subsection{Balance Equations and Steady-State Distribution}
\label{steady_state_distribution_exponential}
We use a triple $ (n,b,j) $ to denote the state of having $ n $ EVs (waiting and in-service) in the open EV-queue, $ b $ FBs (waiting and in-service) in the FB-queue, and $ j $ DBs (waiting and in-service) in the DB-queue, where $ n\in \{0,1,\cdots, N\}$, $b$, $j\in \{0,1,\cdots, B\} $. Note that $ b+j=B $ always holds, we thus simply use $ \pi_{n,b} $ to denote the steady-state probability of being in state $ (n,b,j) $. 

To analyze the steady-state distribution of the MQN, we show the two-dimensional CTMC of the proposed MQN in Fig. \ref{CTMC}. The transition equations  can be organized into five cases which respectively corresponds to the five different operating regions as follows:

\textbf{Region 1}: the first row in Fig. \ref{CTMC}, i.e., when $ n=0 $. We have the following three cases to represent the transition equations:	
\begin{itemize}
	\item for $ n=0 $ and $ b = 0 $, i.e., the state $ (0,0) $,
	\begin{align}
	\pi_{0,0}\Big(\mu\min\{B,C\}+\lambda\Big)=\pi_{1,1}\nu,
	\label{first_case_first_region}
	\end{align}
	\item for $ n=0 $ and $ b=B $, i.e., the state $ (0,B) $,
	\begin{align}
	\pi_{0,B}\lambda=\pi_{0,B-1}\mu,
	\end{align}
	\item for $ n=0 $ and $  1\leq b\leq B-1 $, we have
	\begin{align}
	\pi_{0,b}\Big(\mu\min\{B-b,C\}+\lambda\Big)=\pi_{0,b-1}\mu\min\{B-b+1,C\}+ \pi_{1,b+1}\nu.
	\end{align}
\end{itemize}

\begin{figure}
	\centering
	\begin{tikzpicture}[scale=0.53, transform shape,>=latex]
	\draw[dashed] (-7,5) to (-7,-8) to (9,-8) to (-3,5) to (-7,5);
	
	\draw[dotted] (-2.4,5) to (15,5) to (15, -8) to (9.6,-8) to (-2.4,5);
	
	\path
	(-6,4)  node (00) [shape=circle,draw] {$ \ 0, 0 \  $  }
	(-4,4)  node (01) [shape=circle,draw] {$ \ 0, 1 \  $ }
	( 0,4)  node (02) [shape=circle,draw] {{\small 0, B-7} }
	( 2,4)  node (03) [shape=circle,draw] {{\small 0, B-6} }
	( 4,4)  node (04) [shape=circle,draw] {{\small 0, B-5} }
	( 6,4)  node (05) [shape=circle,draw] {{\small 0, B-4} }
	( 8,4)  node (06) [shape=circle,draw] {{\small 0, B-3} }
	( 10,4)  node (07) [shape=circle,draw] {{\small 0, B-2} }
	( 12,4)  node (08) [shape=circle,draw] {{\small 0, B-1} }
	( 14,4)  node (09) [shape=circle,draw] { $\ 0 $, B }

	(-6,2) node (10)[shape=circle,draw] { $ \ 1, 0 \ $ }
	(-4,2) node (11)[shape=circle,draw] { $ \ 1, 1 \ $ }
	(-2,2) node (12)[shape=circle,draw] { $ \ 1, 2 \ $}
	( 2,2)  node (13) [shape=circle,draw] {{\small 1, B-6} }
	( 4,2)  node (14) [shape=circle,draw] {{\small 1, B-5} }
	( 6,2)  node (15) [shape=circle,draw] {{\small 1, B-4} }
	( 8,2)  node (16) [shape=circle,draw] {{\small 1, B-3} }
	( 10,2)  node (17) [shape=circle,draw] {{\small 1, B-2} }
	( 12,2)  node (18) [shape=circle,draw] {{\small 1, B-1} }
	( 14,2)  node (19) [shape=circle,draw] { $\ 1 $, B }
	
	(-6,0) node (20)[shape=circle,draw] { $ \ 2, 0 \ $ }
	(-4,0) node (21)[shape=circle,draw] { $ \ 2, 1 \ $ }
	(-2,0) node (22)[shape=circle,draw] { $ \ 2, 2 \ $ }
	( 0,0) node (23)[shape=circle,draw] { $ \ 2, 3 \ $ }
	( 4,0)  node (24) [shape=circle,draw] {{\small 2, B-5} }
	( 6,0)  node (25) [shape=circle,draw] {{\small 2, B-4} }
	( 8,0)  node (26) [shape=circle,draw] {{\small 2, B-3} }
	( 10,0)  node (27) [shape=circle,draw] {{\small 2, B-2} }
	( 12,0)  node (28) [shape=circle,draw] {{\small 2, B-1} }
	( 14,0)  node (29) [shape=circle,draw] { $\ 2 $, B }

	(-6,-2) node (30)[shape=circle,draw] { $\ 3, 0 \ $ }
	(-4,-2) node (31)[shape=circle,draw] { $\ 3, 1 \ $ }
	(-2,-2) node (32)[shape=circle,draw] { $\ 3, 2 \ $ }
	( 0,-2) node (33)[shape=circle,draw] { $\ 3, 3 \ $ }	
	( 2,-2) node (34)[shape=circle,draw] { $\ 3, 4 \ $ }
	( 6,-2)  node (35) [shape=circle,draw]  {{\small 3, B-4} }
	( 8,-2)  node (36) [shape=circle,draw]  {{\small 3, B-3} }
	( 10,-2)  node (37) [shape=circle,draw] {{\small 3, B-2} }
	( 12,-2)  node (38) [shape=circle,draw] {{\small 3, B-1} }
	( 14,-2)  node (39) [shape=circle,draw] { $\ 3 $, B }

	(-6,-5) node (n0)[shape=circle,draw] { {\small N-1,0} }
	(-4,-5) node (n1)[shape=circle,draw] { {\small N-1,1} }
	(-2,-5) node (n2)[shape=circle,draw] { {\small N-1,2} }
	( 0,-5) node (n3)[shape=circle,draw] { {\small N-1,3} }	
	( 2,-5) node (n4)[shape=circle,draw] { {\small N-1,4} }
	( 4,-5) node (n5)[shape=circle,draw] { {\small N-1,5} }
	
	( 8, -5) node (n6)[shape=circle,draw] { {\scriptsize N-1,B-3} }
	( 10,-5) node (n7)[shape=circle,draw] { {\scriptsize N-1,B-2} }
	( 12,-5) node (n8)[shape=circle,draw] { {\scriptsize N-1,B-1} }	
	( 14,-5) node (n9)[shape=circle,draw] {  {\footnotesize N-1, B} }

	(-6,-7) node (N0)[shape=circle,draw] { N, $  0\  $ }
	(-4,-7) node (N1)[shape=circle,draw] { N, $  1\ $ }
	(-2,-7) node (N2)[shape=circle,draw] { N, $  2\ $ }
	( 0,-7) node (N3)[shape=circle,draw] { N, $  3\ $ }	
	( 2,-7) node (N4)[shape=circle,draw] { N, $  4\ $ }
	( 4,-7) node (N5)[shape=circle,draw] { N, $  5\ $ }
	( 6, -7) node (N6)[shape=circle,draw] {N, $  6\ $ }	
	( 10,-7) node (N7)[shape=circle,draw] { {\small N,B-2} }
	( 12,-7) node (N8)[shape=circle,draw] { {\small N,B-1} }	
	( 14,-7) node (N9)[shape=circle,draw] {  N, B };
	
	\draw (-2,4) node[]{{\Huge $ \cdots $}};
	\draw (0, 2) node[]{{\Huge $ \cdots $}};
	\draw (2, 0) node[]{{\Huge $ \cdots $}};
	\draw (4,-2) node[]{{\Huge $ \cdots $}};
	\draw (6,-5) node[]{{\Huge $ \cdots $}};
	\draw (8,-7) node[]{{\Huge $ \cdots $}};
	
	\draw (-6,-3.5) node[]{{\LARGE $ \vdots $}};
	\draw (-4,-3.5) node[]{{\LARGE $ \vdots $}};
	\draw (-2,-3.5) node[]{{\LARGE $ \vdots $}};
	\draw ( 0,-3.5) node[]{{\LARGE $ \vdots $}};
	\draw ( 2,-3.5) node[]{{\LARGE $ \vdots $}};
	\draw ( 4,-3.5) node[]{{\LARGE $ \vdots $}};
	\draw ( 6,-3.5) node[]{{\LARGE $ \vdots $}};
	\draw ( 8,-3.5) node[]{{\LARGE $ \vdots $}};
	\draw ( 10,-3.5) node[]{{\LARGE $ \vdots $}};
	\draw ( 12,-3.5) node[]{{\LARGE $ \vdots $}};
	\draw (14,-3.5) node[]{{\LARGE $ \vdots $}};
	\draw[->] (00) to[bend left] node[below]{$ C\mu$ }(01);
	\draw[->] (02) to[bend left] node[below]{$ 7\mu$ }(03);
	\draw[->] (03) to[bend left] node[below]{$ 6\mu$ }(04);
	\draw[->] (04) to[bend left] node[below]{$ 5\mu$ }(05);
	\draw[->] (05) to[bend left] node[below]{$ 4\mu$ }(06);
	\draw[->] (06) to[bend left] node[below]{$ 3\mu$ }(07);
	\draw[->] (07) to[bend left] node[below]{$ 2\mu$ }(08);
	\draw[->] (08) to[bend left] node[below]{$ \mu$ }(09);
	
	\draw[->] (10) to[bend left] node[below]{$ C\mu$ }(11);
	\draw[->] (11) to[bend left] node[below]{$ C\mu$ }(12);
	\draw[->] (13) to[bend left] node[below]{$ 6\mu$ }(14);
	\draw[->] (14) to[bend left] node[below]{$ 5\mu$ }(15);
	\draw[->] (15) to[bend left] node[below]{$ 4\mu$ }(16);
	\draw[->] (16) to[bend left] node[below]{$ 3\mu$ }(17);
	\draw[->] (17) to[bend left] node[below]{$ 2\mu$ }(18);
	\draw[->] (18) to[bend left] node[below]{$ \mu$ }(19);
	
	\draw[->] (20) to[bend left] node[below]{$ C\mu$ }(21);
	\draw[->] (21) to[bend left] node[below]{$ C\mu$ }(22);
	\draw[->] (22) to[bend left] node[below]{$ C\mu$ }(23);
	\draw[->] (24) to[bend left] node[below]{$ 5\mu$ }(25);
	\draw[->] (25) to[bend left] node[below]{$ 4\mu$ }(26);
	\draw[->] (26) to[bend left] node[below]{$ 3\mu$ }(27);
	\draw[->] (27) to[bend left] node[below]{$ 2\mu$ }(28);
	\draw[->] (28) to[bend left] node[below]{$ \mu$  }(29);
	
	\draw[->] (30) to[bend left] node[below]{$ C\mu$ }(31);
	\draw[->] (31) to[bend left] node[below]{$ C\mu$ }(32);
	\draw[->] (32) to[bend left] node[below]{$ C\mu$ }(33);
	\draw[->] (33) to[bend left] node[below]{$ C\mu$ }(34);
	\draw[->] (35) to[bend left] node[below]{$ 4\mu$ }(36);
	\draw[->] (36) to[bend left] node[below]{$ 3\mu$ }(37);
	\draw[->] (37) to[bend left] node[below]{$ 2\mu$ }(38);
	\draw[->] (38) to[bend left] node[below]{$ \mu$  }(39);
	
	\draw[->] (n0) to[bend left] node[below]{$ C\mu$ }(n1);
	\draw[->] (n1) to[bend left] node[below]{$ C\mu$ }(n2);
	\draw[->] (n2) to[bend left] node[below]{$ C\mu$ }(n3);
	\draw[->] (n3) to[bend left] node[below]{$ C\mu$ }(n4);
	\draw[->] (n4) to[bend left] node[below]{$ C\mu$ }(n5);
	\draw[->] (n6) to[bend left] node[below]{$ 3\mu$ }(n7);
	\draw[->] (n7) to[bend left] node[below]{$ 2\mu$ }(n8);
	\draw[->] (n8) to[bend left] node[below]{$ \mu$ } (n9);
	
	\draw[->] (N0) to[bend left] node[below]{$ C\mu$ }(N1);
	\draw[->] (N1) to[bend left] node[below]{$ C\mu$ }(N2);
	\draw[->] (N2) to[bend left] node[below]{$ C\mu$ }(N3);
	\draw[->] (N3) to[bend left] node[below]{$ C\mu$ }(N4);
	\draw[->] (N4) to[bend left] node[below]{$ C\mu$ }(N5);
	\draw[->] (N5) to[bend left] node[below]{$ C\mu$ }(N6);
	\draw[->] (N7) to[bend left] node[below]{$ 2\mu$ }(N8);
	\draw[->] (N8) to[bend left] node[below]{$ \mu$ }(N9);
	
	\draw[->] (00) to[bend right] node[right]{$ \lambda $ }(10);
	\draw[->] (10) to[bend right] node[right]{$ \lambda $ }(20);
	\draw[->] (20) to[bend right] node[right]{$ \lambda $ }(30);
	\draw[->] (n0) to[bend right] node[right]{$ \lambda $ }(N0);
	
	\draw[->] (01) to[bend right] node[right]{$ \lambda $ }(11);
	\draw[->] (11) to[bend right] node[right]{$ \lambda $ }(21);
	\draw[->] (21) to[bend right] node[right]{$ \lambda $ }(31);
	\draw[->] (n1) to[bend right] node[right]{$ \lambda $ }(N1);

	\draw[->] (02) to[bend right] node[right]{$ \lambda $ }(-0.3, 2.7);
	\draw[->] (13)  to[bend right] node[right]{$ \lambda $ }( 1.7, 0.7);
	\draw[->] (24) to[bend right] node[right]{$ \lambda $ }(  3.7, -1.3);
	
	\draw[->] (-0.2, 1.4) to[bend right] node[right]{$ \lambda $ }(23);
	\draw[->] (-2.2, 3.4) to[bend right] node[right]{$ \lambda $ }(12);
	\draw[->] (1.8, -0.6) to[bend right] node[right]{$ \lambda $ }(34);

	\draw[->] (12) to[bend right] node[right]{$ \lambda $ }(22);
	\draw[->] (22) to[bend right] node[right]{$ \lambda $ }(32);
	\draw[->] (n2) to[bend right] node[right]{$ \lambda $ }(N2);

	\draw[->] (03) to[bend right] node[right]{$ \lambda $ }(13);
	\draw[->] (23) to[bend right] node[right]{$ \lambda $ }(33);
	\draw[->] (n3) to[bend right] node[right]{$ \lambda $ }(N3);
	
	\draw[->] (04) to[bend right] node[right]{$ \lambda $ }(14);
	\draw[->] (14) to[bend right] node[right]{$ \lambda $ }(24);
	\draw[->] (n4) to[bend right] node[right]{$ \lambda $ }(N4);
	
	\draw[->] (05) to[bend right] node[right]{$ \lambda $ }(15);
	\draw[->] (15) to[bend right] node[right]{$ \lambda $ }(25);
	\draw[->] (25) to[bend right] node[right]{$ \lambda $ }(35);
	\draw[->] (n5) to[bend right] node[right]{$ \lambda $ }(N5);
	
	\draw[->] (06) to[bend right] node[right]{$ \lambda $ }(16);
	\draw[->] (16) to[bend right] node[right]{$ \lambda $ }(26);
	\draw[->] (26) to[bend right] node[right]{$ \lambda $ }(36);

	\draw[->] (07) to[bend right] node[right]{$ \lambda $ }(17);
	\draw[->] (17) to[bend right] node[right]{$ \lambda $ }(27);
	\draw[->] (27) to[bend right] node[right]{$ \lambda $ }(37);
	\draw[->] (n7) to[bend right] node[right]{$ \lambda $ }(N7);
	
	\draw[->] (08) to[bend right] node[right]{$ \lambda $ }(18);
	\draw[->] (18) to[bend right] node[right]{$ \lambda $ }(28);
	\draw[->] (28) to[bend right] node[right]{$ \lambda $ }(38);
	\draw[->] (n8) to[bend right] node[right]{$ \lambda $ }(N8);
	
	\draw[->] (09) to[bend right] node[right]{$ \lambda $ }(19);
	\draw[->] (19) to[bend right] node[right]{$ \lambda $ }(29);
	\draw[->] (29) to[bend right] node[right]{$ \lambda $ }(39);
	\draw[->] (n9) to[bend right] node[right]{$ \lambda $ }(N9);
	
	\draw[->] (30) to[bend right] node[right]{$ \lambda $ } (-6.3, -3.3);
	\draw[->] (-6.3, -3.6) to[bend right] node[right]{$ \lambda $ } (n0);
	
	\draw[->] (31) to[bend right] node[right]{$ \lambda $ } (-4.3, -3.3);
	\draw[->] (-4.3, -3.6) to[bend right] node[right]{$ \lambda $ } (n1);
	
	\draw[->] (32) to[bend right] node[right]{$ \lambda $ } (-2.3, -3.3);
	\draw[->] (-2.3, -3.6) to[bend right] node[right]{$ \lambda $ } (n2);
	
	\draw[->] (33) to[bend right] node[right]{$ \lambda $ } (-0.3, -3.3);
	\draw[->] (-0.3, -3.6) to[bend right] node[right]{$ \lambda $ } (n3);
	
	\draw[->] (34) to[bend right] node[right]{$ \lambda $ } (1.7, -3.3);
	\draw[->] (1.7, -3.6) to[bend right] node[right]{$ \lambda $ } (n4);
	
	\draw[->] (3.7, -3.6) to[bend right] node[right]{$ \lambda $ } (n5);	
	
	\draw[->] (35) to[bend right] node[right]{$ \lambda $ } (5.7, -3.3);
	\draw[->] (5.8, -5.6) to[bend right] node[right]{$ \lambda $ } (N6);

	\draw[->] (36) to[bend right] node[right]{$ \lambda $ } (7.7, -3.3);
	\draw[->] (7.7, -3.6) to[bend right] node[right]{$ \lambda $ } (n6);
	\draw[->] (n6) to[bend right] node[right]{$ \lambda $ } (7.7, -6.4);
	
	\draw[->] (37) to[bend right] node[right]{$ \lambda $ } (9.7, -3.3);
	\draw[->] (9.7, -3.6) to[bend right] node[right]{$ \lambda $ } (n7);
	
	\draw[->] (38) to[bend right] node[right]{$ \lambda $ } (11.7, -3.3);
	\draw[->] (11.7, -3.6) to[bend right] node[right]{$ \lambda $ } (n8);
	
	\draw[->] (39) to[bend right] node[right]{$ \lambda $ } (13.7, -3.3);
	\draw[->] (13.7, -3.6) to[bend right] node[right]{$ \lambda $ } (n9);
	\draw[->] (n1) to node[below]{$ \nu $ }(-5.6, -3.5);
	\draw[->] (n2) to node[below]{$ 2\nu $ }(-3.6, -3.5);
	\draw[->] (n3) to node[below]{$ 3\nu $ }(-1.6, -3.5);
	\draw[->] (n4) to node[below]{$ 4\nu $ }( 0.4, -3.5);	
	\draw[->] (n5) to node[below]{$ 5\nu $ }( 2.4, -3.5);
	\draw[->] (n6) to node[below]{$ S\nu $ }( 6.4, -3.5);
	\draw[->] (n7) to node[below]{$ S\nu $ }( 8.4, -3.5);
	\draw[->] (n8) to node[below]{$ S\nu $ }( 10.4, -3.5);
	\draw[->] (n9) to node[below]{$ S\nu $ }( 12.4, -3.5);
	\draw[->] (11) to node[below]{$ \nu $ }(00);
	\draw[->] (12) to node[below]{$ \nu $ } (01);
	\draw[->] (13) to node[below]{$ \nu $ } (02);
	\draw[->] (14) to node[below]{$ \nu $ } (03);
	\draw[->] (15) to node[below]{$ \nu $ } (04);
	\draw[->] (16) to node[below]{$ \nu $ } (05);
	\draw[->] (17) to node[below]{$ \nu $ } (06);
	\draw[->] (18) to node[below]{$ \nu $ } (07);
	\draw[->] (19) to node[below]{$ \nu $ } (08);

	\draw[->] (21) to node[below]{$ \nu $ }(10);
	\draw[->] (22) to node[below]{$ 2\nu $ }(11);
	\draw[->] (23) to node[below]{$ 2\nu $ }(12);
	\draw[->] (24) to node[below]{$ 2\nu $ }(13);
	\draw[->] (25) to node[below]{$ 2\nu $ }(14);
	\draw[->] (26) to node[below]{$ 2\nu $ }(15);
	\draw[->] (27) to node[below]{$ 2\nu $ }(16);
	\draw[->] (28) to node[below]{$ 2\nu $ }(17);
	\draw[->] (29) to node[below]{$ 2\nu $ }(18);

	\draw[->] (31) to node[below]{$ \nu $ }(20);
	\draw[->] (32) to node[below]{$ 2\nu $ }(21);
	\draw[->] (33) to node[below]{$ 3\nu $ }(22);
	\draw[->] (34) to node[below]{$ 3\nu $ }(23);
	\draw[->] (35) to node[below]{$ 3\nu $ }(24);
	\draw[->] (36) to node[below]{$ 3\nu $ }(25);
	\draw[->] (37) to node[below]{$ 3\nu $ }(26);
	\draw[->] (38) to node[below]{$ 3\nu $ }(27);
	\draw[->] (39) to node[below]{$ 3\nu $ }(28);

	\draw[->] (N1) to node[below]{$ \nu $ }(n0);
	\draw[->] (N2) to node[below]{$ 2\nu $ }(n1);
	\draw[->] (N3) to node[below]{$ 3\nu $ }(n2);
	\draw[->] (N4) to node[below]{$ 4\nu $ }(n3);
	\draw[->] (N5) to node[below]{$ 5\nu $ }(n4);
	\draw[->] (N6) to node[below]{$ 6\nu $ }(n5);
	\draw[->] (N7) to node[below]{$ S\nu $ }(n6);
	\draw[->] (N8) to node[below]{$ S\nu $ }(n7);
	\draw[->] (N9) to node[below]{$ S\nu $ }(n8);
	\draw[->] (01)       to[bend left] node[below]{$ C\mu$ }(-2.5,4.3);
	\draw[->] (-1.7,4.3) to[bend left] node[below]{$ 8\mu$ }(02);
	
	\draw[->] (12)       to[bend left] node[below]{$ C\mu$ }(-0.5,2.3);
	\draw[->] (0.3,2.3)  to[bend left] node[below]{$ 7\mu$ }(13);
	
	\draw[->] (23)       to[bend left] node[below]{$ C\mu$ }( 1.5,0.3);
	\draw[->] (2.3,0.3)  to[bend left] node[below]{$ 6\mu$ }(24);
	
	\draw[->] (34)       to[bend left] node[below]{$ C\mu$ }( 3.5,-1.7);
	\draw[->] (4.3,-1.7) to[bend left] node[below]{$ 5\mu$ }(35);
	
	\draw[->] (n5)       to[bend left] node[below]{$ C\mu$ }( 5.5,-4.7);
	\draw[->] (6.3,-4.8) to[bend left] node[below]{$ 4\mu$ }(n6);
	
	\draw[->] (N6)       to[bend left] node[below]{$ C\mu$ }( 7.5,-6.7);
	\draw[->] (8.3,-6.7) to[bend left] node[below]{$ 3\mu$ }(N7);
	
	\end{tikzpicture}
	\caption{Illustration of the CTMC for the MQN.  Each node \circled{$ n,b $} with $ n\in\{0, \cdots, N\} $ and $ b\in \{0, \cdots, B\} $ denotes a two-dimensional state with $ n $ EVs  and $ b $ FBs in the open EV-queue and the FB-queue, respectively. For simplicity, we assume $ C\leq B $ in this figure.}
	\label{CTMC}
\end{figure}
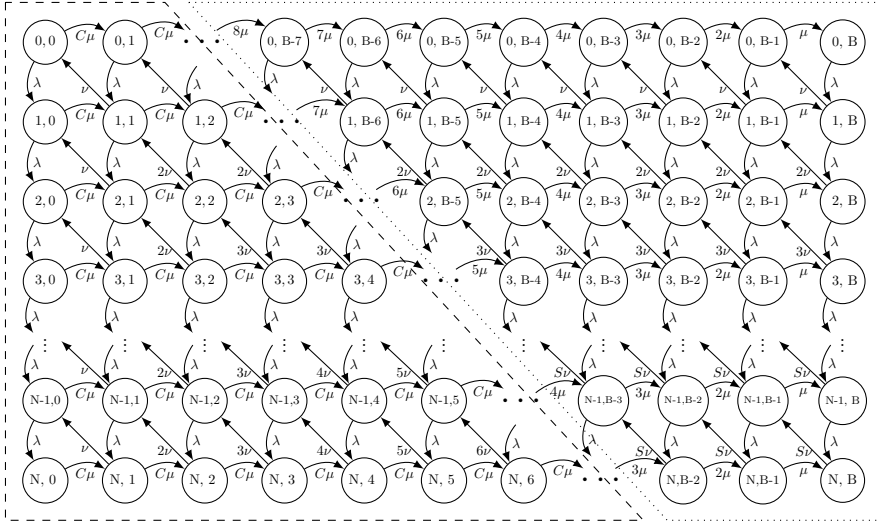

\textbf{Region 2}: the last row in Fig. \ref{CTMC}, i.e., when $ n=N $. We again have the following three cases:
\begin{itemize}
	\item for $ n=N $ and $ b =0 $, i.e., the state $ (N,0) $,
	\begin{align}
	\pi_{N,0}\mu\min\{B,C\}=\pi_{N-1,0}\lambda,
	\end{align}
	
	\item for $ n=N $ and $ b = B $, i.e., the state $ (N,B) $,
	\begin{align}
	\pi_{N,B}\nu\min\{N,B,S\}=\pi_{N-1,B}\lambda+\pi_{N,B-1}\mu,
	\end{align}
	
	\item for $ n=N $ and $ 1\leq b \leq B-1 $, we have
	\begin{align}
	\pi_{N,b}\big(\mu\min\{B-b,C\}+\nu\min\{N,b,S\}\big)
	=  \pi_{N-1,b}\lambda+ \pi_{N,b-1}\mu\min\{B-b+1,C\}.
	\end{align}
\end{itemize}

\textbf{Region 3}: the middle of the leftmost column in Fig. \ref{CTMC}, i.e., for $ b=0 $ and $ 1\leq n \leq N-1 $:
\begin{equation}
\pi_{n,0}\Big(\mu\min\{B,C\}+\lambda\Big) =\pi_{n-1,0}\lambda+\pi_{n+1,1}\nu.
\label{leftmost_column}
\end{equation}

\textbf{Region 4}: the middle of the rightmost column in Fig. \ref{CTMC}, i.e., for $ b=B $  and $ 1\leq n \leq N-1 $:
\begin{equation}
\pi_{n,B}\big(\nu\min\{N,B,S\}+\lambda\big) =\pi_{n-1,B}\lambda+ \pi_{n,B-1}\mu.
\label{rightmost_column}
\end{equation}

\textbf{Region 5}: the middle of the whole transition diagram in Fig. \ref{CTMC}, i.e., for $ 1\leq n \leq N-1 $ and $ 1\leq b\leq B-1 $:
\begin{align}
\nonumber
  &\pi_{n,b}\big(\mu\min\{B-b,C\}+\lambda+\nu\min\{n,b,S\}\big)\\
= &\pi_{n,b-1}\mu\min\{B-b+1,C\} + \pi_{n-1,b}\lambda + \pi_{n+1,b+1}\nu\min\{n+1,b+1,S\}.
\label{middle_whole_digram}
\end{align}

By organizing all the above transition equations into a matrix form, we have the following linear systems:
\begin{equation}
\boldsymbol{\pi}\mathbf{Q}=\mathbf{0}, \mathrm{\ and\ } \boldsymbol{\pi}\mathbf{e}=1, 
\label{finite_linear_system}
\end{equation}
where $ \boldsymbol{\pi}=[\boldsymbol{\pi}_0,\boldsymbol{\pi}_1,\cdots,\boldsymbol{\pi}_B] $ with $ \boldsymbol{\pi}_b=[\pi_{0,b},\pi_{1,b},\cdots,\pi_{N,b}] $, $\forall b \in\{0,1,\cdots,B\}$, and $ \mathbf{e} $ is an $(N+1)(B+1)\times 1$ column vector all of whose entries are 1s. $ \mathbf{Q} $ is the infinitesimal generator matrix or the transition rate matrix, given as follows:
\begin{align}
\mathbf{Q} = \left[
\begin{array}{llllllllllll}
\mathbf{L}_{\mathrm{00}} & \mathbf{F}_{\mathrm{01}} & \mathbf{0} & \cdots\\
\mathbf{D}_{\mathrm{10}} & \mathbf{L} & \mathbf{F}& \mathbf{0} & \cdots\\
\mathbf{0}& \mathbf{D} & \mathbf{L} & \mathbf{F}& \mathbf{0} & \cdots\\
& \ddots  & \ddots & \ddots & \ddots\\
& \cdots& \mathbf{0} & \mathbf{D} & \mathbf{L} & \mathbf{F}^{\mathrm{N}}_{\mathrm{10}}\\
& & \cdots & \mathbf{0} & \mathbf{D}^{\mathrm{N}}_{\mathrm{01}} & \mathbf{L}^{\mathrm{N}}_{\mathrm{00}}\\
\end{array}
\right],
\label{Q_matrix}
\end{align}
where $\mathbf{L}_{\mathrm{00}}$ is an $ (N+1)S\times(N+1)S $ matrix corresponding to $ \boldsymbol{\pi}_{0:S-1}\triangleq [ \boldsymbol{\pi}_0,\cdots,\boldsymbol{\pi}_{S-1}] $, and $\mathbf{L}^{\mathrm{N}}_{\text{00}}$ is an $ (N+1)(C+1)\times(N+1)(C+1) $ matrix corresponding to $\boldsymbol{\pi}_{B-C:B} \triangleq [ \boldsymbol{\pi}_{B-C},\cdots,\boldsymbol{\pi}_{B}] $.  Additionally, matrices $ \mathbf{F} $, $ \mathbf{L} $, and $ \mathbf{D} $ are respectively given by
\begin{align}
\mathbf{F} = \left[
\begin{array}{llll}
C\mu &  \\
& C\mu & \\
&     & \ddots & \\
&     &        & C\mu\\
\end{array}
\right],
\mathbf{L} = \left[
\begin{array}{llll}
m_0 & \lambda \\
& m_1 & \ddots\\
&     & \ddots & \lambda\\
&     &        & m_N\\
\end{array}
\right],
\mathbf{D} = \left[
\begin{array}{llcl}
0   &    \\
d_1 & \ddots  \\
& \ddots & 0\\
& & d_N & 0\\
\end{array}
\right],
\label{matrices_F_L_D}
\end{align}
where $m_n = -(\mathbb{I}_{(n\not=N)}\lambda + C\mu + \nu\min\{n,S\})$, $ \forall n = \{0, 1, \dots, N\}$, and $d_n = \min\{n,S\}\nu$, $\forall n = \{1, 2, \dots, N\}$. Note that these three matrices are all of $ (N+1)\times(N+1) $. For brevity, we skip the details of matrices $\mathbf{L}_{\text{00}}, \mathbf{F}_{\text{01}}, \mathbf{D}_{\text{10}},\mathbf{L}^{\mathrm{N}}_{\text{00}}, \mathbf{D}^{\mathrm{N}}_{\text{01}}$, and $ \mathbf{F}^{\mathrm{N}}_{\text{10}}$, whose entries can be found by the balance equations listed in \eqref{first_case_first_region}-\eqref{middle_whole_digram}. 

From Fig. \ref{CTMC} we can see that the finite-state CTMC defined by $ \mathbf{Q} $  is ergodic (i.e., irreducible and positive recurrent), which means that the finite linear system \eqref{finite_linear_system} has a unique solution. Note that we have $ (N+1)(B+1)+1  $ equations but only $ (N+1)(B+1) $ variables in \eqref{finite_linear_system}. Therefore, one of the equations contained in matrix $ \mathbf{Q} $ should be eliminated  in order to obtain the unique steady-state distribution $ \boldsymbol{\pi} $ (e.g., by eliminating the first column of $ \mathbf{Q}  $).

\subsection{Performance Metric: Blocking Probability}
The blocking probability is a classical performance metric, which measures the EVs' probability of being blocked from joining the open EV-queue in our context. Based on the steady-state distribution and the PASTA property, the blocking probability can be expressed as a function of $ N,S,C $, and $ B $ as follows\footnote{Note that $ V $ can be easily calculated based on $ N $ and $ S $, we thus use $ N, S, C $, and $ B $ to represent the four planning parameters.}:
	\begin{equation}
	\mathbb{P}_{\mathrm{MQN}}(N,S,C,B)=\sum_{b=0}^{B}\pi_{N,b}.
	\label{prob_service_failure}
	\end{equation}	
Many other performance metrics can be defined once the steady-state distribution is obtained. For ease of presentation, our theoretical analysis will focus on the blocking probability, but our numerical simulation  in Section \ref{sec_simulation_results} will show other performance metrics  as well.

\section{Asymptotic Analysis Based on Multi-Stage Capacity Planning}
\label{QoS_Multi_Stage_Capacity_Planning}
In this section, we present the concept of multi-stage capacity planning and characterize the asymptotic behavior of the blocking probability when the number of batteries approaches infinity.

\subsection{A Three-Stage Scheme for Studying the Impact of $ N,S,C, $ and $ B $}
In practice, the four parameters are by nature in different planning timescales. Specifically, i) the capacity of the open EV-queue $ N $ and  the number of SSs $ S $ are primarily constrained by the area of the land and are difficult to change once fixed; ii) the number of CSs $ C $ is mainly constrained by the power transmission capacity from the grid to the BSCS, and is relatively more flexible to change than  $ N $ and $ S $. However, unlike the previous three parameters, the number of batteries $ B $ is very flexible to change during the operation of the BSCS. Therefore, to have a reasonable capacity planning of the BSCS, we propose a three-stage scheme for studying the impact of the four parameters, which includes the first stage (Stage-I) for studying the impact of $ N $ and $ S $, the second stage (Stage-II) for studying the impact of $ C $, and the third stage (Stage-III) for studying the impact of $ B $. 

The main advantage of the three-stage study is the potential of better exploiting different levels of flexibilities among the four parameters to yield an optimal multi-timescale capacity planning solution. However, this paper focuses on only quantifying how these four parameters influence the blocking probability. Therefore, finding a multi-timescale capacity planning solution that optimizes a specifically-defined objective is beyond the scope of this paper. Below, we will show how these four parameters influence the blocking probability in their corresponding stages.

\subsection{Stage-I: Determining the $ (N,S) $-Limiting Lower Bound by $ N $ and $ S $}
\label{m_m_s_n_queue}
If we assume that there always exist enough FBs in the FB-queue, then the open EV-queue can be separated from the MQN and works as an independent $ M/M/S/N $ queue\footnote{The arrivals of EVs and the swapping time of SSs follow the same distributions as the original open EV-queue.}. In this case, the blocking probability of this $ M/M/S/N $ queue serves as the lower bound for the blocking probability of the MQN, no matter how many CSs and batteries are used in the closed battery-queue.


Since the  $ M/M/S/N $ queue only depends on $ N $ and $ S $ ($ \lambda $ and $ \nu $ are known constants), we thus denote its blocking probability by $ \mathbb{P}_{\mathrm{nslb}}(N,S) $, where the subscript represents the term ``$ (N,S) $-limiting lower bound". It is known that $ \mathbb{P}_{\mathrm{nslb}}(N,S) $ is given by
\begin{equation}
\mathbb{P}_{\mathrm{nslb}}(N,S)= \frac{1}{S^{N-S}S!}\Big(\frac{\lambda}{\nu}\Big)^Np_0,
\label{P_best_N_S}
\end{equation} 
where $ p_0 $ denotes the stationary distribution of having no customer in the $ M/M/S/N $ queue. Note that $ p_0 $ can be calculated as follows:
\begin{align}
p_0 = \left[\sum\limits_{n=0}^S \frac{\lambda^n}{\nu^n n!} + \frac{\lambda^S}{\nu^SS!}\sum\limits_{n=S+1}^N \frac{\lambda^{n-S}}{\nu^{n-S} S^{n-S}}\right]^{-1}.
\label{z_0}
\end{align}

\begin{figure}
	\centering
	\includegraphics[width=8cm] {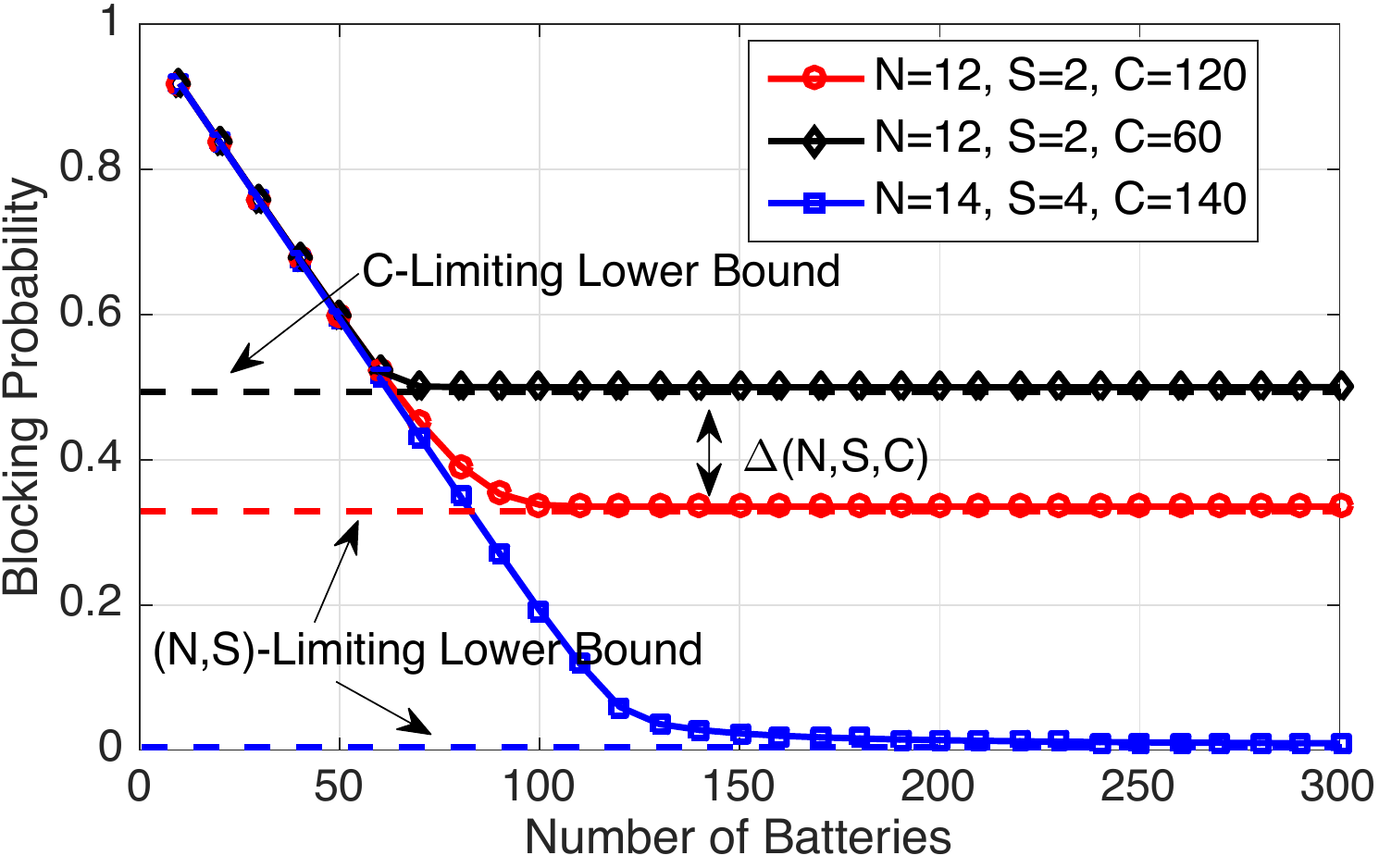}
	\caption{Convergence of the blocking probability  when $ \lambda = 1/30, \nu = 1/90 $, and $ \mu = 1/3600 $.}
	\label{quick_saturation_PoB}
\end{figure}

It is intuitive that the blocking probability of the MQN will be non-increasing when $ C $ and $ B $ increases in Stage-II and Stage-III, respectively. For instance,  as shown by the curve with circles in Fig. \ref{quick_saturation_PoB}, the blocking probability is non-increasing in $ B $ and finally converges to the $(N,S)$-limiting lower bound when $ B $ is sufficiently large, as shown by the middle dashed horizontal curve in Fig. \ref{quick_saturation_PoB}. Note that if $ N $ and $ S $ are not properly designed, then it is possible that even $ \mathbb{P}_{\mathrm{nslb}}(N,S) $ is still too high. For instance, the curve with circles in Fig. \ref{quick_saturation_PoB} depicts that the blocking probability quickly converges to $ \mathbb{P}_{\mathrm{nslb}}(12,2) = 0.336$ when $ N=12, S=2 $, and $ C=120 $, which means that over one third of the swapping requests will be blocked. Note that i) this $(N,S)$-limiting lower bound cannot be reduced in Stage-II and Stage-III by $ C $ and $ B $, and ii) it is difficult to change  $ N $ and $ S $ once they are fixed. Therefore, it is important to design appropriate values of $ N $ and $ S $ in Stage-I to facilitate a good blocking probability performance in Stage-II and Stage-III. For instance, the curve with squares in Fig. \ref{quick_saturation_PoB} illustrates the blocking probability with $ N=14,  S=4 $, and $ C=140 $. It is shown that the $(N,S)$-limiting lower bound of the latter design is $ \mathbb{P}_{\mathrm{nslb}}(14,4) = 0.083 $, which means that the blocking probability can be reduced to a much smaller value as long as $ B $ is sufficiently large.



In Fig. \ref{quick_saturation_PoB}, the curve with diamonds shows the blocking probability when $ N=12, S=2 $, and $ C=60 $. An important observation in this case is that, the $(N,S)$-limiting lower bound $ \mathbb{P}_{\mathrm{nslb}}(12,2)$ is not achievable even if $ B $ is sufficiently large. Instead, the blocking probability converges to another lower bound (i.e., the top horizontal dashed line in Fig. \ref{quick_saturation_PoB}) that is larger than the $(N,S)$-limiting lower bound by a certain gap $ \Delta(N,S,C) $. To avoid confusion between these two lower bounds, we name the new lower bound as the \textit{$ C $-limiting lower bound} (expression will be derived in \eqref{clm_asymptotic_convergence}). To facilitate a lower blocking probability in Stage-II and Stage-III, it is therefore important to understand when the $(N,S)$-limiting lower bound  is achievable/non-achievable with the given values of $ N $ and $ S $. Moreover, if the $(N,S)$-limiting lower bound is non-achievable, then it is also of practical importance to quantify the gap $ \Delta(N,S,C) $ between the two lower bounds.  The following two subsections are particularly motivated to answer these two questions. 

\subsection{Stage-II: Two Different Operating  Modes Dependent on the Value of $ C $}
\label{Achievable_Lower_Bound_of_the_QoS}
Note that in the MQN, the \textit{maximum input rate} of the FB-queue, or alternatively, the \textit{maximum output rate} of the DB-queue, is $ C\mu $. Meanwhile,  $\lambda\big(1-\mathbb{P}_{\mathrm{nslb}}(N,S)\big)$ denotes the \textit{maximum input rate} of the DB-queue, or alternatively, the \textit{maximum output rate} of the FB-queue. Once $ N $ and $ S $ are determined in Stage-I, $ \mathbb{P}_{\mathrm{nslb}}(N,S) $ can then be determined by \eqref{P_best_N_S}. Based on $ \mathbb{P}_{\mathrm{nslb}}(N,S), \lambda $, and $ \mu $, we can define two different operating modes distinguished by the value of $ C $ as follows:
\begin{itemize}[leftmargin=*]
	\item If $ C\leq \lfloor \lambda\big(1-\mathbb{P}_{\mathrm{nslb}}(N,S)\big)/\mu\rfloor $, the maximum output rate of the DB-queue is less than or equal to the maximum input rate of the DB-queue. Therefore, the charging service is the bottleneck of the BSCS, and we thus name this region as the \textit{charging-limiting mode}. 
	
	\item If $ C\geq \lceil \lambda\big(1-\mathbb{P}_{\mathrm{nslb}}(N,S)\big)/\mu\rceil $, then the maximum output rate of the FB-queue is less than or equal to the maximum input rate of the FB-queue. Therefore, the swapping service is the bottleneck of the BSCS, and we thus name this region as the \textit{swapping-limiting mode}. 
\end{itemize}
Note that we use $ \lfloor x \rfloor $ to denote the maximum integer that is no larger than $ x $, and use $ \lceil x \rceil $ to denote the minimum integer that is no less than $ x $. Since $ \lambda\big(1-\mathbb{P}_{\mathrm{nslb}}(N,S)\big)/\mu $ is an analytical threshold for $ C $ that determines the above two operating modes, we thus name this threshold as the \textit{$ C $-limiting threshold}. As  we will show in the next subsection, the $ C $-limiting threshold as well as the above two operating modes will directly influence the asymptotic convergence of the blocking probability in Stage-III.


%



\subsection{Stage-III: Asymptotic Convergence When $ B $ Approaches Infinity}
\label{asymptotic_properties_in_stage_III}
Recall that the number of batteries is flexible to change during the operation of the BSCS, it is therefore of practical importance to quantify the asymptotic behavior of the blocking probability when $ B $ becomes sufficiently large. Specifically, we have the following Theorem \ref{asymptotic_property_PoB} to demonstrate the asymptotic convergence of the blocking probability.

\begin{theorem} 
	Given $ N $ and $ S $ in Stage-I,  the $ C $-limiting threshold $ \lambda\big(1-\mathbb{P}_{\mathrm{nslb}}(N,S)\big)/\mu $  distinguishes the asymptotic performance of the blocking probability into the following two cases:
	\begin{itemize}[leftmargin=*]
		\item If $ C\leq \lfloor \lambda\big(1-\mathbb{P}_{\mathrm{nslb}}(N,S)\big)/\mu\rfloor $, i.e.,  the charging-limiting mode is active, then the blocking probability asymptotically converges to the $ C $-limiting lower bound when $ B $ approaches infinity. Mathematically, we have 
		\begin{align}
		\lim\limits_{B\rightarrow\infty}\mathbb{P}_{\mathrm{MQN}}(N,S,C,B) = 1-C\mu/\lambda \triangleq \mathbb{P}_{\mathrm{clb}}(C),
		\label{clm_asymptotic_convergence}
		\end{align}
		where $ \mathbb{P}_{\mathrm{clb}}(C) $ denotes the $ C $-limiting lower bound.
		\item If $C\geq \lceil \lambda\big(1-\mathbb{P}_{\mathrm{nslb}}(N,S)\big)/\mu\rceil $, i.e.,  the swapping-limiting mode is active, the the blocking probability asymptotically converges to the $(N,S)$-limiting lower bound when $ B $ approaches infinity. Mathematically, we have 
		\begin{align}
		\lim\limits_{B\rightarrow\infty}\mathbb{P}_{\mathrm{MQN}}(N,S,C,B) = \mathbb{P}_{\mathrm{nslb}}(N,S) = \frac{1}{S^{N-S}S!}\Big(\frac{\lambda}{\nu}\Big)^Np_0,
		\label{slm_asymptotic_convergence}
		\end{align}
		where $ p_0 $ is given by \eqref{z_0}. 
	\end{itemize}
	\label{asymptotic_property_PoB}
\end{theorem}
\begin{proof}
	Note that the $ C $-limiting threshold $ \lambda\big(1-\mathbb{P}_{\mathrm{nslb}}(N,S)\big)/\mu $  is purely determined by $ N $ and $ S $ in Stage-I, and works only for $ C $ in Stage-II when $ B $ approaches infinity in Stage-III. Therefore, our proposed multi-stage capacity planning concept not only follows the practice, but also has a clear mathematical interpretation. The proof of this theorem is constructive but requires a lot of space, we thus present the entire proof in the next section, i.e., Section \ref{proof_of_theorem_1}.
\end{proof}

An interesting result established by  \eqref{clm_asymptotic_convergence} is that,  the $C$-limiting lower bound for the blocking probability when $ B $ approaches infinity, i.e., $ \mathbb{P}_{\mathrm{clb}}(C) $,  is purely determined by $ C $ and is independent of $ N $ and $ S $. Therefore, $ N $ and $ S $ cannot \textit{directly} influence the best performance of the MQN once it is operating in the charging-limiting mode. However, we cannot say that $ N $ and $ S $ have no impact on the best performance of the MQN since $ \mathbb{P}_{\mathrm{nslb}}(N,S) $ that defines the $ C $-limiting threshold depends on $ N $ and $ S $. 

Another interesting observation is that, $\mathbb{P}_{\mathrm{clb}}(C)$ is always larger than $ \mathbb{P}_{\mathrm{nslb}}(N,S) $ when the charging-limiting mode is active.  Therefore, there must be a non-zero probability that some EVs not only need to wait for the EVs in front of them, but also need to wait for FBs (i.e., two types of waiting). In contrast, we can see from \eqref{slm_asymptotic_convergence} that the blocking probability of the MQN is equal to that of the $ M/M/S/N $ queue when $ B $ approaches infinity. This is equivalent to saying that when $ C\geq \lceil \lambda\big(1-\mathbb{P}_{\mathrm{nslb}}(N,S)\big)/\mu\rceil $, it is with probability 1 that there exist enough FBs for the swapping-service when $ B $ approaches infinity. Therefore, when $ B $ approaches infinity, it is with probability 1 that there exists only one type of waiting for the EVs in the swapping-limiting mode,  which is more appealing in practice.

Theorem \ref{asymptotic_property_PoB} also shows that the gap between the $C$-limiting lower bound and the $(N,S)$-limiting lower bound (i.e., $ \Delta(N,S,C) $ in Fig. \ref{quick_saturation_PoB})  can be calculated  as
\begin{align}
\Delta(N,S,C) = \max\Big\{1-C\mu/\lambda- \mathbb{P}_{\mathrm{nslb}}(N,S),0\Big\},
\label{uniform_Delta}
\end{align}
where the `$ \max $' operator guarantees that only one of the two terms in the bracket is active. 
Specifically, when $ \Delta(N,S,C) > 0 $, the first term is active, and we say that the $(N,S)$-limiting lower bound is non-achievable. In this case, the charging-limiting mode is active because $ C $ is too small and below the $ C $-limiting threshold, and the $(N,S)$-limiting lower bound cannot be achieved  by any value of $ B $. In contrast,  when $ \Delta(N,S,C) = 0 $, the second term is active, and we say that the $(N,S)$-limiting lower bound is achievable since $ C $ is above the threshold (i.e., swapping-limiting mode is active). In this case, the blocking probability will asymptotically converge to the $(N,S)$-limiting lower bound when the number of batteries approaches infinity.



\begin{remark}
	Our simulation results show that the above asymptotic properties will ``almost" hold as long as $ B $ is ``slightly" larger than $ C $. For instance, the curve with diamonds in Fig. \ref{quick_saturation_PoB} shows that as long as $ B> C=60 $, the blocking probability quickly converges to the $C$-limiting lower bound $\mathbb{P}_{\mathrm{clb}}(C) = 0.5 $. This quick convergence phenomenon demonstrates that the lower bounds derived in Theorem \ref{asymptotic_property_PoB}  is very useful in practice, since $ B $ is not required to be too large.  Meanwhile, this quick convergence phenomenon also demonstrates that our assumption of the infinite buffer sizes for the DB-queue and the FB-queue is amenable.
\end{remark}



\section{Proof of Theorem \ref{asymptotic_property_PoB}}
\label{proof_of_theorem_1}
This section sketches the proof of Theorem \ref{asymptotic_property_PoB}. Our proof consists of four steps. Step 1 is presented in Subsection \ref{two_sub_queueing_networks}, where we define two auxiliary queueing networks based on the original MQN, i.e., the EV-FB queue and the EV-DB queue. Step 2 is further separated into Step 2(a) in Subsection \ref{transition_rate_matrix_EVFB_queue} and Step 2(b) in Subsection \ref{transition_rate_matrix_EVDB_queue}, where we derive the balance equations for the EV-FB queue  and the EV-DB queue, respectively.  Step 3 is presented in Subsection \ref{neccessary_and_sufficient_conditions}, where we prove that the transition rate matrices of the EV-FB queue and the EV-DB queue are ergodic in the charging-limiting mode and the swapping-limiting mode, respectively. Meanwhile, we show that the blocking probabilities of these two queueing networks are respectively $ \mathbb{P}_{\mathrm{clb}}(C)  $ and $ \mathbb{P}_{\mathrm{nslb}}(N,S) $. The last step (i.e., Step 4) is presented in Subsection \ref{final_proof_of_asymptotic_properties}, where we show that the MQN asymptotically converges to the EV-FB queue (the EV-DB queue) if the charging-limiting mode (the swapping-limiting mode) is active. We thus prove the correctness of \eqref{clm_asymptotic_convergence} and \eqref{slm_asymptotic_convergence} in Theorem \ref{asymptotic_property_PoB}.

\subsection{Step 1: Definitions of Two Sub-Queueing Networks}
\label{two_sub_queueing_networks}
To show how the two operating modes can facilitate the demonstration of the asymptotic properties of the MQN, we introduce the following two queueing networks.

\begin{definition}[\textbf{EV-FB Queue}]
	An EV-FB queue is a sub-queueing network of the original MQN that consists of only the open EV-queue and the FB-queue. In the EV-FB queue, the EVs' arrivals and the FBs' arrivals follow a Poisson process with rate $ \lambda $ and $ C\mu $, respectively.  
\end{definition}

\begin{definition}[\textbf{EV-DB Queue}]
	An EV-DB queue is a sub-queueing network of the original MQN that consists of only the open EV-queue  and the DB-queue. In the EV-DB queue, the EVs' arrivals follow a Poisson process with rate $ \lambda $.
\end{definition}

Basically, the EV-FB queue is the remaining part of the MQN after removing the DB-queue, while the EV-DB queue is the remaining part of the MQN after removing the FB-queue. Based on the definitions, the EVs' arrivals of the MQN and these two newly defined queues all follow the same Poisson process with rate $ \lambda $. Meanwhile, we assume that the swapping rate $ \nu $ and the charging rate $ \mu $ are respectively the same among these three queues. Therefore, both the EV-FB queue and the EV-DB queue will be uniquely determined by parameters $ N, S $, and $ C $. In the next two subsections, we will derive the transition rate matrices for the EV-FB queue and the EV-DB queue, respectively.

\begin{figure}
	\centering
	\subfigure[CTMC for the EV-FB queue.]{
		\begin{tikzpicture}[scale=0.6, transform shape, >=latex]
		\path
		(-6,4) node (00)[shape=circle,draw] { 0, 0 }
		(-4,4) node (01)[shape=circle,draw] { 0, 1 }
		(-2,4) node (02)[shape=circle,draw] { 0, 2 }
		( 0,4) node (03)[shape=circle,draw] { 0, 3 }	
		( 2,4) node (04)[shape=circle,draw] { 0, 4 }
		
		(-6,2) node (10)[shape=circle,draw] { 1, 0 }
		(-4,2) node (11)[shape=circle,draw] { 1, 1 }
		(-2,2) node (12)[shape=circle,draw] { 1, 2 }
		( 0,2) node (13)[shape=circle,draw] { 1, 3 }	
		( 2,2) node (14)[shape=circle,draw] { 1, 4 }
		
		(-6,0) node (20)[shape=circle,draw] { 2, 0 }
		(-4,0) node (21)[shape=circle,draw] { 2, 1 }
		(-2,0) node (22)[shape=circle,draw] { 2, 2 }
		( 0,0) node (23)[shape=circle,draw] { 2, 3 }	
		( 2,0) node (24)[shape=circle,draw] { 2, 4 }
		
		(-6,-4) node (n0)[shape=circle,draw] { N,0 }
		(-4,-4) node (n1)[shape=circle,draw] { N,1 }
		(-2,-4) node (n2)[shape=circle,draw] { N,2 }
		( 0,-4) node (n3)[shape=circle,draw] { N,3 }	
		( 2,-4) node (n4)[shape=circle,draw] { N,4 };
		
		\draw (4,4) node[]{{\Huge $ \cdots $}};
		\draw (4,2) node[]{{\Huge $ \cdots $}};
		\draw (4,0) node[]{{\Huge $ \cdots $}};
		\draw (4,-2) node[]{{\Huge $ \cdots $}};
		\draw (4,-4) node[]{{\Huge $ \cdots $}};
		
		\draw (-6,-2) node[]{{\LARGE $ \vdots $}};
		\draw (-4,-2) node[]{{\LARGE $ \vdots $}};
		\draw (-2,-2) node[]{{\LARGE $ \vdots $}};
		\draw ( 0,-2) node[]{{\LARGE $ \vdots $}};
		\draw ( 2,-2) node[]{{\LARGE $ \vdots $}};
		\draw[->] (00) to[bend left] node[below]{$ C\mu$ }(01);
		\draw[->] (01) to[bend left] node[below]{$ C\mu$ }(02);
		\draw[->] (02) to[bend left] node[below]{$ C\mu$ }(03);
		\draw[->] (03) to[bend left] node[below]{$ C\mu$ }(04);
		\draw[->] (04) to[bend left] node[below]{$ C\mu$ }(3.5,4.2);

		\draw[->] (10) to[bend left] node[below]{$ C\mu$ }(11);
		\draw[->] (11) to[bend left] node[below]{$ C\mu$ }(12);
		\draw[->] (12) to[bend left] node[below]{$ C\mu$ }(13);
		\draw[->] (13) to[bend left] node[below]{$ C\mu$ }(14);
		\draw[->] (14) to[bend left] node[below]{$ C\mu$ }(3.5,2.2);
		
		\draw[->] (20) to[bend left] node[below]{$ C\mu$ }(21);
		\draw[->] (21) to[bend left] node[below]{$ C\mu$ }(22);
		\draw[->] (22) to[bend left] node[below]{$ C\mu$ }(23);
		\draw[->] (23) to[bend left] node[below]{$ C\mu$ }(24);
		\draw[->] (24) to[bend left] node[below]{$ C\mu$ }(3.5,0.2);
		
		\draw[->] (n0) to[bend left] node[below]{$ C\mu$ }(n1);
		\draw[->] (n1) to[bend left] node[below]{$ C\mu$ }(n2);
		\draw[->] (n2) to[bend left] node[below]{$ C\mu$ }(n3);
		\draw[->] (n3) to[bend left] node[below]{$ C\mu$ }(n4);
		\draw[->] (n4) to[bend left] node[below]{$ C\mu$ }(3.5,-3.8);
		\draw[->] (00) to[bend right] node[right]{$ \lambda $ }(10);
		\draw[->] (10) to[bend right] node[right]{$ \lambda $ }(20);
		\draw[->] (20) to[bend right] node[right]{$ \lambda $ }(30);
		\draw[->] (-6.2,-2.4) to[bend right] node[right]{$ \lambda $ }(n0);
		
		\draw[->] (01) to[bend right] node[right]{$ \lambda $ }(11);
		\draw[->] (11) to[bend right] node[right]{$ \lambda $ }(21);
		\draw[->] (21) to[bend right] node[right]{$ \lambda $ }(31);
		\draw[->] (-4.2,-2.4) to[bend right] node[right]{$ \lambda $ }(n1);
		
		\draw[->] (02) to[bend right] node[right]{$ \lambda $ }(12);
		\draw[->] (12)  to[bend right] node[right]{$ \lambda $ }(22);
		\draw[->] (22) to[bend right] node[right]{$ \lambda $ }(32);
		\draw[->] (-2.2,-2.4) to[bend right] node[right]{$ \lambda $ }(n2);

		\draw[->] (03) to[bend right] node[right]{$ \lambda $ }(13);
		\draw[->] (13) to[bend right] node[right]{$ \lambda $ }(23);
		\draw[->] (23) to[bend right] node[right]{$ \lambda $ }(33);
		\draw[->] (-0.2,-2.4) to[bend right] node[right]{$ \lambda $ }(n3);

		\draw[->] (04) to[bend right] node[right]{$ \lambda $ }(14);
		\draw[->] (14) to[bend right] node[right]{$ \lambda $ }(24);
		\draw[->] (24) to[bend right] node[right]{$ \lambda $ }(34);
		\draw[->] (1.8,-2.4) to[bend right] node[right]{$ \lambda $ }(n4);

		\draw[->] (11) to node[below]{$ \nu $ }(00);
		\draw[->] (12) to node[below]{$ \nu $ } (01);
		\draw[->] (13) to node[below]{$ \nu $ } (02);
		\draw[->] (14) to node[below]{$ \nu $ } (03);

		\draw[->] (21) to node[below]{$ \nu $ }(10);
		\draw[->] (22) to node[below]{$ 2\nu $ }(11);
		\draw[->] (23) to node[below]{$ 2\nu $ }(12);
		\draw[->] (24) to node[below]{$ 2\nu $ }(13);
		
		\draw[->] (n1) to node[below]{$ \nu $ }(30);
		\draw[->] (n2) to node[below]{$ 2\nu $ }(31);
		\draw[->] (n3) to node[below]{$ 3\nu $ }(32);
		\draw[->] (n4) to node[below]{$ 4\nu $ }(33);
		\end{tikzpicture}
	} 
\qquad 
	\subfigure[CTMC for the EV-DB queue.]{\begin{tikzpicture}[scale=0.6, transform shape,>=latex]
		\path
		(-6,4) node (00)[shape=circle,draw] { 0, 0 }
		(-4,4) node (01)[shape=circle,draw] { 0, 1 }
		(-2,4) node (02)[shape=circle,draw] { 0, 2 }
		( 0,4) node (03)[shape=circle,draw] { 0, 3 }	
		( 2,4) node (04)[shape=circle,draw] { 0, 4 }
		
		(-6,2) node (10)[shape=circle,draw] { 1, 0 }
		(-4,2) node (11)[shape=circle,draw] { 1, 1 }
		(-2,2) node (12)[shape=circle,draw] { 1, 2 }
		( 0,2) node (13)[shape=circle,draw] { 1, 3 }	
		( 2,2) node (14)[shape=circle,draw] { 1, 4 }
		
		(-6,0) node (20)[shape=circle,draw] { 2, 0 }
		(-4,0) node (21)[shape=circle,draw] { 2, 1 }
		(-2,0) node (22)[shape=circle,draw] { 2, 2 }
		( 0,0) node (23)[shape=circle,draw] { 2, 3 }	
		( 2,0) node (24)[shape=circle,draw] { 2, 4 }
		
		(-6,-4) node (n0)[shape=circle,draw] { N, 0 }
		(-4,-4) node (n1)[shape=circle,draw] { N, 1 }
		(-2,-4) node (n2)[shape=circle,draw] { N, 2 }
		( 0,-4) node (n3)[shape=circle,draw] { N, 3 }	
		( 2,-4) node (n4)[shape=circle,draw] { N, 4 };
		
		\draw (4,4) node[]{{\Huge $ \cdots $}};
		\draw (4,2) node[]{{\Huge $ \cdots $}};
		\draw (4,0) node[]{{\Huge $ \cdots $}};
		\draw (4,-2) node[]{{\Huge $ \cdots $}};
		\draw (4,-4) node[]{{\Huge $ \cdots $}};
		
		\draw (-6,-2) node[]{{\LARGE $ \vdots $}};
		\draw (-4,-2) node[]{{\LARGE $ \vdots $}};
		\draw (-2,-2) node[]{{\LARGE $ \vdots $}};
		\draw ( 0,-2) node[]{{\LARGE $ \vdots $}};
		\draw ( 2,-2) node[]{{\LARGE $ \vdots $}};
		\draw[<-] (00) to[bend left] node[below]{$ \mu$ }(01);
		\draw[<-] (01) to[bend left] node[below]{$ 2\mu$ }(02);
		\draw[<-] (02) to[bend left] node[below]{$ 3\mu$ }(03);
		\draw[<-] (03) to[bend left] node[below]{$ 4\mu$ }(04);
		\draw[<-] (04) to[bend left] node[below]{$ 5\mu$ }(3.5,4.2);

		\draw[<-] (10) to[bend left] node[below]{$ \mu$ }(11);
		\draw[<-] (11) to[bend left] node[below]{$ 2\mu$ }(12);
		\draw[<-] (12) to[bend left] node[below]{$ 3\mu$ }(13);
		\draw[<-] (13) to[bend left] node[below]{$ 4\mu$ }(14);
		\draw[<-] (14) to[bend left] node[below]{$ 5\mu$ }(3.5,2.2);
		
		\draw[<-] (20) to[bend left] node[below]{$ \mu$ }(21);
		\draw[<-] (21) to[bend left] node[below]{$ 2\mu$ }(22);
		\draw[<-] (22) to[bend left] node[below]{$ 3\mu$ }(23);
		\draw[<-] (23) to[bend left] node[below]{$ 4\mu$ }(24);
		\draw[<-] (24) to[bend left] node[below]{$ 5\mu$ }(3.5,0.2);
		
		\draw[<-] (n0) to[bend left] node[below]{$ \mu$ }(n1);
		\draw[<-] (n1) to[bend left] node[below]{$ 2\mu$ }(n2);
		\draw[<-] (n2) to[bend left] node[below]{$ 3\mu$ }(n3);
		\draw[<-] (n3) to[bend left] node[below]{$ 4\mu$ }(n4);
		\draw[<-] (n4) to[bend left] node[below]{$ 5\mu$ }(3.5,-3.8);
		\draw[->] (00) to[bend right] node[right]{$ \lambda $ }(10);
		\draw[->] (10) to[bend right] node[right]{$ \lambda $ }(20);
		\draw[->] (20) to[bend right] node[right]{$ \lambda $ }(30);
		\draw[->] (-6.2,-2.4) to[bend right] node[right]{$ \lambda $ }(n0);
		
		\draw[->] (01) to[bend right] node[right]{$ \lambda $ }(11);
		\draw[->] (11) to[bend right] node[right]{$ \lambda $ }(21);
		\draw[->] (21) to[bend right] node[right]{$ \lambda $ }(31);
		\draw[->] (-4.2,-2.4) to[bend right] node[right]{$ \lambda $ }(n1);
		
		\draw[->] (02) to[bend right] node[right]{$ \lambda $ }(12);
		\draw[->] (12)  to[bend right] node[right]{$ \lambda $ }(22);
		\draw[->] (22) to[bend right] node[right]{$ \lambda $ }(32);
		\draw[->] (-2.2,-2.4) to[bend right] node[right]{$ \lambda $ }(n2);

		\draw[->] (03) to[bend right] node[right]{$ \lambda $ }(13);
		\draw[->] (13) to[bend right] node[right]{$ \lambda $ }(23);
		\draw[->] (23) to[bend right] node[right]{$ \lambda $ }(33);
		\draw[->] (-0.2,-2.4) to[bend right] node[right]{$ \lambda $ }(n3);

		\draw[->] (04) to[bend right] node[right]{$ \lambda $ }(14);
		\draw[->] (14) to[bend right] node[right]{$ \lambda $ }(24);
		\draw[->] (24) to[bend right] node[right]{$ \lambda $ }(34);
		\draw[->] (1.8,-2.4) to[bend right] node[right]{$ \lambda $ }(n4);

		\draw[->] (10) to node[below]{$ \nu $ }(01);
		\draw[->] (20) to node[below]{$ 2\nu $ } (11);
		\draw[->] (n0) to node[below]{$ S\nu $ } (31);

		\draw[->] (11) to node[below]{$ \nu $ } (02);
		\draw[->] (21) to node[below]{$ 2\nu $ }(12);
		\draw[->] (n1) to node[below]{$ S\nu $ } (32);
		
		\draw[->] (12) to node[below]{$ \nu $ }(03);
		\draw[->] (22) to node[below]{$ 2\nu $ }(13);
		\draw[->] (n2) to node[below]{$ S\nu $ }(33);
		
		\draw[->] (13) to node[below]{$ \nu $ } (04);
		\draw[->] (23) to node[below]{$ 2\nu $ }(14);
		\draw[->] (n3) to node[below]{$ S\nu $ } (34);
		
		\end{tikzpicture}
	}
	\caption{Illustration of the CTMC for the two sub-queueing networks.  Each node \circled{$ n,b $} with $ n\in\{0, \cdots, N\} $ and $ b\in \{0, \cdots, \infty\} $ in subfigure (a)  denotes a two-dimensional state with $ n $ EVs  and $ b $ FBs in the EV-FB queue; while each node \circled{$ n,j $} with $ n\in\{0, \cdots, N\} $ and $ j\in \{0, \cdots, \infty\} $ in subfigure (b)  denotes a two-dimensional state with $ n $ EVs  and $ j $ DBs in the EV-DB queue.}
	\label{CTMC_equivalent_model_I}
\end{figure}
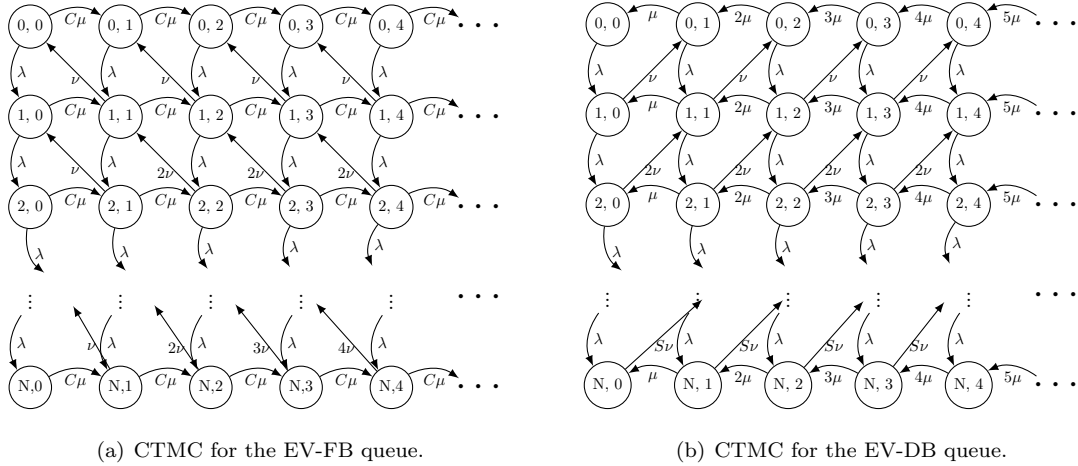

\subsection{Step 2(a): Transition Rate Matrix of the EV-FB Queue}
\label{transition_rate_matrix_EVFB_queue}
We denote the steady-state distribution of the EV-FB queue by $ \{\pi^{(\mathrm{EVFB})}_{n,b}\}_{\forall n,b} $. The two-dimensional CTMC for the EV-FB queue is shown in Fig. \ref{CTMC_equivalent_model_I}(a). Note that different from $ \pi_{n,b} $ of the original MQN, the number of FBs in $ \pi^{(\mathrm{EVFB})}_{n,b} $ can go to infinity. In particular, we have the following balance equations to show the state transitions:
\begin{align}
\label{equivalent_balance_equation_1}
&\pi^{(\mathrm{EVFB})}_{0,0}\Big(\lambda+C\mu\Big) =\pi^{(\mathrm{EVFB})}_{1,1}\nu;\\
&\pi^{(\mathrm{EVFB})}_{N,0}C\mu =\pi^{(\mathrm{EVFB})}_{N-1,0}\lambda;\\
&\pi^{(\mathrm{EVFB})}_{0,b}\Big(\lambda+C\mu\Big) =\pi^{(\mathrm{EVFB})}_{0,b-1}C\mu + \pi^{(\mathrm{EVFB})}_{1,b+1}\nu,\forall b\in \{1,\cdots, \infty\};\\
&\pi^{(\mathrm{EVFB})}_{n,0}\Big(\lambda+C\mu\Big) =\pi^{(\mathrm{EVFB})}_{n-1,0}\lambda + \pi^{(\mathrm{EVFB})}_{n+1,1}\nu,\forall n\in \{1,\cdots,N-1\};\\
&\pi^{(\mathrm{EVFB})}_{N,b}\Big(C\mu+\min\{b,S\}\nu\Big) =\pi^{(\mathrm{EVFB})}_{N-1,b}\lambda +\pi^{(\mathrm{EVFB})}_{N,b-1}C\mu,\forall b\in \{1,\cdots, \infty\};\\\nonumber
&\pi^{(\mathrm{EVFB})}_{n,b}\Big(C\mu+\min\{n,b,S\}\nu+\lambda\Big) =\pi^{(\mathrm{EVFB})}_{n-1,b}\lambda +\pi^{(\mathrm{EVFB})}_{n,b-1}C\mu+\\ & \pi^{(\mathrm{EVFB})}_{n+1,b+1}\nu\min\{n+1,b+1,S\},\forall n\in \{1,\cdots,N-1\}, b\in \{1,\cdots, \infty\}. \label{equivalent_balance_equation_6}
\end{align}

The above  infinite linear systems \eqref{equivalent_balance_equation_1}-\eqref{equivalent_balance_equation_6} can be organized into a matrix form as follows:
\begin{equation}
\boldsymbol{\pi}^{(\mathrm{EVFB})}\mathbf{Q}^{(\mathrm{EVFB})}=\mathbf{0}, \mathrm{\ and\ } \boldsymbol{\pi}^{(\mathrm{EVFB})}\mathbf{e}=1,
\label{linear_system_clm}
\end{equation}
where $ \boldsymbol{\pi}^{(\mathrm{EVFB})}=[\boldsymbol{\pi}_0^{(\mathrm{EVFB})},\boldsymbol{\pi}_1^{(\mathrm{EVFB})},\cdots, \boldsymbol{\pi}_b^{(\mathrm{EVFB})},\cdots] $ with $ \boldsymbol{\pi}_b^{(\mathrm{EVFB})} =[\pi_{0,b}^{(\mathrm{EVFB})}$, $\pi_{1,b}^{(\mathrm{EVFB})},\cdots,\pi_{N,b}^{(\mathrm{EVFB})}] $, $\forall b \in\{0,1,\cdots,\infty\}$, and $ \mathbf{Q}^{(\mathrm{EVFB})} $ is the transition rate matrix given as follows:
\begin{align}
\mathbf{Q}^{(\mathrm{EVFB})}= \left[
\begin{array}{llllllllllll}
\mathbf{L}_{\mathrm{00}}   & \mathbf{F}_{\mathrm{01}} & \mathbf{0}&\cdots\\
\mathbf{D}_{\mathrm{10}}   & \mathbf{L} & \mathbf{F}& \mathbf{0}&\cdots\\
\mathbf{0}& \mathbf{D} & \mathbf{L}  & \mathbf{F}& \mathbf{0}&\cdots\\
\cdots&   \mathbf{0}         & \mathbf{D}  & \mathbf{L} & \mathbf{F}& \mathbf{0}\\
&            &   \ddots          & \ddots     & \ddots     & \ddots
\end{array}
\right].
\label{Q_EVFB}
\end{align}
It can be observed that $ \mathbf{Q}^{(\mathrm{EVFB})} $ is the same as the previous finite matrix $ \mathbf{Q} $ in the upper-left-corner part (i.e., matrices $ \mathbf{L}_{00}, \mathbf{F}_{00} $ and $ \mathbf{D}_{10} $) and the repetitive part (i.e., matrices $ \mathbf{L}, \mathbf{F} $, and $ \mathbf{D} $  defined in Sec. \ref{steady_state_distribution_exponential}), but different from $ \mathbf{Q} $  in the right-bottom-corner part.  In fact, when $ B $ approaches infinity, the dashed part of the CTMC in Fig. \ref{CTMC} after removing the dotted part is exactly the CTMC illustrated in Fig. \ref{CTMC_equivalent_model_I}(a). This is equivalent to saying that $ \mathbf{Q}^{(\mathrm{EVFB})} $ is the remaining part of $ \mathbf{Q} $ after removing the right-bottom-corner blocks (i.e., matrices $ \mathbf{L}^{\mathrm{N}}_{00}, \mathbf{F}^{\mathrm{N}}_{00} $, and $ \mathbf{D}^{\mathrm{N}}_{10} $).

\subsection{Step 2(b): Transition Rate Matrix for the EV-DB Queue}
\label{transition_rate_matrix_EVDB_queue}
We denote the steady-state distribution of the EV-FB queue by $ \{\pi^{(\mathrm{EVDB})}_{n,j}\}_{\forall n,j} $. Here, we use $ (n,j) $ to denote the state of having $ n $ EVs and $ j $ DBs in the EV-DB queue, and the steady-state transitions of the EV-DB queue can be given as follows:
\begin{align}
&\pi^{(\mathrm{EVDB})}_{0,0} \lambda  =\pi^{(\mathrm{EVDB})}_{0,1}\mu;\\\nonumber
&\pi^{(\mathrm{EVDB})}_{0,j}\Big(\lambda+\min\{j,C\}\mu\Big) =\pi^{(\mathrm{EVDB})}_{0,j+1}\mu\min\{j+1,C\} +\\
& \pi^{(\mathrm{EVDB})}_{1,j-1}\nu\min\{j-1,S\},\forall j\in \{1,2,\cdots, \infty\};\\
&\pi^{(\mathrm{EVDB})}_{n,0}\Big(\lambda+\nu\min\{n,S\}\Big) = \pi^{(\mathrm{EVDB})}_{n-1,0}\lambda + \pi^{(\mathrm{EVDB})}_{n,1}\mu,\forall n\in \{1,\cdots,N-1\};\\\nonumber
&\pi^{(\mathrm{EVDB})}_{N,j}\Big(\nu\min\{N,S\}+\mu\min\{j,C\}\Big) =\pi^{(\mathrm{EVDB})}_{N-1,j}\lambda +\\ &\pi^{(\mathrm{EVDB})}_{N,j+1}\mu\min\{j+1,C\},\forall j\in \{0,1,2,\cdots, \infty\};\\\nonumber
&\pi^{(\mathrm{EVDB})}_{n,j}\Big(\mu\min\{j,C\}+\nu\min\{n,S\}+\lambda\Big) =\pi^{(\mathrm{EVDB})}_{n-1,b}\lambda + \pi^{(\mathrm{EVDB})}_{n,j+1}\mu\min\{j+1,C\}+\\
&\pi^{(\mathrm{EVDB})}_{n+1,j-1}\nu\min\{n+1,S\},\forall n\in \{1,\cdots,N-1\}, j\in \{1,2,\cdots, \infty\}.
\end{align}
The steady-state distribution can be obtained by solving the following infinite linear systems:
\begin{equation}
\boldsymbol{\pi}^{(\mathrm{EVDB})}\mathbf{Q}^{(\mathrm{EVDB})}=\mathbf{0}, \mathrm{\ and\ } \boldsymbol{\pi}^{(\mathrm{EVDB})}\mathbf{e}=1,
\end{equation}
where $ \boldsymbol{\pi}^{(\mathrm{EVDB})}=[\boldsymbol{\pi}_0^{(\mathrm{EVDB})},\boldsymbol{\pi}_1^{(\mathrm{EVDB})},\cdots, \boldsymbol{\pi}_j^{(\mathrm{EVDB})},\cdots] $ with $ \boldsymbol{\pi}_j^{(\mathrm{EVDB})}=[\pi_{0,j}^{(\mathrm{EVDB})}$, $\pi_{1,j}^{(\mathrm{EVDB})},\cdots,\pi_{N,j}^{(\mathrm{EVDB})}] $, $\forall j \in\{0,1,\cdots,\infty\}$, and $ \mathbf{Q}^{(\mathrm{EVDB})} $ is the transition rate matrix given by
\begin{align}
\mathbf{Q}^{(\mathrm{EVDB})}= \left[
\begin{array}{llllllllllll}
\mathbf{L}^{\mathrm{N}}_{\mathrm{00}} & \mathbf{D}^{\mathrm{N}}_{\mathrm{01}}& \mathbf{0}&\cdots\\
\mathbf{F}^{\mathrm{N}}_{\mathrm{10}} & \mathbf{L} & \mathbf{D}& \mathbf{0}&\cdots\\
\mathbf{0}& \mathbf{F} & \mathbf{L} & \mathbf{D}& \mathbf{0}&\cdots\\
\cdots&       \mathbf{0}     & \mathbf{F} & \mathbf{L} & \mathbf{D}& \mathbf{0}\\
&            &         \ddots    & \ddots     & \ddots     & \ddots
\end{array}
\right].
\end{align}

We can observe that the repetitive part of $ \mathbf{Q}^{(\mathrm{EVDB})} $ is constructed by the same block matrices defined in Sec. \ref{steady_state_distribution_exponential} (i.e., $ \mathbf{D}, \mathbf{L} $, and $ \mathbf{F} $), and is equal to the transpose of that of matrix $ \mathbf{Q} $. Meanwhile, the upper-left-corner part (i.e., matrices $ \mathbf{L}^{\mathrm{N}}_{00}, \mathbf{F}^{\mathrm{N}}_{00} $, and $ \mathbf{D}^{\mathrm{N}}_{10} $) of $ \mathbf{Q}^{(\mathrm{EVFB})} $ is also equal to the transpose of the previous finite matrix $ \mathbf{Q} $. In fact, when $ B $ approaches infinity, the remaining dotted part of the CTMC illustrated in Fig. \ref{CTMC} after removing the dashed part is exactly the CTMC in Fig. \ref{CTMC_equivalent_model_I}(b). This is equivalent to saying that $ \mathbf{Q}^{(\mathrm{EVFB})} $ is the remaining part of $ \mathbf{Q} $ after removing the upper-left-corner blocks (i.e., matrices $ \mathbf{L}_{00}, \mathbf{F}_{00} $, and $ \mathbf{D}_{10} $).

\subsection{Step 3: Ergodicity of the Two Sub-Queueing Networks} 
\label{neccessary_and_sufficient_conditions}
A key question related to the transition rate matrix of a CTMC is whether there exists a unique stationary distribution or not, or equivalently, whether the Markov chain is ergodic (i.e., irreducible and positive recurrent). This subsection shows that the above two transition rate matrices are indeed irreducible and positive recurrent under the charging-limiting mode and the swapping-limiting mode, respectively. Below, we first present a lemma that shows the irreducibility of $ \mathbf{Q}^{(\mathrm{EVFB})} $ and $ \mathbf{Q}^{(\mathrm{EVDB})} $. 
\begin{lemma}
	The Markov chains defined by $ \mathbf{Q}^{(\mathrm{EVFB})} $ and $ \mathbf{Q}^{(\mathrm{EVDB})} $ are irreducible.
	\label{irreducibility}
\end{lemma}
\begin{proof}
	Please refer to \ref{appendix_proof_irreducibility}.
\end{proof}


Furthermore, we have the following important lemma to show the necessary and sufficient conditions for the irreducible Markov chains $ \mathbf{Q}^{(\mathrm{EVFB})} $ and $ \mathbf{Q}^{(\mathrm{EVDB})} $ to be positive recurrent.
\begin{lemma}
	The irreducible Markov chain $ \mathbf{Q}^{(\mathrm{EVFB})} $ is positive recurrent if and only if $ C\leq \lfloor \lambda\big(1-\mathbb{P}_{\mathrm{nslb}}(N,S)\big)/\mu\rfloor $, and the irreducible Markov chain $ \mathbf{Q}^{(\mathrm{EVDB})} $ is positive recurrent if and only if $ C\geq \lceil \lambda\big(1-\mathbb{P}_{\mathrm{nslb}}(N,S)\big)/\mu\rceil $.
	\label{positive_recurrent}
\end{lemma}

\begin{proof}
	Please refer to \ref{appendix_proof_positive_recurrent}.
\end{proof}

Lemma \ref{irreducibility} and Lemma \ref{positive_recurrent} guarantee that as long as the MQN falls within the charging-limiting mode (the swapping-limiting mode), the transition rate matrix corresponding to the EV-FB queue (EV-DB queue) is irreucibile and positive recurrent, and thus ergodic. Therefore, there exists a unique steady-state distribution for the EV-FB queue (EV-DB queue) if and only if the BSCS is working in the charging-limiting mode (swapping-limiting mode). As a result, there also exists a steady-state blocking probability for the EV-FB queue (the EV-DB queue) in the charging-limiting mode (swapping-limiting mode). To aid our proof in Step 4, we denote these two types of blocking probabilities as $ \mathbb{P}_{\mathrm{EVFB}}(N,S,C)$ and  $\mathbb{P}_{\mathrm{EVDB}}(N,S,C)$, respectively. 

Note that the EV-FB queue can be considered as a rate-control throttle \cite{atm_network} \cite{rate_control_throttle}, where the EVs' arrivals are controlled by an infinite token bank, i.e., the buffer of the FB-queue. For the EV-FB queue being ergodic, equation $ \lambda(1-\mathbb{P}_{\mathrm{EVFB}}(N,S,C))=C\mu $ must hold in the steady-state. Thus, 
\begin{align}
\mathbb{P}_{\mathrm{EVFB}}(N,S,C) = 1-C\mu/\lambda = \mathbb{P}_{\mathrm{clb}}(C).
\label{PoB_EVFB}
\end{align}
Meanwhile, the EV-DB queue is a two-stage tandem queueing network, and the open EV-queue will not be affected by the DB-queue. Thus, the blocking probability of the EV-DB queue is equal to $\mathbb{P}_{\mathrm{nslb}}(N,S)$, i.e., 
\begin{align}
\mathbb{P}_{\mathrm{EVDB}}(N,S,C)=\mathbb{P}_{\mathrm{nslb}}(N,S).
\label{PoB_EVDB}
\end{align}

Interestingly, the right-hand-side of \eqref{PoB_EVFB} and \eqref{PoB_EVDB} are respectively the right-hand-side of \eqref{clm_asymptotic_convergence} and \eqref{slm_asymptotic_convergence}. Therefore, to prove Theorem \ref{asymptotic_property_PoB}, it suffices to prove  that i) if the charging-limiting mode is active, the open EV-queue and the FB-queue of the MQN asymptotically converge to the EV-FB queue when $ B $ approaches infinity;  i) if the swapping-limiting mode is active, then the open EV-queue and the DB-queue of the MQN asymptotically converge to the EV-DB queue when $ B $ approaches infinity. The following subsection shows the detailed proof of these two types of convergence.

\subsection{Step 4: Asymptotic Convergence of the Blocking Probability}
\label{final_proof_of_asymptotic_properties}
To aid our proof, let us define the steady-state probability that there always exist enough FBs for the swapping service as the \textit{probability of having enough FBs} ( denoted as $ \mathbb{P}_{\mathrm{enough}} $). Meanwhile,  we also define the steady-state probability that all CSs are busy as the \textit{probability of all CSs busy} (denoted as $ \mathbb{P}_{\mathrm{busy}} $). Intuitively, both of these two probabilities can be calculated as
\begin{align}
\mathbb{P}_{\mathrm{enough}} = \sum_{n=0}^{N}\sum_{b=\min\{n,S\}}^{B}\pi_{n,b},  \mathbb{P}_{\mathrm{busy}} = \sum_{n=0}^{N}\sum_{b=0}^{B-C}\pi_{n,b}.
\label{definition_enough_busy}
\end{align}

Based on the above definition, we have the following lemma which shows the convergence of these two probabilities when $ B $ approaches infinity.

\begin{lemma}
	Given $ N $ and $ S $ in Stage-I, and $ C $ in Stage-II, if the charging-limiting mode is active, then we have 
	\begin{align}
	\lim\limits_{B\rightarrow\infty}\mathbb{P}_{\mathrm{enough}} =  
	\phi(N,S,C),\lim\limits_{B\rightarrow\infty}\mathbb{P}_{\mathrm{busy}} = 1,
	\label{probability_of_enough_and_busy_clm}
	\end{align}
	where $ \phi(N,S,C)\in (0,1) $ is determined by $ N, S$, and $ C\leq \lfloor \lambda\big(1-\mathbb{P}_{\mathrm{nslb}}(N,S)\big)/\mu\rfloor $. Otherwise, if the swapping-limiting mode is active, then we have
	\begin{align} \lim\limits_{B\rightarrow\infty}\mathbb{P}_{\mathrm{enough}} = 1,\lim\limits_{B\rightarrow\infty}\mathbb{P}_{\mathrm{busy}} =  
	\psi(N,S,C),
	\label{probability_of_enough_and_busy_slm}
	\end{align}
	where $ \psi(N,S,C)\in (0,1) $ is determined by $ N, S$, and $ C\geq \lceil \lambda\big(1-\mathbb{P}_{\mathrm{nslb}}(N,S)\big)/\mu\rceil $.
	\label{probability_of_enough_and_busy}
\end{lemma}
\begin{proof}
	Please refer to \ref{appendix_proof_probability_of_having_FB_shortage} for the proof and computation of $ \phi(N,S,C)$ and $ \psi(N,S,C) $. 
\end{proof}

Basically, Lemma \ref{probability_of_enough_and_busy} shows that in the charging-limiting mode,  no matter how many batteries are used in the closed battery-queue, there always exists a non-zero probability $ 1-\phi(N,S,C) $ that the number of FBs is not enough to serve all the queued EVs. Meanwhile, all the CSs will become busy with probability 1 when $ B $ approaches infinity, which means that the output rate of the DB-queue or the input rate of the FB-queue will approach $ C\mu $. Therefore, the open EV-queue and the FB-queue of the original MQN will asymptotically converge to the EV-FB queue when $ B $ approaches infinity. Based on \eqref{PoB_EVFB}, we have
\begin{align}
\lim\limits_{B\rightarrow\infty}\mathbb{P}_{\mathrm{MQN}}(N,S,C,B)=\mathbb{P}_{\mathrm{EVFB}}(N,S,C) = 1-C\mu/\lambda=\mathbb{P}_{\mathrm{clm}}(C).
\end{align}
We thus complete the proof of \eqref{clm_asymptotic_convergence} in Theorem \ref{asymptotic_property_PoB}. 

Similarly, Lemma \ref{probability_of_enough_and_busy} also shows that if the swapping-limiting mode is active, then the probability of having enough FBs approaches 1 when $ B $ approaches infinity. Therefore, the open EV-queue will asymptotically work as an independent $ M/M/S/N $ queueing system when $ B $ approaches infinity. As a result, the open EV-queue and the DB-queue of the original MQN will asymptotically converge to the EV-DB queue. Based on \eqref{PoB_EVDB}, we have  
\begin{align}
\lim\limits_{B\rightarrow\infty}\mathbb{P}_{\mathrm{MQN}}(N,S,C,B)=\mathbb{P}_{\mathrm{EVDB}}(N,S,C)=\mathbb{P}_{\mathrm{nslb}}(N,S).
\end{align}
We thus complete the proof of \eqref{slm_asymptotic_convergence} in Theorem \ref{asymptotic_property_PoB}.

\section{Numerical Evaluations}
\label{sec_simulation_results}
In this section, we validate our theoretical analysis with extensive simulation. We will also particularly focus on illustrating the lower bound of the blocking probability with different average arrival rates $ \lambda $, average swapping rates $ \nu $, and average charging rates $ \mu $. 
Based on the specific simulation, the EVs' arrival rate $ \lambda  $ varies between $ 0.01 $ and $ 0.06 $, which corresponds to the average number of EV arrivals  being between 0.6 (light traffic) and 3.6 (heavy traffic) within a 1-minute duration.
The average swapping time, i.e., $ 1/\nu $, is assumed to be between 100 seconds and 500 seconds \cite{tesla_battery_swapping, bss_qingdao}.
We also assume that the average charging time $ \frac{1}{\mu} $ is between 1 hour and 4 hours, which follows the state-of-the-art charging technology \cite{tesla_supercharging_stations}.  

\begin{figure}
	\centering
	\subfigure[The Charging-Limiting Mode]{\includegraphics[width=8cm] {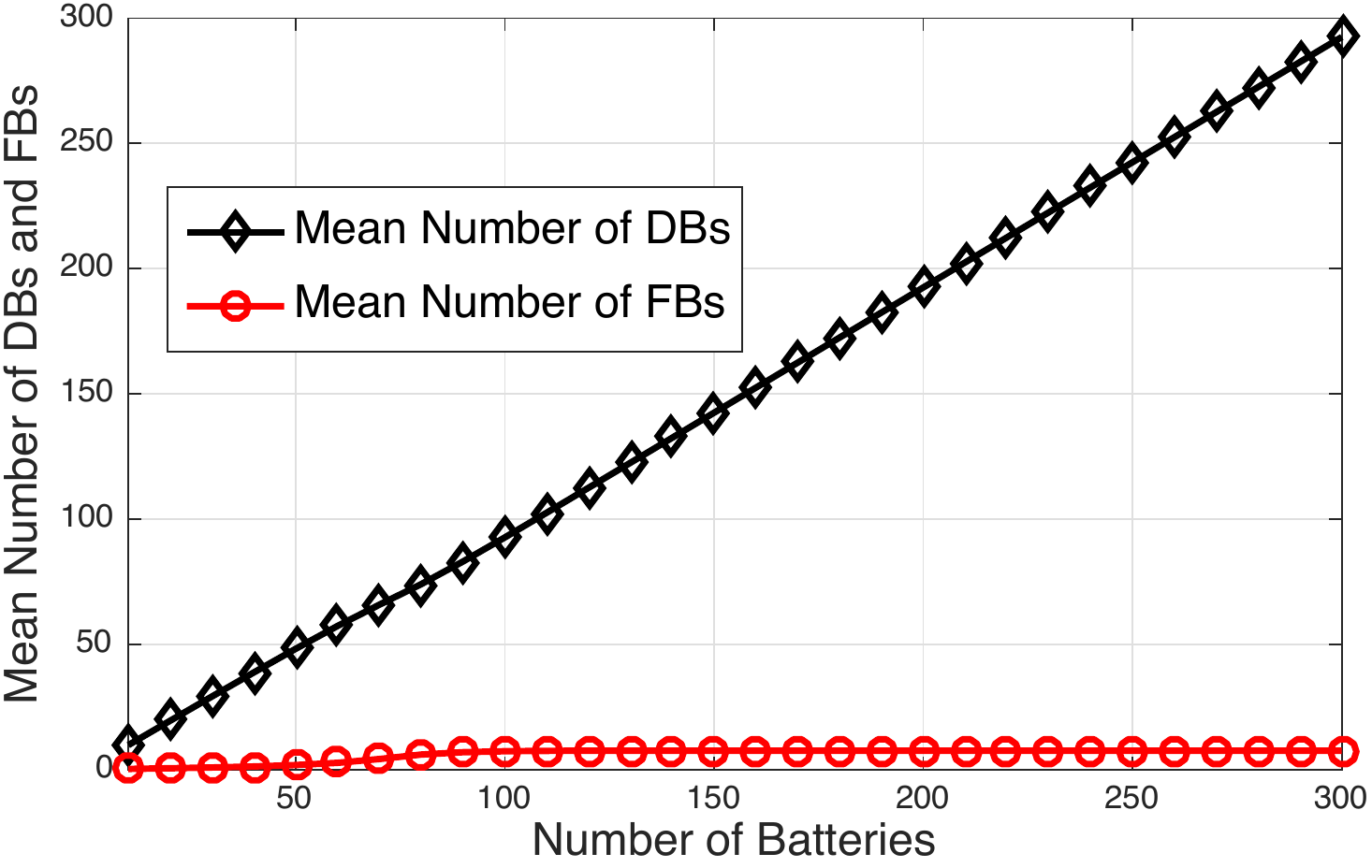}}
	\subfigure[The Swapping-Limiting Mode]{\includegraphics[width=8cm] {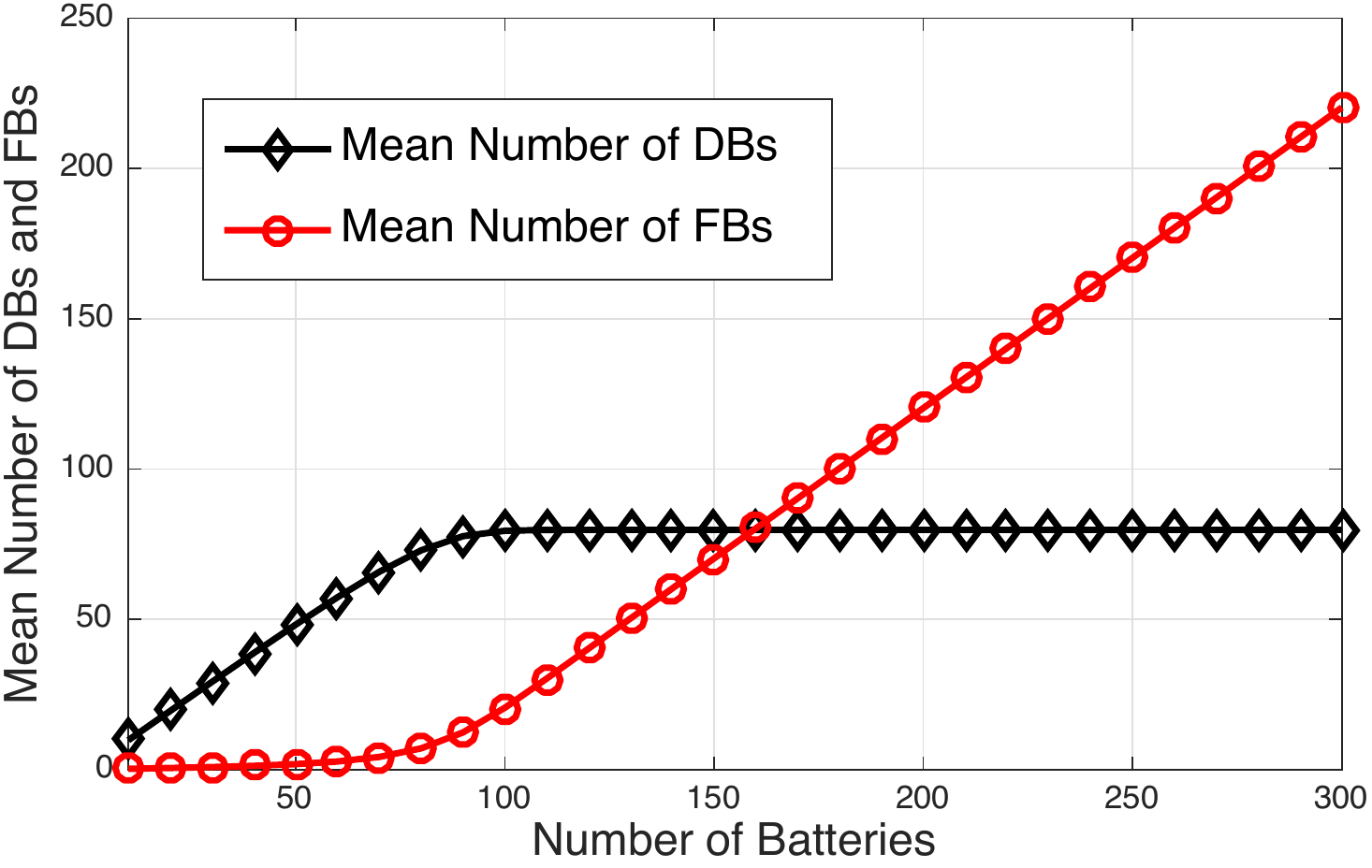}}
	\caption{Illustrations of the mean number of DBs and FBs. We set $ C=70 $ and $ C=120 $ for subfigure (a) and (b), respectively. Other parameters are chosen as follows: $N=12, S=2,  \lambda = 1/30, \nu=1/90 $, and $\mu=1/3600 $. }
	\label{number_of_DBs_FBs}
\end{figure}

\subsection{Validation of Lemma \ref{irreducibility} and Lemma \ref{positive_recurrent}}
The ergodicity of the two sub-queueing networks can be illustrated by their mean queue length. As can be seen from Fig. \ref{number_of_DBs_FBs}(a), in the charging-limiting mode,  the average number of FBs converges to a constant even when $ B $ becomes sufficiently large. This follows Lemma \ref{irreducibility} and Lemma \ref{positive_recurrent} that the EV-FB queue is ergodic in the charging-limiting mode, and there exists a unique steady-state average number of FBs in the EV-FB queue. Therefore, in this operating mode, further increasing $ B $ does not help reduce the blocking probability, since all the newly added batteries will be backlogged as DBs in the DB-queue, as shown by the curve with diamonds in Fig. \ref{number_of_DBs_FBs}(a). In comparison, as depicted by Fig. \ref{number_of_DBs_FBs}(b), the mean number of DBs converges to a constant in the swapping-limiting mode, but the average number of FBs keeps increasing when $ B $ increases. This validates Lemma \ref{positive_recurrent} that the FB-queue is not ergodic (i.e., unstable) when the swapping-limiting mode is active.

\begin{figure}
	\centering
	\subfigure[Prob. of having enough FBs]{\includegraphics[width=8cm] {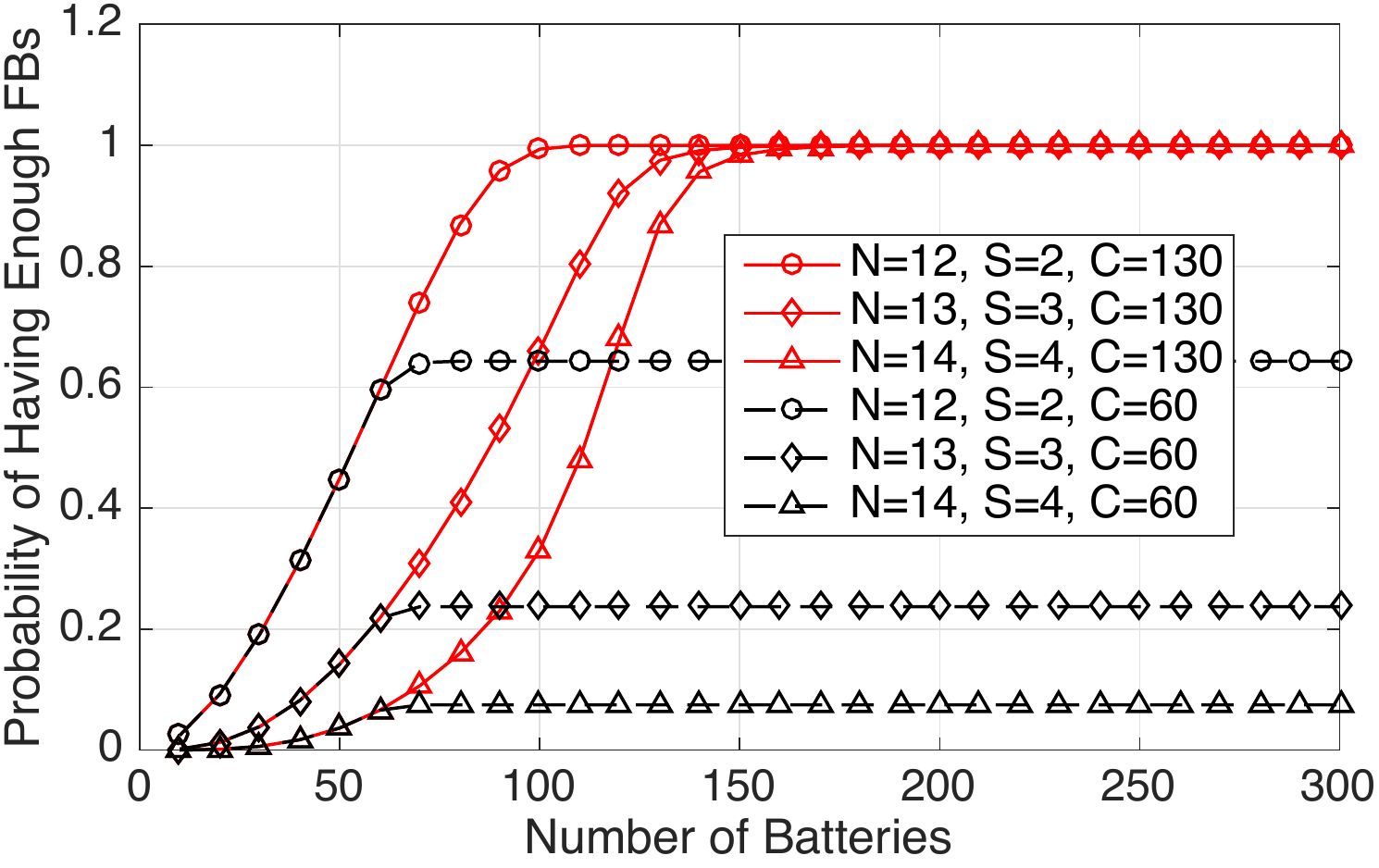}}
	\subfigure[Prob. of all CSs busy]{\includegraphics[width=8cm] {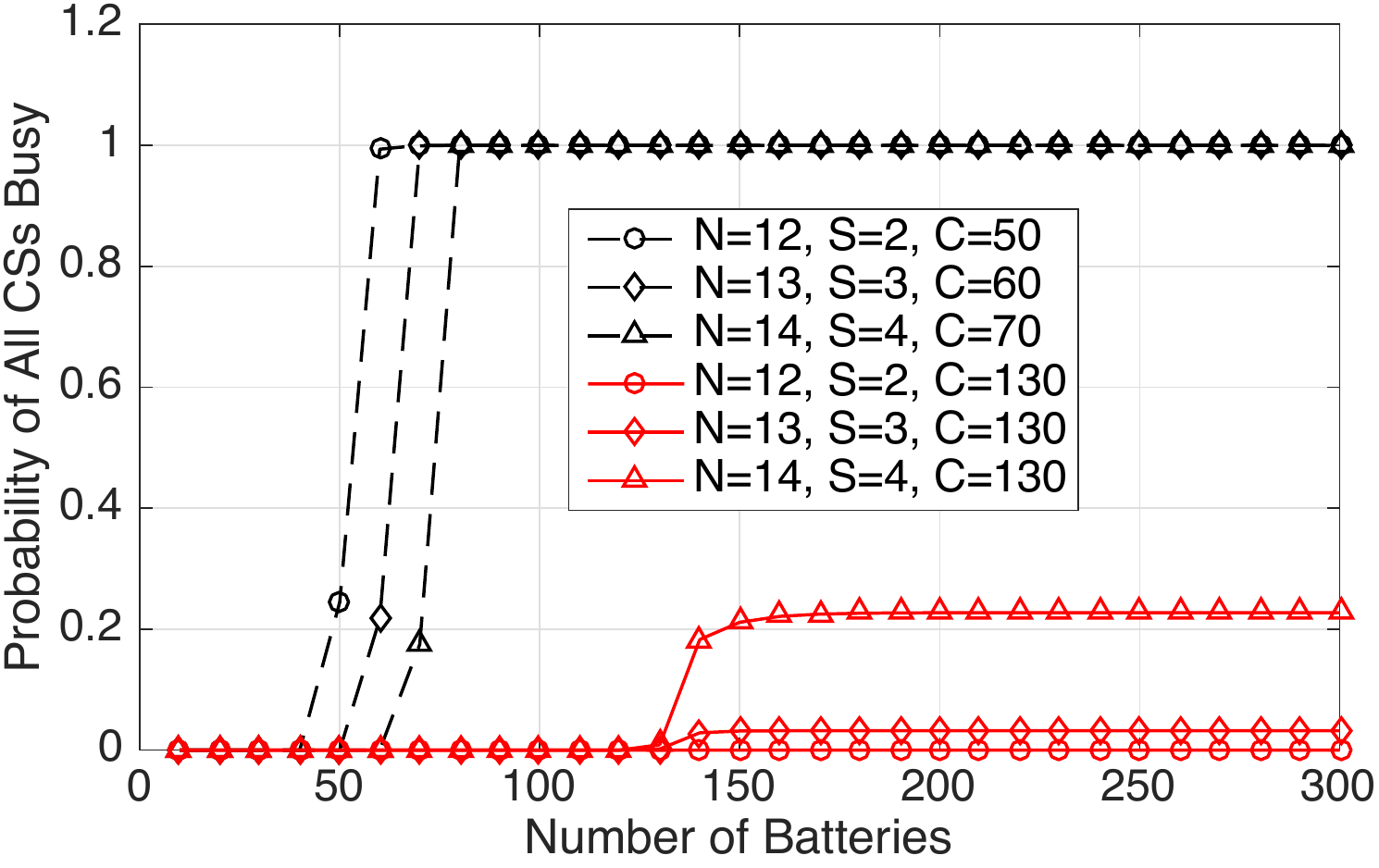}}
	\caption{Convergence of the probability of having enough FB in subfigure (a), and convergence of the probability of all CSs busy in subfigure (b). For both subfigures, solid curves represent the swapping-limiting mode and dashed curves represent the charging-limiting mode. Other parameters are chosen as follows: $\lambda = 1/30, \nu=1/90 $, and $\mu=1/3600 $.}
	\label{probability_of_having_enough_FBs}
\end{figure}

\begin{figure}
	\centering
	\subfigure[$ \phi(N,S,C) $ in the charging-limiting mode.]{\includegraphics[width=8cm] {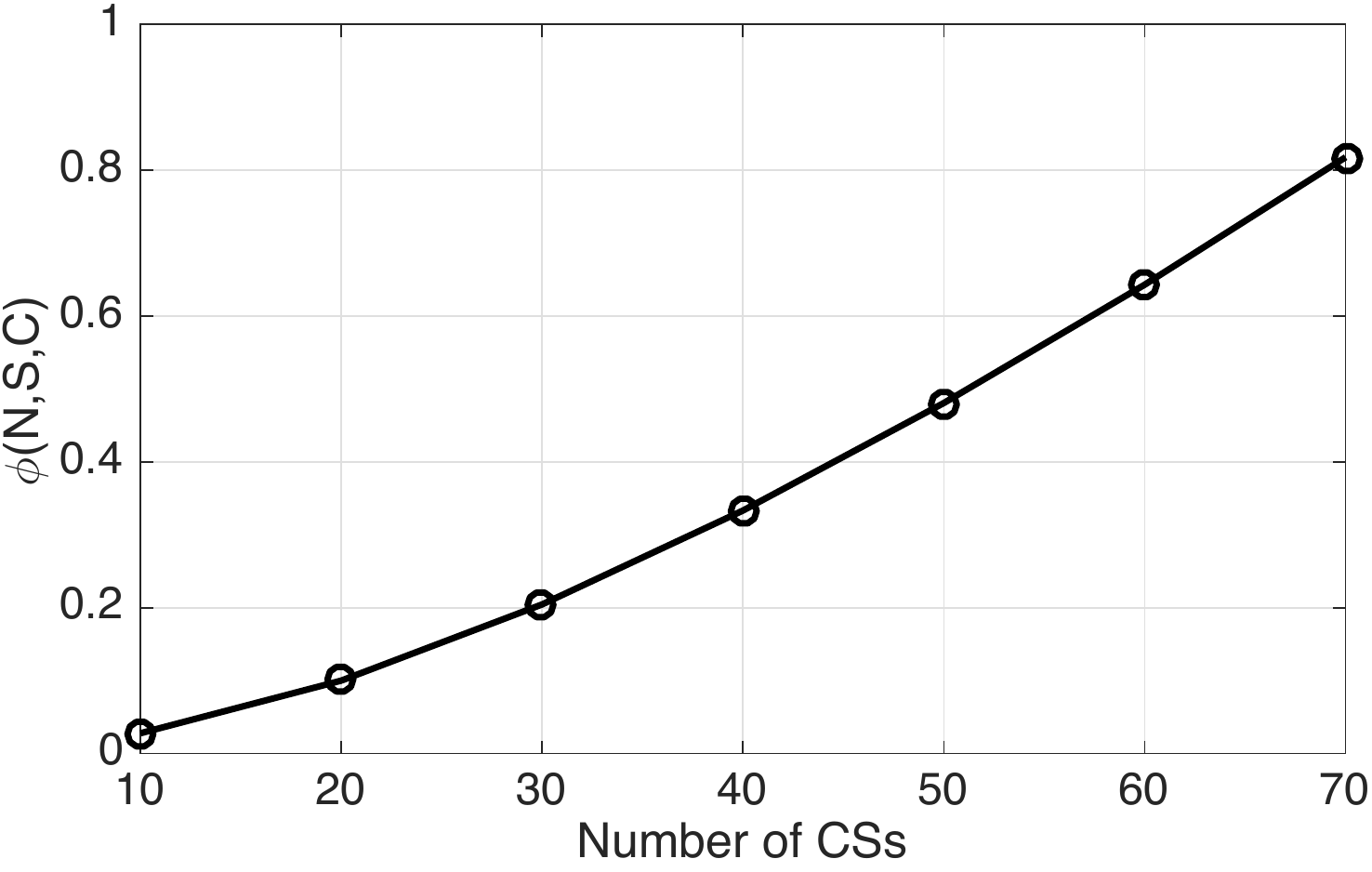}}
	\subfigure[$ \psi(N,S,C) $ in the swapping-limiting mode.]{\includegraphics[width=8cm] {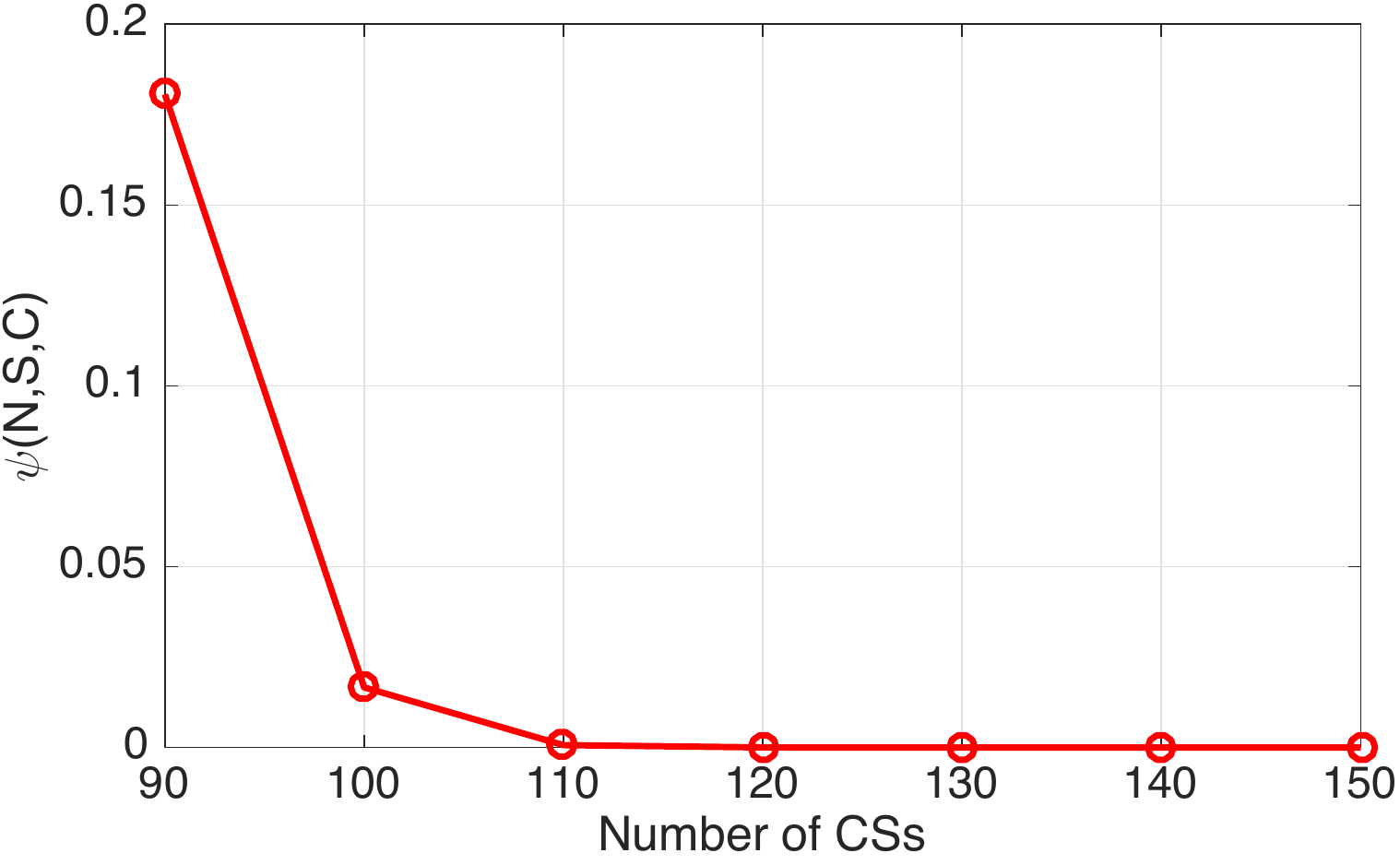}}
	\caption{Illustration of $ \phi(N,S,C) $ and $ \psi(N,S,C) $ with different numbers of CSs. For both figures, we choose  $N=12, S=2,  \lambda = 1/30, \nu=1/90 $, and $\mu=1/3600 $, thus $ \lambda\big(1-\mathbb{P}_{\mathrm{nslb}}(N,S)\big)/\mu = 79.74 $.  Therefore, $ 1\leq C\leq 79 $ represents the charging-limiting mode and $  C \geq 80$ represents the swapping-limiting mode.}
	\label{phi_N_S_C}
\end{figure}

\subsection{Validation of Lemma \ref{probability_of_enough_and_busy}}
As shown by the bottom three dashed curves (which denote the result when the charging-limiting mode is active) in Fig. \ref{probability_of_having_enough_FBs}(a), no matter how many batteries are used in the closed battery-queue, there always exists a non-zero probability  that the number of FBs is not enough to serve all the queued EVs.  Meanwhile, as shown by the top three dashed curves in Fig. \ref{probability_of_having_enough_FBs}(b), all the CSs will become busy with probability 1 when $ B $ approaches infinity. However, in the swapping-limiting mode, as depicted by the top three curves in Fig. \ref{probability_of_having_enough_FBs}(a), the probability of having enough FBs (i.e., $ \mathbb{P}_{\mathrm{enough}} $) always converges to 1, which verifies the result in \eqref{slm_asymptotic_convergence}. Therefore, the open EV-queue in the MQN will asymptotically work as an independent $ M/M/S/N $ queue as long as i) $ B $ is sufficiently large, and ii) the MQN is working in the swapping-limiting mode.

We also illustrate $ \phi(N,S,C) $ and $ \psi(N,S,C) $ with different numbers of CSs in their corresponding operating modes. As shown in Fig. \ref{phi_N_S_C}(a), $ \phi(N,S,C) $ is strictly increasing in $ C $, this follows our intuition  that more CSs will increase the probability of having enough FBs. Figure \ref{phi_N_S_C}(b) shows that $ \psi(N,S,C) $ is strictly decreasing in $ C $, which indicates that using less CSs in the closed battery-queue will make the CSs more busy. By combining these two figures together, it can be observed that when $ C $ is increasing between $ [10, 150] $,  the MQN first works in the charging-limiting mode and $ \phi(N,S,C) $ gradually increases and  approaches 1. When $ C $ becomes larger than or equal to $ \lceil \lambda\big(1-\mathbb{P}_{\mathrm{nslb}}(N,S)\big)/\mu\rceil = 80 $, the MQN switches to the swapping-limiting mode and $ \psi(N,S,C) $ sharply decreases and approaches 0. Therefore, switching between the two operating modes can trigger different asymptotic probabilistic performances for $ \mathbb{P}_{\mathrm{enough}} $ and $ \mathbb{P}_{\mathrm{busy}} $.

\begin{figure}
	\centering
	\subfigure[$ C=30 $]{\includegraphics[width=10cm]{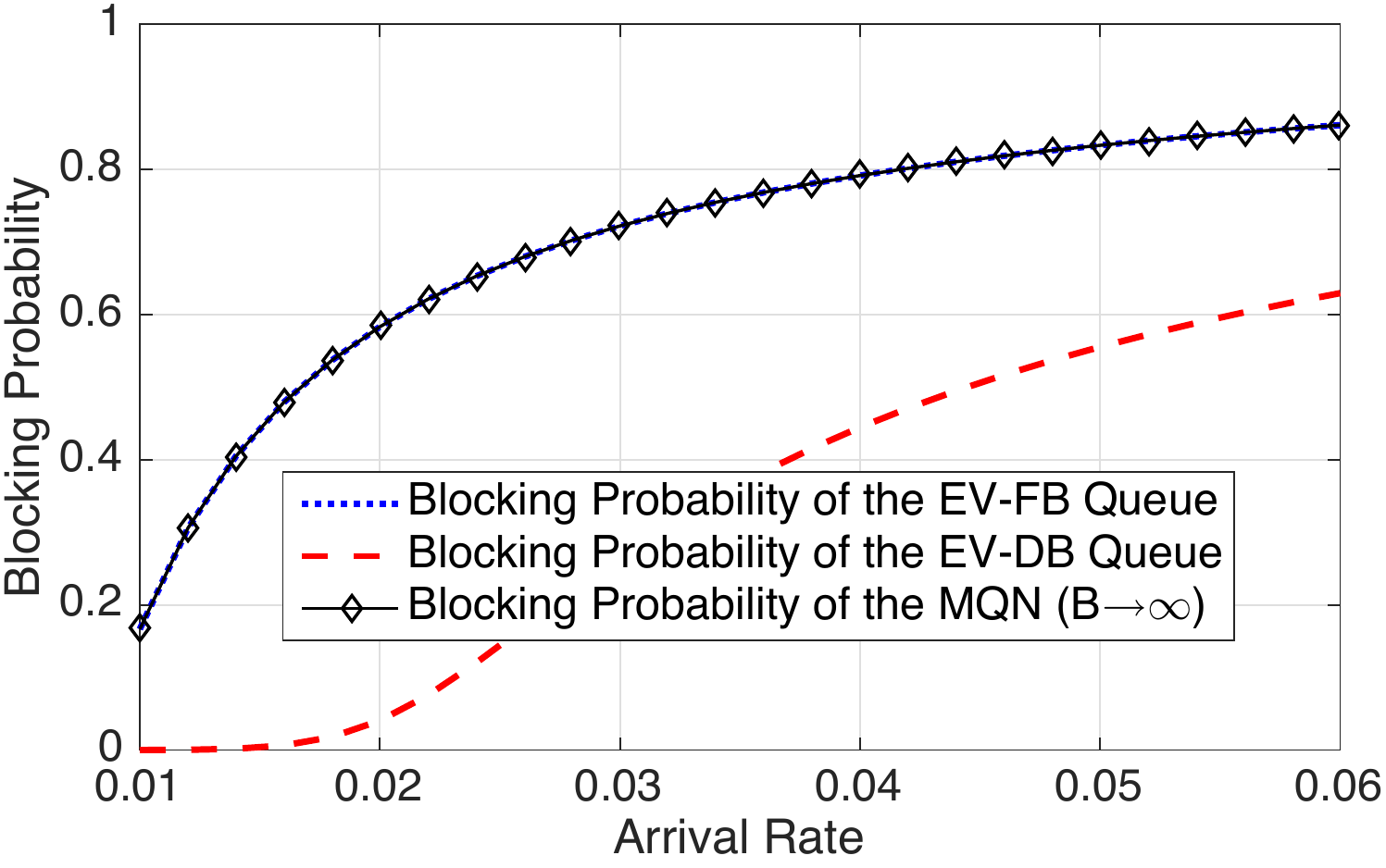}}
	\subfigure[$ C=60 $]{\includegraphics[width=10cm]{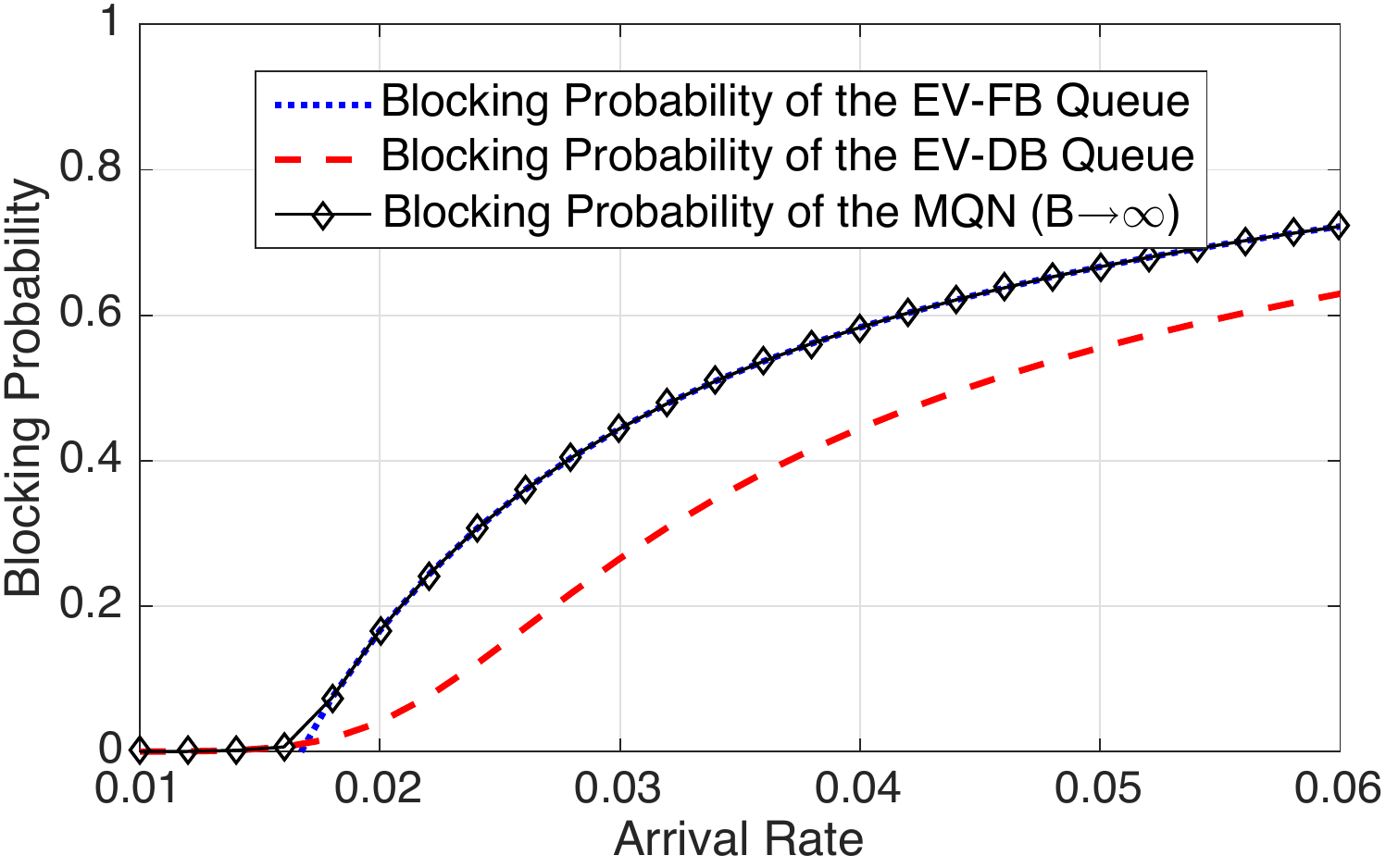}}
	\subfigure[$ C=90 $]{
		\includegraphics[width=10cm]{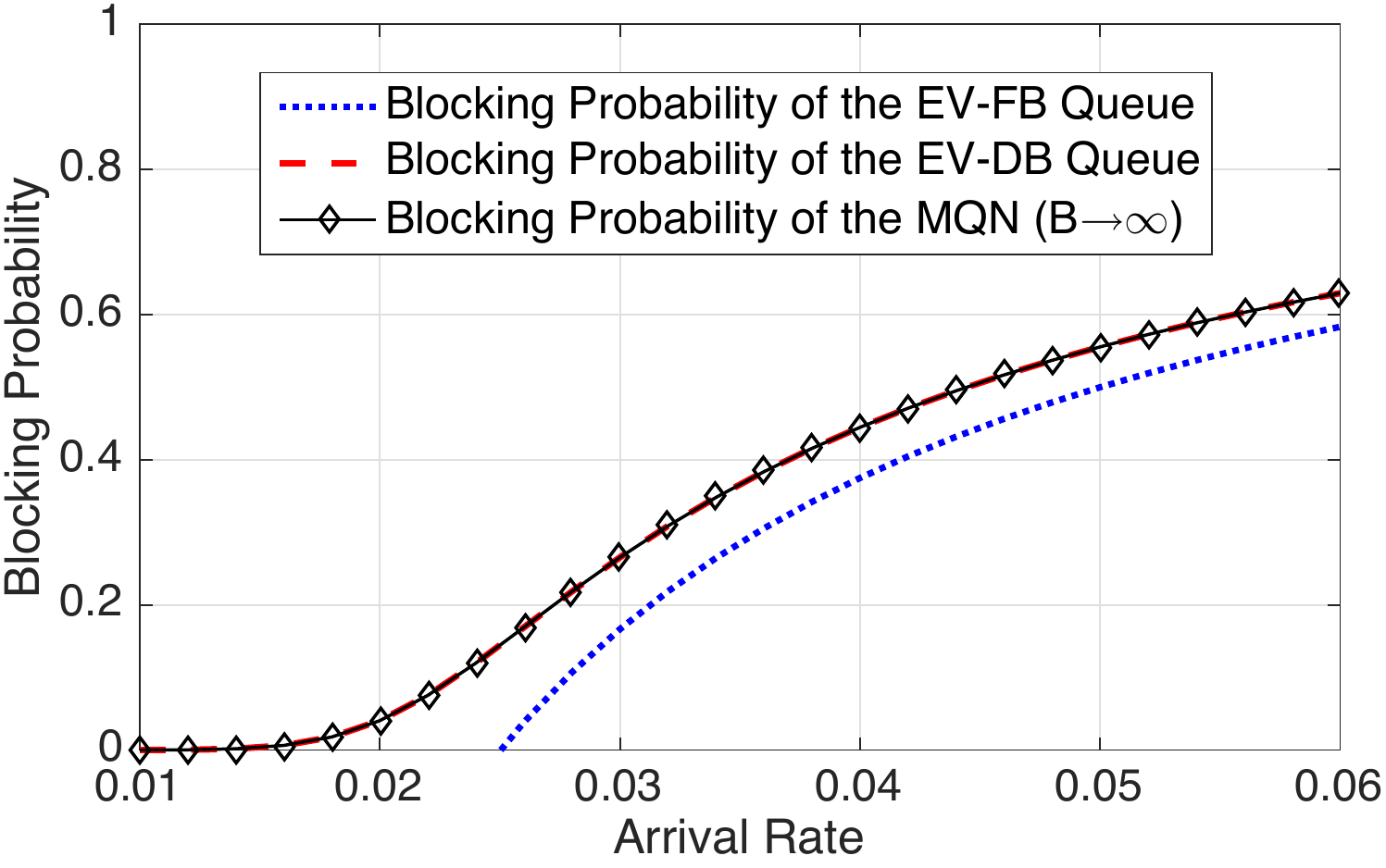}}
	\caption{Blocking probabilities with different traffic loads for different $ C $ values. The common parameters for the three figures are as follows:  $ N = 12, S =2, \nu = 1/90 $, and $ \mu = 1/3600 $.}
	\label{PoB_arrival_rate}
\end{figure}

\subsection{Blocking Probabilities with Different Traffic Loads}
Figure \ref{PoB_arrival_rate} plots  the blocking probabilities of the MQN (when $ B $ approaches infinity) with different arrival rates. The three subfigures in Fig. \ref{PoB_arrival_rate} have the same $ N, S, \nu $ and $ \mu $, and thus the $(N,S)$-limiting lower bound (i.e., the blocking probability of the EV-DB queue illustrated by the dashed curve) in these three subfigures are the same.  As shown in Fig. \ref{PoB_arrival_rate}(a), when $ C = 30 $ and $ \lambda\in[0.01,0.06] $, the blocking probability of the MQN exactly overlaps with the blocking probability of the EV-FB queue. As a result, increasing $ N $ and $ S $ cannot reduce the blocking probability of the MQN, since the blocking probability of the EV-FB queue is independent of both $ N $ and $ S $. Therefore, if the MQN is working in such a case, the operator should not invest resources for building more parking lots and SSs. Instead, effort should be spent in increasing $ C $  since it can drag the $C$-limiting lower bound closer to the $(N,S)$-limiting lower bound, as shown in the comparison between Fig. \ref{PoB_arrival_rate}(a) and Fig. \ref{PoB_arrival_rate}(b). However, after the blocking probability of the MQN completely overlaps with the $(N,S)$-limiting lower bound, as shown in Fig. \ref{PoB_arrival_rate}(c), further increasing $ C $ cannot reduce the blocking probability of the MQN anymore. In fact, Fig. \ref{PoB_arrival_rate}(c) shows that the MQN is always working in the swapping-limiting mode for traffic load $ \lambda \in[0.01,0.06] $, and the only way to reduce the blocking probability is to increase $ N $ and $ S $, instead of $ C $.  In summary, when the MQN is working in the charging-limiting mode under a certain traffic load, the blocking probability of the MQN can be reduced by increasing $ C $. However, when the MQN is working in the swapping-limiting mode as shown in Fig. \ref{PoB_arrival_rate}(a), the blocking probability of the MQN can only be reduced by increasing $ N $  and $ S $.

Note that  as illustrated in Fig. \ref{PoB_arrival_rate}(b),  the MQN switch from the swapping-limiting mode when $ \lambda\in [0.01,0.06] $ to the charging-limiting mode when $ \lambda = [0.016,0.06] $.  This switching phenomenon shows that the asymptotic performance of the BSCS can be considerably changed even the traffic load only has a minor change around a particular threshold (e.g., $ \lambda = 0.016 $ in Fig. \ref{PoB_arrival_rate}(b)). 

\begin{figure}
	\centering
	\subfigure[$ C=10 $]{\includegraphics[width=10cm]{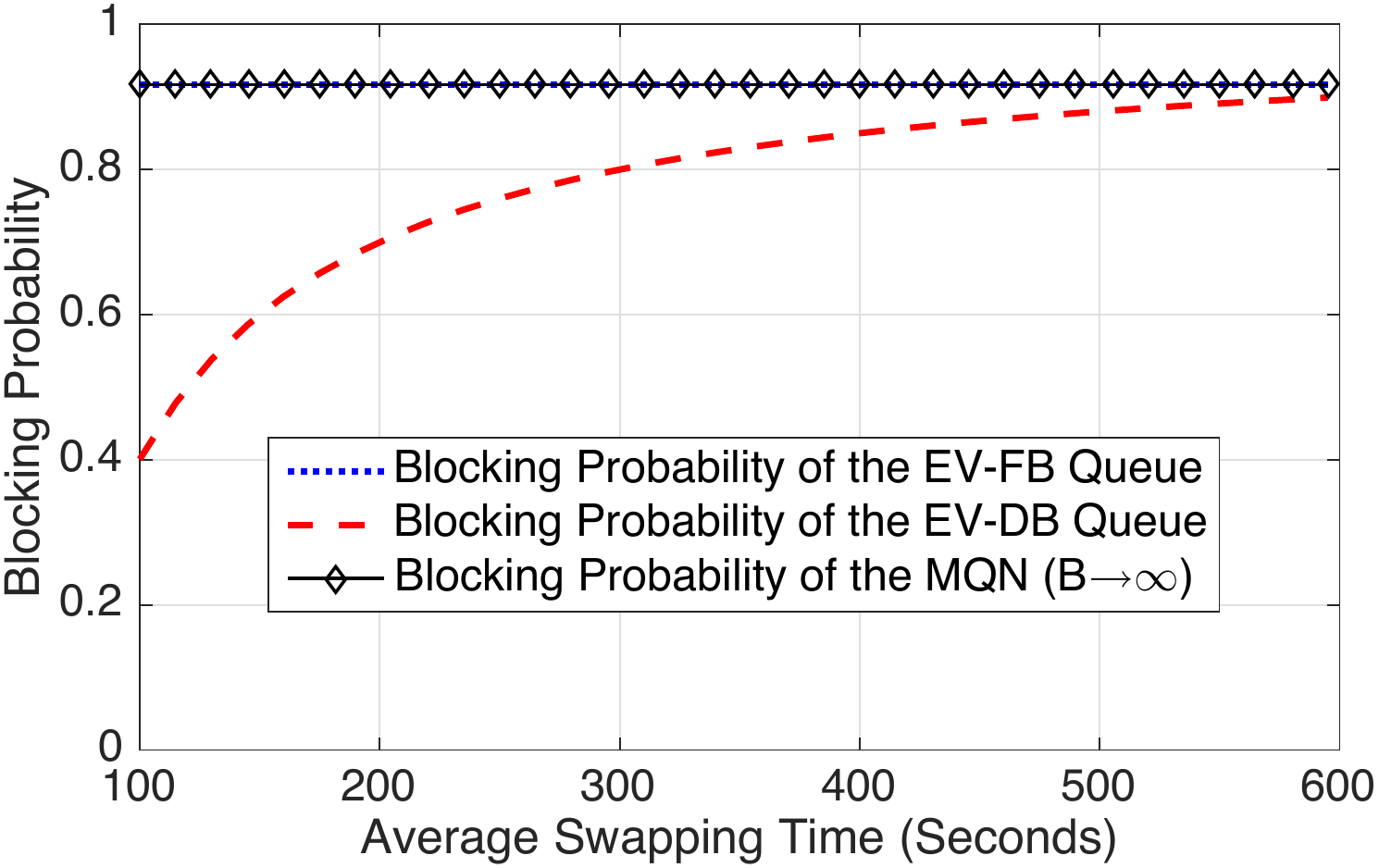}}
	\subfigure[$ C=60 $]{\includegraphics[width=10cm]{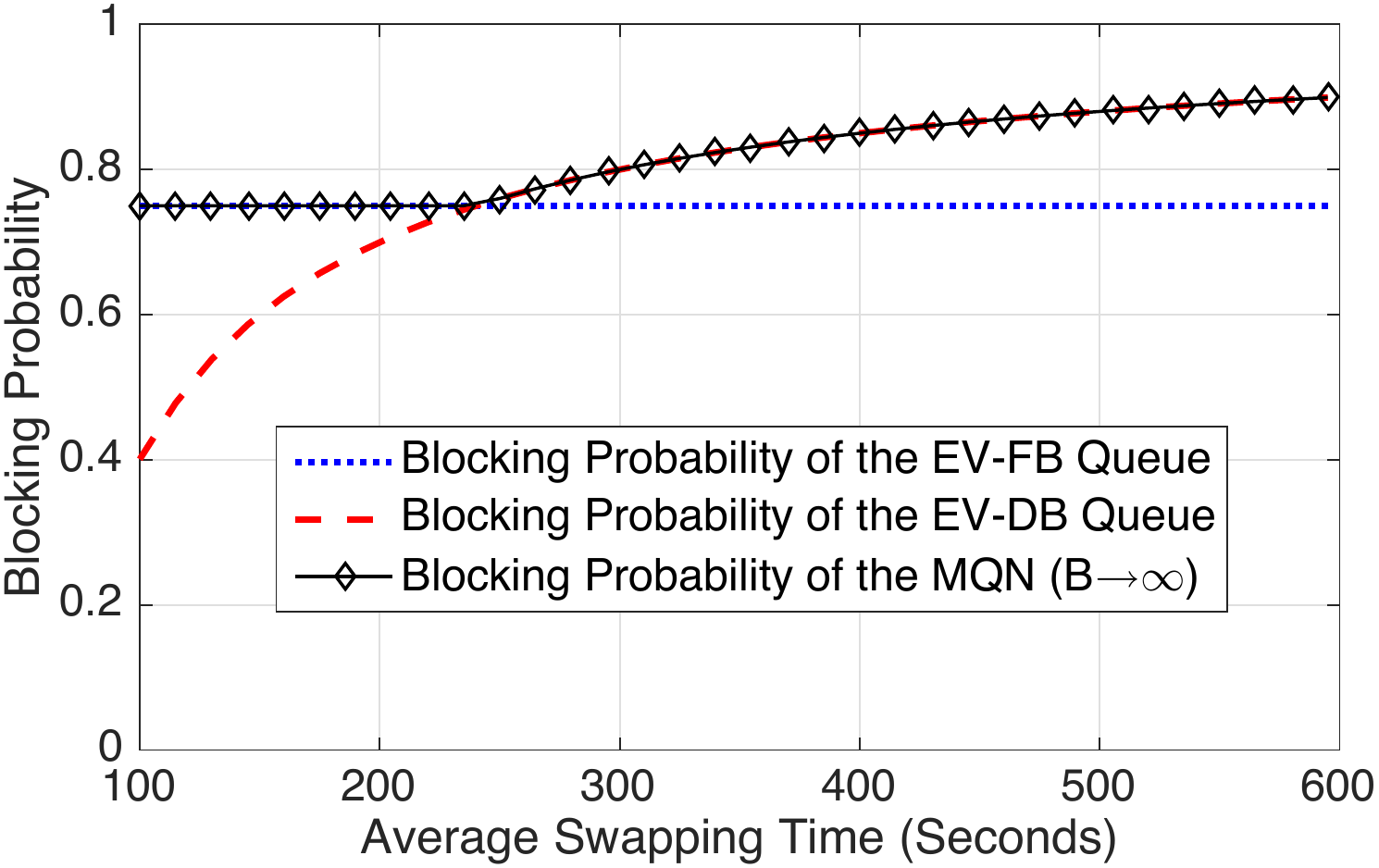}}
	\subfigure[$ C=90 $ ]{\includegraphics[width=10cm]{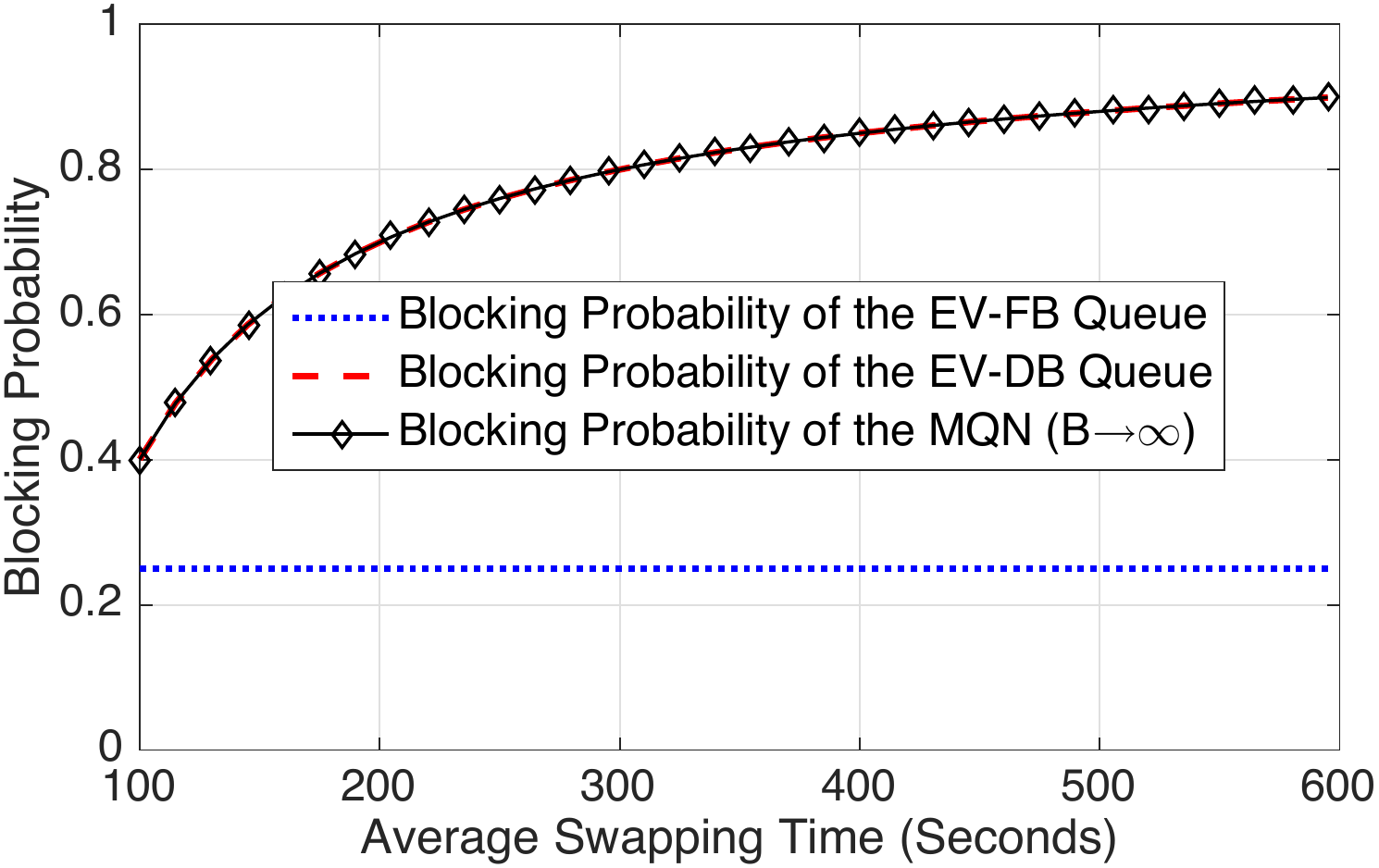}}
	\caption{Blocking probabilities with different swapping rates. The common parameters for the three figures are as follows:  $ N = 12, S =2, \lambda = 1/30 $, and $ \mu = 1/3600 $.}
	\label{PoB_swapping_rate}
\end{figure}

\begin{figure}
	\centering
	\subfigure[$ C=20 $]{\includegraphics[width=10cm]{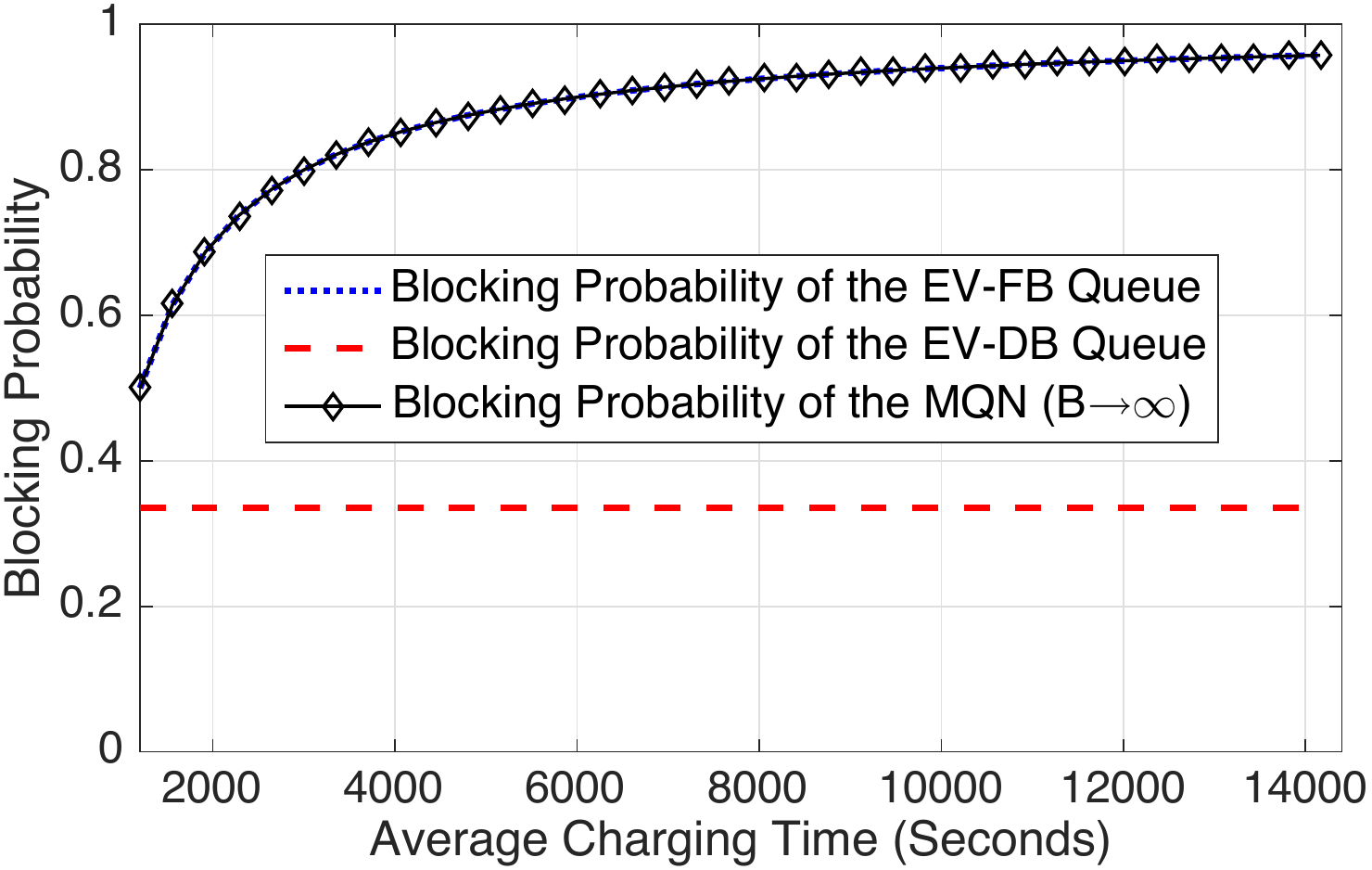}}
	\subfigure[$ C=120 $]{\includegraphics[width=10cm]{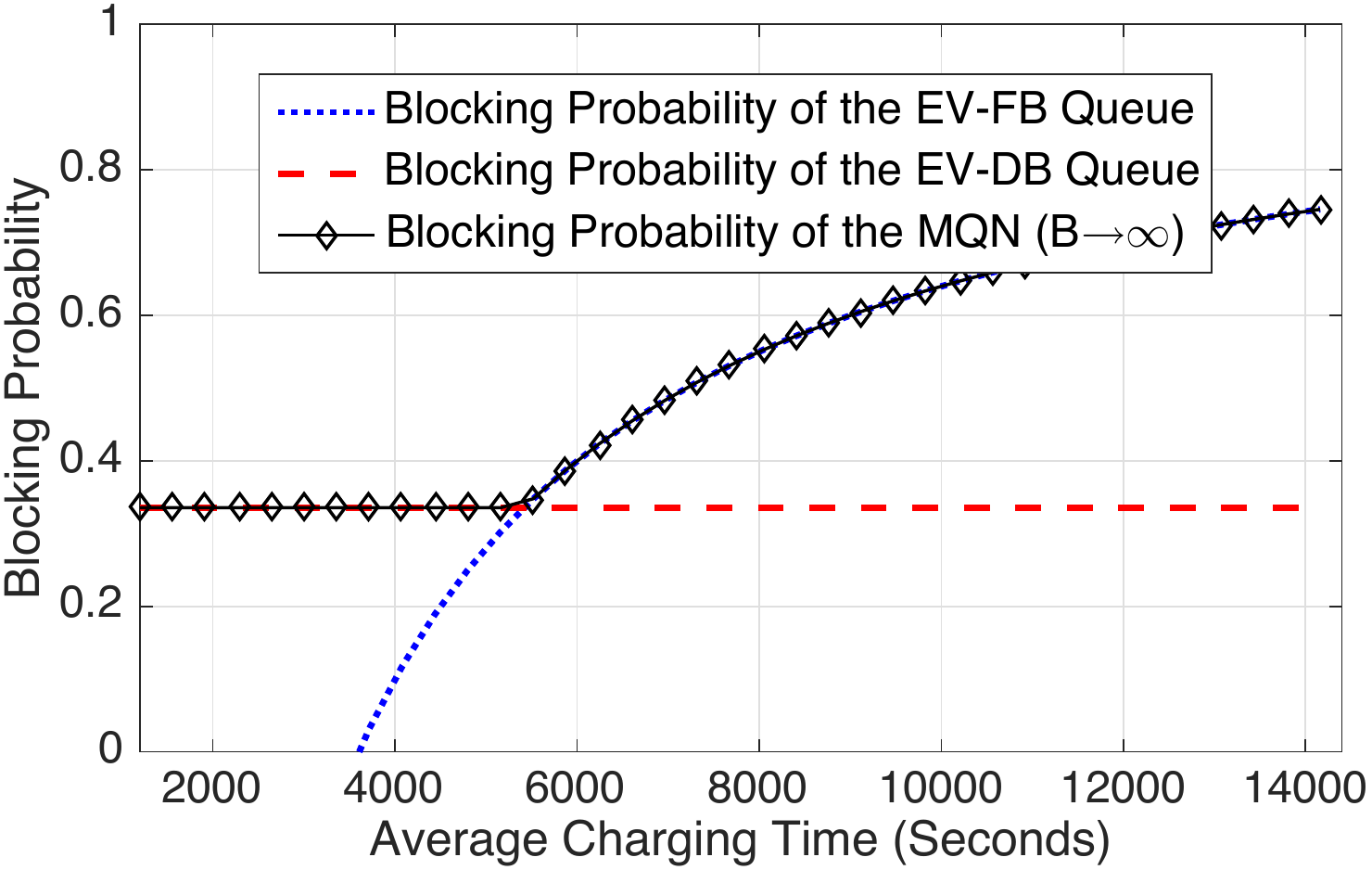}}
	\subfigure[$ C=330 $ ]{\includegraphics[width=10cm]{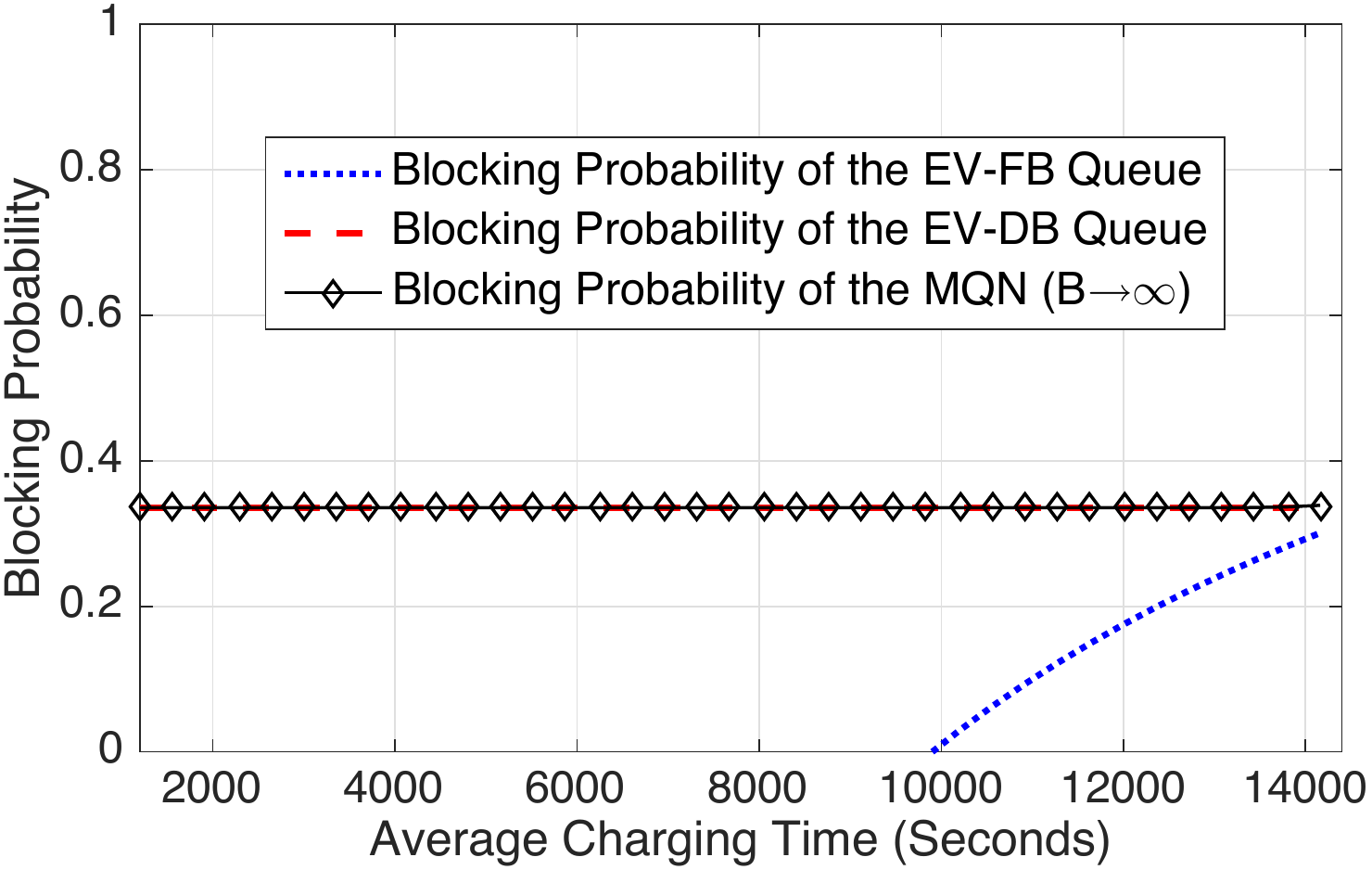}}
	\caption{Blocking probabilities with different charging rates. The common parameters for the three figures are chosen as follows:  $ N = 12, S =2, \lambda = 1/30 $, and $ \nu = 1/90 $.}
	\label{PoB_charging_rate}
\end{figure}

\subsection{Blocking Probabilities with Different Swapping Rates}
Figure \ref{PoB_swapping_rate} shows the blocking probabilities of the MQN (when $ B $ approaches infinity) with different swapping rates. The three subfigures in Fig. \ref{PoB_swapping_rate} have the same $ N, S, \lambda$ and $ \nu $, thus the $(N,S)$-limiting lower bound is the same among these three subfigures. As we can see from Fig. \ref{PoB_swapping_rate}(a), when the MQN is working in the charging-limiting mode, the blocking probability of the MQN remains as a constant when $ \nu $ changes. Therefore, it is impossible to reduce the blocking probability of the MQN by unilaterally increasing the swapping rate (or equivalently, decreasing the average swapping time) in the charging-limiting mode. In comparison, the blocking probability of the MQN is sensitive to the swapping rate when the MQN is working in the swapping-limiting mode, as shown in Fig. \ref{PoB_swapping_rate}(c). In particular, the comparison between Fig. \ref{PoB_swapping_rate}(a) and Fig. \ref{PoB_swapping_rate}(c) demonstrates that increasing $ C $ can drag the blocking probability of the MQN closer to the $(N,S)$-limiting lower bound, and the impact is more significant when the swapping rate is high. However, an exceptional case is illustrated by Fig. \ref{PoB_swapping_rate}(b), where the blocking probability of the MQN remains as a constant when the swapping-rate is high. Therefore, it is important to check the current operating modes before making any further investment in increasing the swapping rate. For instance, Fig. \ref{PoB_swapping_rate}(b) depicts that further decreasing  the average swapping time when it is already less than 240 seconds has no positive effect on reducing the blocking probability of the MQN.

\subsection{Blocking Probabilities with Different Charging Rates}
Figure \ref{PoB_charging_rate} shows the blocking probabilities of the MQN (when $ B $ approaches infinity) with different charging rates.  The three subfigures in Fig. \ref{PoB_charging_rate} have the same $ N, S, \nu $ and $ \mu $, thus the $(N,S)$-limiting lower bound (i.e., the blocking probability of the EV-DB queue illustrated by the dashed curve) in these three subfigures are the same constant.  When the MQN is working in the charging-limiting mode, the blocking probability of the MQN is strictly increasing in $ 1/\mu$, meaning that a longer charging time will increase the achievable lower bound. Moreover, this blocking probability can be greatly reduced if  $ C $ becomes larger, as shown by the comparison between Fig. \ref{PoB_charging_rate}(a) and Fig. \ref{PoB_charging_rate}(b). However, further increasing $ C $ cannot arbitrarily reduce the blocking probability of the MQN, as shown by the comparison between Fig. \ref{PoB_charging_rate}(b) and Fig. \ref{PoB_charging_rate}(c). In fact, Fig. \ref{PoB_charging_rate}(c) shows that the blocking probability of the MQN will not be influenced by the charging rate as long as the MQN is working in the swapping-limiting mode. Therefore, there is no need to increase the charging speed once the MQN in the swapping-limiting mode.

\subsection{Justification of Theoretic Results}
Our theoretic analysis is mainly based on the assumption that both the charging time and the swapping time are exponentially distributed. In practice, the swapping and charging distributions have finite supports and thus deviate from the exponential distribution.  Therefore, it is important to quantify the gap between our analytical results and the numerical results based on practical distributions with finite supports. To this end, we evaluate the blocking probability of the proposed MQN via Monte Carlo (MC) simulation.  The distributions for the swapping and charging time used in our MC simulation are shown in Fig. \ref{distributions}. In particular, the two examples of swapping distribution (SD) in Fig. \ref{distributions}(a) both follow the truncated normal distribution within range $ [60,120] $, and the mean is chosen to be 90 seconds, which is the same as that of the exponential distribution (i.e., $ 1/\nu = 90 $)\footnote{SD-I and SD-II can be considered the distributions for autonomous swapping and manual swapping, respectively. We keep the mean of SD-I and SD-II the same as the exponential distribution in order to have a fair comparison.}. Similarly, the two truncated normal distributions in Fig. \ref{distributions}(b) represent two types of charging distribution (CD) with the same finite support (i.e.,  $ [3000, 4200] $) but different variance. Meanwhile, both CD-I and CD-II have the same mean as the corresponding exponential swapping distribution (i.e., $ 1/\mu = 3600 $). Without loss of generality, we set all the other parameters of the BSCS as follows: $ N=12, S=2, C=120, \lambda = 1/30$. 

\begin{figure}
	\centering
	\subfigure[Two examples of swapping distributions]{\includegraphics[width=8cm]{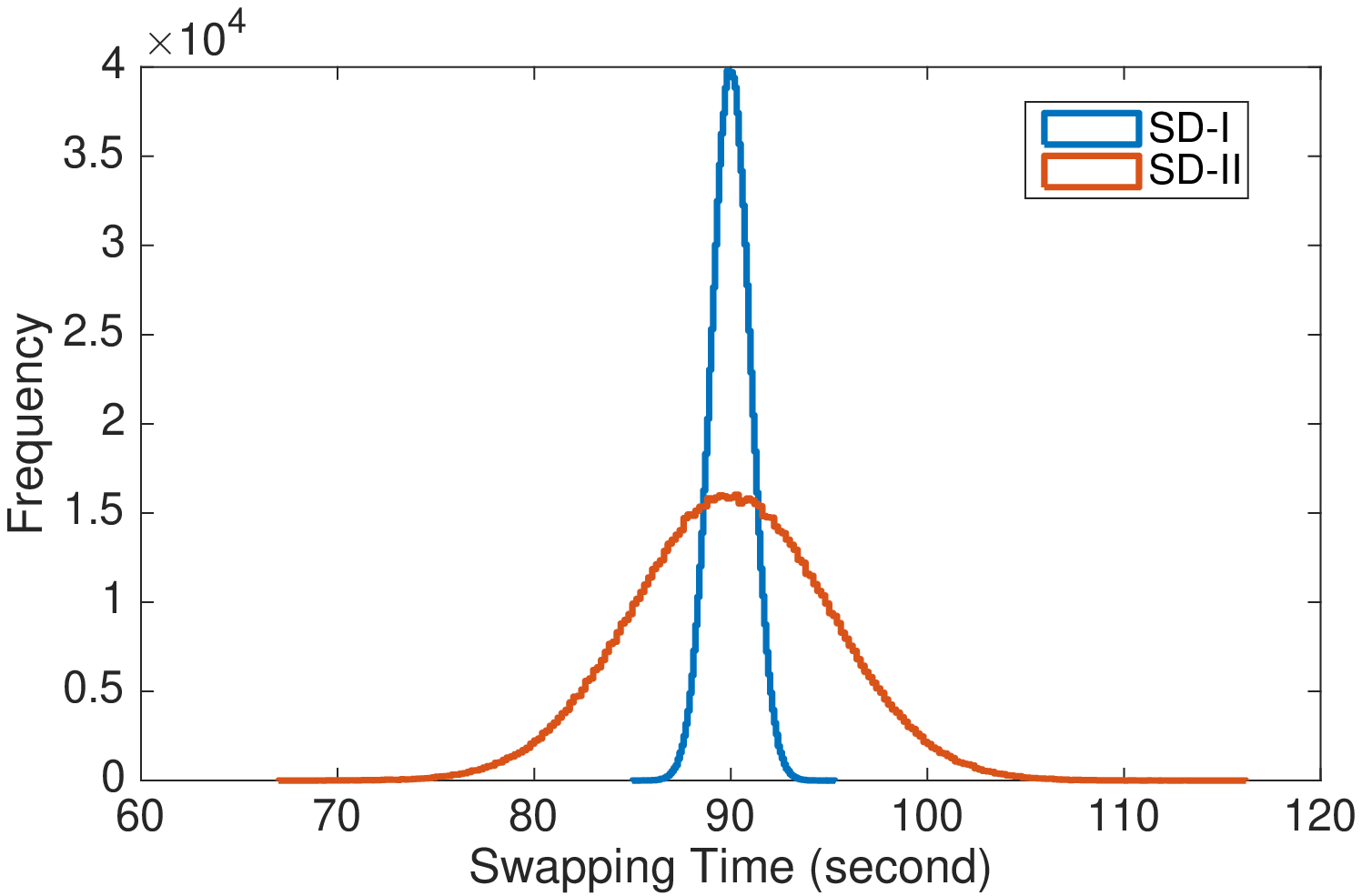}}
	\subfigure[Two examples of charging distributions]{\includegraphics[width=8cm]{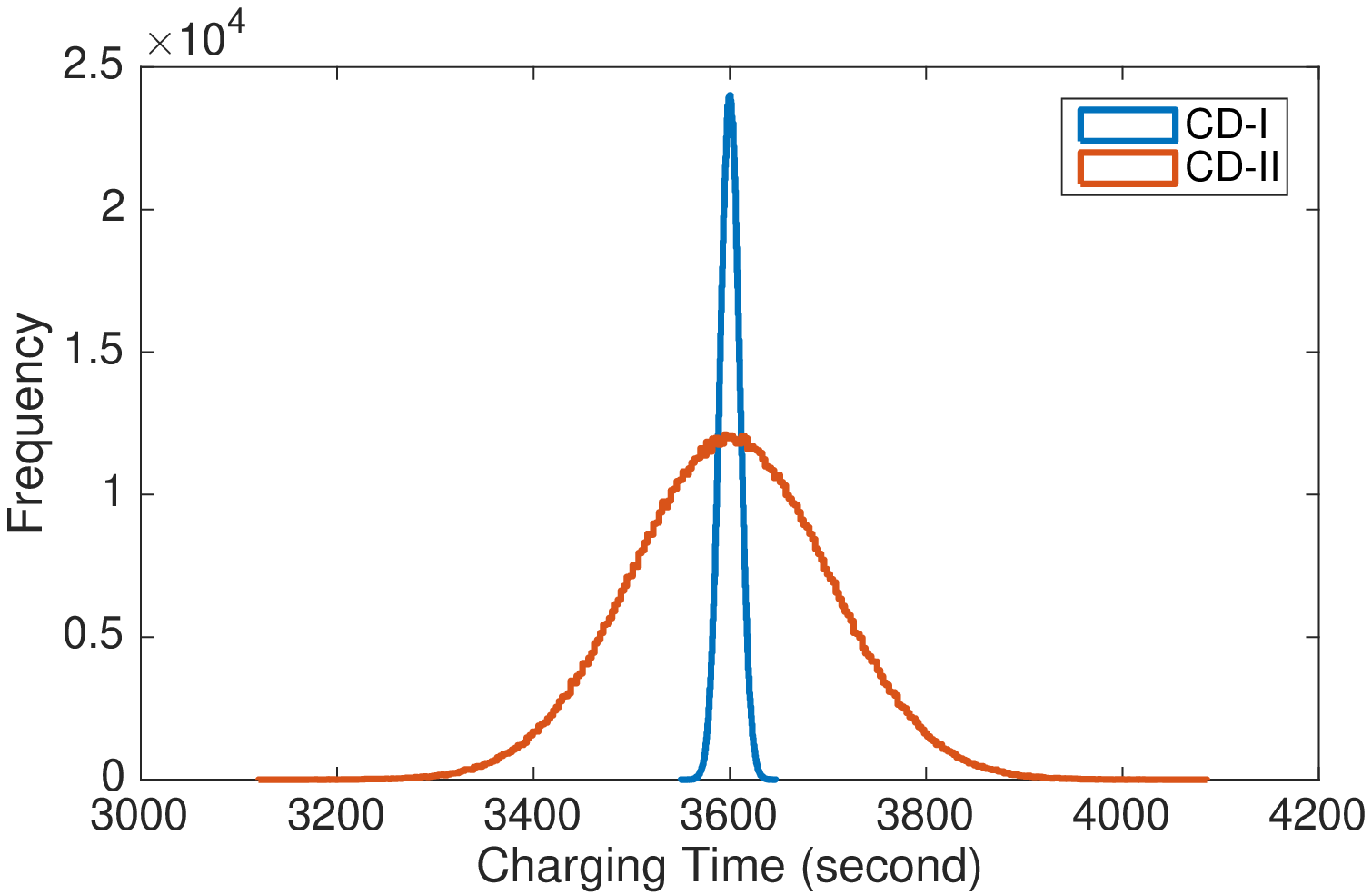}}
	\caption{Examples for the swapping and charging distributions. }
	\label{distributions}
\end{figure}

\begin{table}
	\caption[]{Comparison between theoretic analysis and MC simulations}.
	\centering
	\begin{tabular}{ c c  c c c c}
		\hline\hline
		$ B $ & \textbf{AR}  & \textbf{SR-I (Gap)} & \textbf{SR-II (Gap)} & \textbf{SR-III (Gap)}  & \textbf{SR-IV (Gap)} \\
		\hline
		
		10 &	0.9187  & 0.9186 (0.01\%) & 0.9185 (0.02\%) & 0.9185 (0.02\%) & 0.9187 (0.00\%) \\
		
		30 & 0.7569  & 0.7563 (0.08\%) & 0.7566 (0.04\%) & 0.7560 (0.12\%) & 0.7562 (0.09\%) \\
		
		50 & 0.5983  & 0.5945 (0.64\%) & 0.5950 (0.55\%) & 0.5947 (0.61\%) & 0.5963 (0.34\%) \\
		
		70 & 0.4518  & 0.4352 (3.81\%) & 0.4413  (2.38\%) & 0.4350 (3.86\%) & 0.4421 (2.19\%)\\
		
		90 & 0.3543  & 0.3326 (6.52\%) & 0.3356 (5.57\%) & 0.3335 (6.24\%) & 0.3358 (5.51\%)\\
		
		110 & 0.3383  & 0.3329 (1.62\%) & 0.3343 (1.20\%) & 0.3331 (1.56\%) & 0.3329 (1.62\%)\\
		
		130 & \underline{0.3382}  & 0.3331 (1.53\%) & 0.3339 (1.29\%) & 0.3331 (1.53\%) & 0.3324 (1.74\%)\\
		\hline
	\end{tabular}
	\label{table_slm}
\end{table}

We denote the analytical blocking probability of the MQN as AR (i.e., analytical result),  the simulation result (SR) based on SD-I and CD-I (CD-II) as SR-I (SR-II), and the SR based on SD-II and CD-I (CD-II) as SR-III (SR-IV). Table \ref{table_slm} shows the comparison between AR, SR-I, SR-II, SR-III, and SR-IV. Note that all the MC simulation results in Table \ref{table_slm} are average of 100 times MC simulations. For each simulation, we set the length of total time horizon as 30 days, which is long enough to reach the steady-state of the MQN. We can see from Table \ref{table_slm} that for most of the cases, the ARs and the SRs are very close to each other (below 2\%), and even the maximum gap between ARs and SRs is still less than 7\% (when $ B= 90 $). Meanwhile, it is worth pointing out that all the lower bounds of the blocking probability are very close to the theoretic $ (N,S)$-limiting lower bound (i.e., the underlined 0.3382 shown in Table \ref{table_slm}). The comparison shown in Table \ref{table_slm} demonstrates that the blocking probability will not be greatly affected by the specific swapping and charging distribution as long as the mean is kept the same. Therefore, the theoretic results can serve as a very accurate approximation for the BSCS with practical swapping and charging distributions. 

\section{Conclusions}
\label{conclusion}
In this paper, we adopted a queueing theoretical approach to evaluate the asymptotic performance of a BSCS. The BSCS was modeled as a novel MQN comprised of an open EV-queue and a closed battery-queue. We derived the balance equations for the queueing system and obtained the corresponding steady-state distribution. We further proposed the concept of multi-stage capacity planning, by which the four critical parameters of the BSCS will be considered in three different planning stages. In each stage, we quantified the relationship between the blocking probability and the corresponding planning parameters. As the main result, we showed that the number of CSs will distinguish two difference operating modes for the BSCS, namely, the charging-limiting mode and the swapping-limiting mode, and each limiting mode has a completely different asymptotic convergence property with respect to the number of batteries. Moreover, we proved the asymptotic ergodicity of the system in both the charging-limiting mode and the swapping-limiting mode. Extensive numerical results validated the proposed queueing model and showed practical insights for the planning and operations of BSCSs.

\newpage
\appendix
\section{Proof of Lemma \ref{irreducibility}}
\label{appendix_proof_irreducibility}
The irreducibility of $ \mathbf{Q}^{(\mathrm{EVFB})} $ and $ \mathbf{Q}^{(\mathrm{EVDB})} $ can be seen from the CTMC illustrated in Fig. \ref{CTMC_equivalent_model_I}. Here, we briefly sketch the proof for $ \mathbf{Q}^{(\mathrm{EVFB})} $ and skip that for $ \mathbf{Q}^{(\mathrm{EVDB})} $ due to similarity. To prove $ \mathbf{Q}^{(\mathrm{EVFB})} $ is irreducible, it is equivalent to show that every two states can be mutually accessible from each other. In fact, let us randomly pick two states $ (n,b) $ and $ (n',b') $, where $ n,n'\in\mathcal{N} $, and $ b,b'\in\{0,1,\cdots, B\} $. We consider a case when  $ n'\geq n $ and $ b'\geq b $. In order to prove $ (n,b)\leftrightarrow (n',b')$, i.e., $ (n,b)  $ and $ (n',b') $ are mutually accessible from each other, it suffices to prove $
(n,b)\leftrightarrow (n',b) \leftrightarrow (n',b') $.
First, since $ n'\geq n $, it is trivial to see that there is direct transition between state $ (n,b) $ and state $ (n',b) $ with a non-zero probability (i.e., the transitions moving downward in each column of Fig. \ref{CTMC_equivalent_model_I}(a)). Second, since $ b'\geq b $, we can always have $(n',b)\leftrightarrow (n',b+1)\leftrightarrow\cdots\leftrightarrow (n',b')$ (i.e., the transitions moving rightward in each row of Fig. \ref{CTMC_equivalent_model_I}(a)). Therefore, we prove the communicability between these two states. Likewise for other three cases based on combinations of $ n'\geq n, n'<n, b'\geq b $ and $ b'<b $. 

\section{Proof of Lemma \ref{positive_recurrent}}
\label{appendix_proof_positive_recurrent}
We only prove the first half of this lemma since the proof of the second half is similar. Define matrix $\mathbf{A} \triangleq \mathbf{F} + \mathbf{L} + \mathbf{D}$, where $ \mathbf{F}, \mathbf{L} $ and $ \mathbf{D} $ are given by \eqref{matrices_F_L_D}. Thus, $ \mathbf{A} $ is given as
\begin{align}
\mathbf{A} = \left[
\begin{array}{ccccc}
-\lambda & \lambda \\
\nu & -(\lambda + \nu) & \lambda\\
& 2\nu & -(\lambda + 2\nu) & \lambda\\
& & \ddots & \ddots & \ddots\\
& & & S \nu & -S \nu\\
\end{array}
\right].
\end{align}
Observe that $ \mathbf{A} $ is the transition rate matrix of an $M/M/S/N$ queue whose arrival rate and exponential service rate are $\lambda$ and $\nu$, respectively. Therefore, the Markov chain $ \mathbf{A} $ is definitely irreducible and positive recurrent, and it has a unique stationary distribution. Recall that the blocking probability of the $ M/M/S/N $ queue corresponds to the Markov chain $ \mathbf{A} $, which has been given by $\mathbb{P}_{\mathrm{nslb}}(N,S)  $  in \eqref{P_best_N_S}.

We denote the unique stationary distribution of the Markov chain $ \mathbf{A} $ as $\mathbf{p} =(p_0,p_1,\cdots,p_N) $. When being stationary, the average departure rate of the $M/M/S/N$ queue equals the rate at which customers (EVs in this paper) arrive and enter the system. Therefore, the following equality holds
\begin{align}
\sum_{n=0}^{N}p_n\min\{n,S\}\nu
= \ \lambda \Big(1 - \mathbb{P}_{\mathrm{nslb}}(N,S)\Big).
\label{definition_of_P_best}
\end{align}
The left-hand-side of \eqref{definition_of_P_best} can be equivalently transformed into a more compact form as
\begin{align}
\sum_{n=0}^{N}p_n\min\{n,S\}\nu = \mathbf{p}\mathbf{D}\mathbf{e}.
\end{align}
Therefore, $ C\leq \lfloor \lambda\big(1-\mathbb{P}_{\mathrm{nslb}}(N,S)\big)/\mu\rfloor $ is equivalent to the following inequality:
\begin{align}
C\mu = \mathbf{p}\mathbf{F}\mathbf{e} < \lambda\Big(1 - \mathbb{P}_{\mathrm{nslb}}(N,S)\Big) = \mathbf{p}\mathbf{D}\mathbf{e}.
\end{align}

Theorem 3.1.1 in \cite{matrix_geometric_method} has proved that an irreducible Markov chain in the form of $ \mathbf{Q}^{(\mathrm{EVFB})} $ defined in \eqref{Q_EVFB} is positive recurrent if and only if $ \mathbf{p}\mathbf{F}\mathbf{e} < \mathbf{p}\mathbf{D}\mathbf{e} $. We thus complete the proof.

\section{Proof of Lemma \ref{probability_of_enough_and_busy}}
\label{appendix_proof_probability_of_having_FB_shortage}
\textbf{First, we show how to calculate the steady-state distributions for the two sub-queueing networks}. 
The CTMCs defined by $ \mathbf{Q}^{(\mathrm{EVFB})} $ and $ \mathbf{Q}^{(\mathrm{EVDB})} $ are two quasi-birth-death (QBD) processes, whose stationary distributions can be computed by using the matrix geometric method \cite{matrix_geometric_method}. Since the irreducible Markov chain $ \mathbf{Q}^{(\mathrm{EVFB})} $ is positive recurrent when the BSCS is working in the charging-limiting mode, based on Theorem 1.2.1 in \cite{matrix_geometric_method}, the minimal nonnegative solution $\mathbf{R}$ to the matrix-quadratic equation $\mathbf{R}^2 \mathbf{D} + \mathbf{R} \mathbf{L} + \mathbf{F} = 0 $ has all its eigenvalues being less than 1, and the finite system of equations
\begin{align}
[\boldsymbol{\pi}_{0:S-1}^{(\mathrm{EVFB})}, \boldsymbol{\pi}_{S}^{(\mathrm{EVFB})}] \left[
\begin{array}{ccc}
\mathbf{e} & \tilde{\mathbf{L}}_{\mathrm{00}} &\mathbf{F}_{\mathrm{01}}\\
(\mathbf{I}-\mathbf{R})^{-1} \mathbf{e} & \tilde{\mathbf{D}}_{\mathrm{10}} & \mathbf{L} + \mathbf{R} \mathbf{D}
\end{array}\right] = [1,\mathbf{0}]
\label{clm_finite_equations}
\end{align} 
has a unique positive solution $\boldsymbol{\pi}_{0:S-1}^{(\mathrm{EVFB})}$ and $\boldsymbol{\pi}_S^{(\mathrm{EVFB})}$, where $ \boldsymbol{\pi}_{0:S-1}^{(\mathrm{EVFB})} \triangleq [ \boldsymbol{\pi}_0^{(\mathrm{EVFB})}$,  $\cdots,\boldsymbol{\pi}_{S-1}^{(\mathrm{EVFB})}] $. Moreover, we have $ \boldsymbol{\pi}_b^{(\mathrm{EVFB})} = \boldsymbol{\pi}_{S}^{(\mathrm{EVFB})} \mathbf{R}^{b-S}, \forall b=\{S+1, \cdots,\infty\} $. Note that $ \tilde{\mathbf{L}}_{\mathrm{00}} $ and $ \tilde{\mathbf{D}}_{\mathrm{10}} $ in \eqref{clm_finite_equations} are respectively $ \mathbf{L}_{\mathrm{00}} $ and $ \mathbf{D}_{\mathrm{10}} $ with the first column being eliminated.

Similarly, in the swapping-limiting mode, the unique stationary distribution of $ \mathbf{Q}^{(\mathrm{EVDB})} $ can be computed by solving the following finite system of equations
\begin{align}
[\boldsymbol{\pi}^{(\mathrm{EVDB})}_{0:C-1}, \boldsymbol{\pi}^{(\mathrm{EVDB})}_{C}] \left[
\begin{array}{ccc}
\mathbf{e} & \tilde{\mathbf{L}}^{\mathrm{N}}_{\mathrm{00}} &\mathbf{D}^{\mathrm{N}}_{\mathrm{01}}\\
(\mathbf{I}-\mathbf{M})^{-1} \mathbf{e} & \tilde{\mathbf{F}}^{\mathrm{N}}_{\mathrm{10}} & \mathbf{L} + \mathbf{M} \mathbf{F}
\end{array}\right] = [1,\mathbf{0}],
\label{slm_finite_equations}
\end{align} 
where $ \boldsymbol{\pi}_{0:C-1}^{(\mathrm{EVDB})}=[ \boldsymbol{\pi}_0^{(\mathrm{EVDB})},\cdots,\boldsymbol{\pi}_{C-1}^{(\mathrm{EVDB})}] $, and $ \boldsymbol{\pi}^{(\mathrm{EVDB})}_{C} $ denote the stationary distribution. $\mathbf{M}$ is the minimal nonnegative solution to the matrix-quadratic equation
$ \mathbf{M}^2 \mathbf{F} + \mathbf{M} \mathbf{L} + \mathbf{D} = 0 $, whose eigenvalues are all less than 1. Moreover, we  have $ \boldsymbol{\pi}^{(\mathrm{EVDB})}_{j} = \boldsymbol{\pi}^{(\mathrm{EVDB})}_{C} \mathbf{M}^{(j-C)},\forall j=\{ C+1, \cdots,\infty\} $.  Note that $ \tilde{\mathbf{L}}^{\mathrm{N}}_{\mathrm{00}} $ and $ \tilde{\mathbf{F}}^{\mathrm{N}}_{\mathrm{10}} $ in \eqref{slm_finite_equations} are respectively $ \mathbf{L}^{\mathrm{N}}_{\mathrm{00}} $ and $ \mathbf{F}^{\mathrm{N}}_{\mathrm{10}} $ with the first column being eliminated.

\textbf{Second, we show the proof of \eqref{probability_of_enough_and_busy_clm} in Lemma \ref{probability_of_enough_and_busy}}. Based on Lemma \ref{positive_recurrent}, we know that if the charging-limiting mode is active, i.e., $ C\leq \lfloor \lambda\big(1-\mathbb{P}_{\mathrm{nslb}}(N,S)\big)/\mu\rfloor $, then the EV-FB queue is stable and there exists a unique steady-state distribution for the EV-FB queue. Note that the average input rate of the FB-queue in the MQN is always less than or equal to $ C\mu $, thus when $ B $ approaches infinity, the open EV-queue and the FB-queue in the original MQN must also be stable. Therefore, there exists a unique steady-state distribution for the open EV-queue and the FB-queue of the original MQN. However, when $ B $ approaches infinity, at least one of the two queues (i.e., the DB-queue and the FB-queue) in the closed battery-queue  is  unstable\footnote{This can be proved by contradiction as follows: if both of these two queues are stable, meaning that they both have unique stationary distributions, then there must be a non-zero stationary probability that the total number of batteries within the closed battery-queue is finite, which is definitely not true as $ B $ is infinity.}. Therefore, in the charging-limiting mode, the DB-queue in the MQN must be unstable, which means that the queue length of the DB-queue does not have a stationary distribution and keeps increasing when $ B $ increases. As a result, an infinite number of  DBs will be backlogged in the DB-queue, which makes all the CSs close to busy at all the time. Since each CS's charging time follows an exponential distribution with rate $ \mu $, the total input rate of the FB-queue will thus approach $ C\mu $. Therefore, in the charging-limiting mode, the open EV-queue and the FB-queue of the MQN will converge to the EV-FB queue when $ B $ approaches infinity.

Recall that when the charging-limiting mode is active,  the EV-FB queue is ergodic and it  has a unique steady-state distribution. Moreover, this unique distribution $\{ \boldsymbol{\pi}_b^{(\mathrm{EVFB})}\}_{\forall b} $ can be obtained by solving the finite system of equations \eqref{clm_finite_equations}. Therefore, based on the definition of $ \mathbb{P}_{\mathrm{enough}} $ (i.e., Equation \eqref{definition_enough_busy}), we have
\begin{align} 
\lim\limits_{B\rightarrow\infty} \mathbb{P}_{\mathrm{enough}}= & 1-\lim\limits_{B\rightarrow\infty}\sum_{n=1}^{N}\sum_{b=0}^{\min\{n,S\}-1}\pi_{n,b} \\
= &1 - \sum_{n=1}^{N}\sum_{b=0}^{\min\{n,S\}-1}\pi^{(\mathrm{EVFB})}_{n,b}\\
=& 1- \sum_{b=0}^{S-1}\boldsymbol{\pi}_b^{(\mathrm{EVFB})}\mathbf{e}_b \triangleq \phi(N,S,C),
\end{align}
where  $ \mathbf{e}_b $ is a $ (N+1)\times1 $ column vector with entries between $ (b+2) $-th and $ (N+1) $-th being 1 and 0 otherwise. For instance, if $ b = 0 $, $ \mathbf{e}_b = (0, 1, 1, \cdots, 1)^{\intercal} $. Meanwhile, based on the definition of $ \mathbb{P}_{\mathrm{busy}} $ (i.e., Equation \eqref{definition_enough_busy}), we have
\begin{align}
\lim\limits_{B\rightarrow\infty} \mathbb{P}_{\mathrm{busy}}= \lim\limits_{B\rightarrow\infty}\sum_{n=0}^{N}\sum_{b=0}^{B-C}\pi_{n,b}
=\sum_{n=0}^{N}\sum_{b=0}^{\infty}\pi^{(\mathrm{EVFB})}_{n,b}=1.
\end{align}
We thus complete the proof of \eqref{probability_of_enough_and_busy_clm} in Lemma \ref{probability_of_enough_and_busy}.

\textbf{Finally, we show the proof of  \eqref{probability_of_enough_and_busy_slm} in Lemma \ref{probability_of_enough_and_busy}}. 
Similar to the proof of  \eqref{probability_of_enough_and_busy_clm}, we can first prove that when the swapping-limiting mode is active, the open EV-queue and the DB-queue will converge to the EV-DB queue when $ B $ approaches infinity. We skip the details for brevity and only focus on deriving the two probabilities in  \eqref{probability_of_enough_and_busy_slm}.
	
Recall that when the swapping-limiting mode is active, the EV-DB queue is ergodic and  it has a unique steady-state distribution. Moreover, this unique distribution $\{ \boldsymbol{\pi}_j^{(\mathrm{EVDB})}\}_{\forall j} $ can be obtained by solving the finite system of equations \eqref{slm_finite_equations}. Based on the definition of $ \mathbb{P}_{\mathrm{enough}} $, we have
	\begin{align}
	\lim\limits_{B\rightarrow\infty} \mathbb{P}_{\mathrm{enough}}=   & 1- \lim\limits_{B\rightarrow\infty}\sum_{n=1}^{N}\sum_{b=0}^{\min\{n,S\}}\pi_{n,b}\\
	\geq  &  1- \lim\limits_{B\rightarrow\infty} \sum_{b=0}^{S}\boldsymbol{\pi}_{B-b}^{(\mathrm{EVDB})}\mathbf{e}\\
	= & 1- \lim\limits_{B\rightarrow\infty} \sum_{b=0}^{S}\boldsymbol{\pi}_{C}^{(\mathrm{EVDB})}\mathbf{M}^{B-b-C}\mathbf{e},
	\label{EVDB_B_b_C}
	\end{align}
where the last equality comes from the fact that  $ \boldsymbol{\pi}^{(\mathrm{EVDB})}_{j} = \boldsymbol{\pi}^{(\mathrm{EVDB})}_{C} \mathbf{M}^{(j-C)},\forall j=\{ C+1, \cdots,\infty\} $. Since vector $ \mathbf{e} $ can be written as 
	\begin{align}
	\mathbf{e} = \xi_1\mathbf{x}_1+\xi_2\mathbf{x}_2+\cdots+\xi_{N+1}\mathbf{x}_{N+1},
	\label{expan_e}
	\end{align} 
where $ \mathbf{x}_i $ denote the eigenvectors of matrix $ \mathbf{M} $, and $\xi_i $ are real coefficients, $ \forall  i=\{1,\cdots,N+1\} $. Based on Equations \eqref{EVDB_B_b_C} and \eqref{expan_e}, we have		
	\begin{align}
	\nonumber
	& \lim\limits_{B\rightarrow\infty} \sum_{b=0}^{S}\boldsymbol{\pi}_{C}^{(\mathrm{EVDB})}\mathbf{M}^{B-b-C}\mathbf{e}\\\nonumber
	= & \lim\limits_{B\rightarrow\infty} \sum_{b=0}^{S}\boldsymbol{\pi}_{C}^{(\mathrm{EVDB})}\mathbf{M}^{B-b-C}(\xi_1\mathbf{x}_1+\xi_2\mathbf{x}_2+\cdots+\xi_{N+1}\mathbf{x}_{N+1})\\
	= & \lim\limits_{B\rightarrow\infty} \sum_{b=0}^{S}\boldsymbol{\pi}_{C}^{(\mathrm{EVDB})}\Big(\sigma_1^{B-b-C}\xi_1\mathbf{x}_1+\sigma^{B-b-C}_2\xi_2\mathbf{x}_2+\cdots +\sigma^{B-b-C}_{N+1}\xi_{N+1}\mathbf{x}_{N+1}\Big),
	\end{align}
where $ \sigma_i, i=\{1,\cdots,N+1\} $ denotes the eigenvalue of matrix $ \mathbf{M} $ corresponds to the eigenvector $ \mathbf{x}_i, i=\{1,\cdots,N+1\}$.
	
According to \cite{matrix_geometric_method}, if Markov chain $ \mathbf{Q}^{(\mathrm{EVDB})} $ is irreducible and positive recurrent, then $ |\sigma_i|<1, \forall i=\{1,\cdots,N+1\} $. Therefore, the ergodicity of the EV-DB queue in the swapping-limiting mode indicates that the following equality holds
	\begin{align}
	\lim\limits_{B\rightarrow\infty} \sum_{b=0}^{S}\boldsymbol{\pi}_{C}^{(\mathrm{EVDB})}\mathbf{M}^{B-b-C}\mathbf{e} = 0.
	\end{align}
	Therefore, we have
	\begin{align}
	 \lim\limits_{B\rightarrow\infty} \mathbb{P}_{\mathrm{enough}} = 1.
	\end{align} 
	
Meanwhile, based on the definition of $ \mathbb{P}_{\mathrm{busy}} $, we have
	\begin{align}
	\lim\limits_{B\rightarrow\infty} \mathbb{P}_{\mathrm{busy}} = &\lim\limits_{B\rightarrow\infty}\sum_{n=0}^{N}\sum_{b=0}^{B-C}\pi_{n,b}  \\
	= & \sum_{n=0}^{N}\sum_{j=C}^{\infty}\pi^{(\mathrm{EVDB})}_{n,b}\\
	= &1 - \sum_{n=0}^{N}\sum_{j=0}^{C-1}\pi^{(\mathrm{EVDB})}_{n,b}\\
	= & 1- \sum_{j=0}^{C-1} \boldsymbol{\pi}_{j}^{(\mathrm{EVDB})}\mathbf{e} \triangleq \psi(N,S,C).
	\end{align}
We thus complete the proof of \eqref{probability_of_enough_and_busy_slm} in Lemma \ref{probability_of_enough_and_busy}.
\end{document}